\documentclass[11 pt]{article}
\input{packages.sty}

\title{Error Correction in Lattice Quantum Electrodynamics  with Quantum Reference Frames\vspace{0.25cm}}
\author[1]{Elias Rothlin}

\author[1]{Carla Ferradini}

\author[1,2,3]{Lin-Qing Chen\vspace{0.15cm}}

\affil[1]{Institute for Theoretical Physics, ETH Zurich, Wolfgang-Pauli-Strasse 27,
8093 Zürich, Switzerland}
\affil[2]{Institute of Quantum Optics and Quantum Information Vienna (IQOQI-Vienna), Austrian Academy of Sciences, Boltzmanngasse 3, 1090 Vienna, Austria}
\affil[3]{Faculty of Physics, University of Vienna, Boltzmanngasse 5, 1090 Vienna, Austria}

\date{}

\begin{document}

\maketitle

\begin{abstract}

\let\thefootnote\relax\footnote{E-mail: \href{mailto:erothlin@student.ethz.ch}{erothlin@student.ethz.ch}, \href{mailto:cferradini@phys.ethz.ch}{cferradini@phys.ethz.ch}, \href{mailto:linqing.nehc@gmail.com}{linqing.nehc@gmail.com}.}

Is gauge symmetry merely a redundancy in our description, or does it carry a deeper information-theoretic significance? Quantum error-correcting codes (QECCs) show that redundancy can serve as a resource for protecting information against noise. In this work, we ask whether gauge theories can be understood in similar terms, and make this idea concrete in lattice quantum electrodynamics (QED), building on and extending earlier works that established a bridge between gauge systems, stabilizer codes, and quantum reference frames (QRFs). 
For Abelian gauge groups, we show that explicit recovery operations can be constructed using group-theoretical methods for error sets determined by both ideal and non-ideal QRFs. Applied to lattice QED, this yields two QECC structures: one in the pure-gauge sector and one including fermions. We construct a gauge-field QRF based on spanning trees of the lattice and a fermionic field QRF from the matter field, thereby making explicit how physical information is encoded. While  the syndromes of gauge-violating errors associated with constraint measurements are generically degenerate,
 QRFs resolve this degeneracy and single out families of correctable errors. This establishes lattice QED as a QECC beyond the stabilizer setting and shows concretely how gauge symmetry provides an  encoding structure that supports error correction.
\end{abstract}
\newpage
\tableofcontents

\section{Introduction} \label{sec:introduction}
Gauge theories provide the most successful framework for describing all fundamental interactions. Notably, they are not formulated directly in terms of  gauge-invariant quantities, but rather in terms of kinematical variables in a redundant description subject to constraints \cite{dirac1964, Hennex_Teitelboim_1992}.  Different gauge choices then yield different, but physically equivalent, descriptions of the same system.  This raises a basic conceptual question concerning the role of such redundancy:  
Is it merely a convenient way of describing physics, or  is it instead a necessary feature for consistency \cite{Rovelli:2013fga}, reflecting deeper structural principles of nature?

Gauge theories can be naturally understood in relational terms. As Henneaux and Teitelboim write in the opening paragraph of \textit{Quantization of Gauge Systems} \cite{Hennex_Teitelboim_1992}: ``A gauge theory may be thought of as one in which the dynamical variables are specified with respect to a  reference frame whose choice is arbitrary at every instant of time."
Gauge-invariant observables are those that do not depend on any particular choice of the reference frame. Inspired by Dirac quantization of gauge theories \cite{dirac1964}, the perspective-neutral (PN) approach to quantum reference frames (QRFs) \cite{Vanrietvelde2020changeof, delahamette_2021, Krumm_2021, H_hn_2022} provides a procedure for obtaining quantum states relative to a chosen QRF from a so-called perspective-neutral (PN) state. These PN states are essentially the physical states in a gauge theory and contain only relational information. They encompass all  internal perspectives, while descriptions relative to different frames are obtained by applying different reduction maps to the PN states. We note that other frameworks of QRFs \cite{Bartlett_2007, Gour_2008, Gour_2009, castroruiz2023, Miyadera_2016, Loveridge2018}
are suitable to different physical questions, and the relationships between these approaches remain an open research topic.

From an information-theoretic point of view, gauge symmetry not only restricts the admissible states and physical observables of a theory, but also governs how physical information is redundantly encoded within a larger space. This is very similar to quantum error correction (QEC).
In quantum computing, quantum error correction codes (QECCs) are essential to protect quantum information from noise by redundantly encoding logical data into a larger quantum system in a way that keeps the logical content recoverable even when environmental interactions perturb the quantum state. Different encodings amount to different arrangements of the same logical information into this larger space.
Thus, while QRFs make explicit the relational character of physical information in the presence of gauge-redundancy, QECCs turn redundant encoding structures into information robustness. 
Together they provide the natural tools for analyzing the information structure of gauge theory, suggesting that gauge freedom may be understood as a concrete mechanism for  encoding, and potentially protecting, physical information. 

The idea of using QEC descriptions in  settings well beyond quantum computing has already proven fruitful in a range of contexts: 
The renormalization-group flow can be reinterpreted as as an encoding map and preserves information across scales \cite{F1_2022, F2_2022}, QECCs as tensor-network toy models have been constructed as bulk-to-boundary encodings for holography \cite{Almheiri_2015, HaPPY_2015, PP_2017, Harlow_2017, Kibe2022}, and there have also been investigations of the error correction capacities of gauge systems by leveraging a superselection rule to protect logical data \cite{bao2023}.

The initial bridge between QECCs (particularly stabilizer codes), gauge systems and PN quantum reference frames has been made in very recent studies \cite{sem_proj,CCHM_2024}. In both works, the stabilizer group of a quantum code \cite{gottesman1997stabilizer} is interpreted as a gauge group, an algebraic symmetry whose constraint surface defines the code space, in close analogy with how gauge symmetry selects the physical Hilbert space of a constrained system.  Different types of QRFs are identified for these codes, and these are related to different sets of correctable errors.  In particular, \cite{CCHM_2024} shows that the  gauge-fixing operators associated with ideal QRFs form a correctable error set, dual to a set of Pauli errors.

In the present work, we investigate whether the information-encoding structure associated with gauge symmetry can be used to identify notions of protection and recovery against gauge-violating noise in gauge systems beyond stabilizer codes.
We address this question in lattice quantum electrodynamics (QED), whose continuum limit is the simplest Abelian gauge field theory. Working on the lattice also avoids the technical complication that, in the continuum, the group of spacetime-dependent gauge transformations is not locally compact.

Building on the insights from \cite{sem_proj,CCHM_2024} and extending them beyond stabilizer codes, we show that for Abelian gauge groups, explicit information recovery channels can be constructed by group-theoretical methods for both ideal and non-ideal QRFs. We then construct two types of QRFs for lattice QED: an ideal QRF from the gauge field and a non-ideal QRF from the fermionic field, and use them to cast lattice QED as a QECC. Since a choice of QRF amounts to a choice of gauge-fixable subsystem, these QRFs and their associated reduction maps give insights into the encoded physical information. This, in turn, allows us to identify sets of correctable gauge-violating errors in lattice QED. Although a violation of the gauge constraints is necessary for error detection, the full class of errors
associated with the same syndrome is generically highly degenerate, and choosing a QRF resolves this degeneracy.
This makes the syndrome sufficient for recovery, and we further construct the explicit recovery channels from charge-sector measurements based on measuring the local constraints. For a summary of the ideas and construction, see \cref{ssec:summary}.

The connection between QECCs and quantum gauge theories is intriguing at both practical and fundamental level. For quantum simulations of lattice gauge theories, fault-tolerant encodings based on Gauss' law provide a way to reduce the simulation overhead compared to naive schemes and have been the focus of recent study \cite{Stryker_2019, Rajput2023, spagnoli2024, CLLL_2024, GL_2023, yao_2025, pato_2026, spagnoli_2026}. More broadly, the promotion of gauge redundancy to an explicit error correcting structure can provide steps toward a deeper understanding of the role of gauge symmetry through the lens of quantum information.

\paragraph{Structure.}
An outline of the main ideas and the results of this work is given in \cref{ssec:summary}. Because this paper brings together several different research fields, \cref{sec:intro_to_QECQRFGT} provides introductions to the essential basics of gauge systems, the perspective-neutral framework of QRFs and quantum error correction.
Readers already familiar with these topics may wish to proceed directly to \cref{sec:QEC_in_Gauge}, where we discuss gauge systems with a perspective-neutral QRF structure as QECCs and their correctable errors and recovery channels.
In \cref{sec:QRFs_in_LQED}, we introduce lattice QED and construct QRFs for this theory. We use the tools from \cref{sec:QEC_in_Gauge} and \cref{sec:QRFs_in_LQED} to understand lattice QED as a QECC in \cref{sec:LQED_as_QECC}, where we analyze some of its sets of correctable errors and the associated recoveries. 
As a first step toward extending these QRFs and error-correction structures from lattice to continuum QED, \cref{sec:continuum} discusses their continuum counterparts.
We conclude in \cref{sec:conclusion} with a discussion and outlook.

\paragraph{Note.}
The present paper builds in large part on E.R.'s Master thesis (July 2025), available online \cite{mthesis}. Relative to the thesis version, this paper contains improved exposition and substantial new results.

\paragraph{Note added.}  After the completion of E.R.'s master thesis, we became aware of Javier Pagan Lacambra,  Aidan Chatwin-Davies,  Masazumi Honda and Philipp A. Höhn pursuing a closely related topic
\cite{PaganLacambra_forthcoming}, scheduled to appear simultaneously on the arXiv.

\subsection{Overview of the Paper} 
\label{ssec:summary}
The purpose of this section is to provide an overview of the main ideas of this paper and summarize the results, which are discussed in greater detail in the main body of the text.

\paragraph{Quantum error correction and gauge systems.}
\label{ssec:QEC_and_gauge}

The connection between gauge systems and quantum error correction codes (QECCs) that we adopt is based on identifying the gauge-invariant physical states in the physical Hilbert space $\HS_\phys$, embedded in the kinematical space $\HS_\kin$ of the gauge system, with the code subspace $\HS_\code$ embedded in a larger Hilbert space $\HS_\physical$\footnote{Since the term ``physical" appears in both contexts with different meanings, we will use distinct subscripts to separate ``physical" as gauge-invariant from the physical system  used to encode the logical information in QEC.} that stores the encoded information. In short, these identifications are:
\begin{center}
    \begin{tabular}{c c c}
        Gauge System & &  QECC \\
        $\HS_\phys$ & $\leftrightarrow$ & $\HS_\code$ \\
        $\HS_\kin$ & $\leftrightarrow$ & $\HS_\physical$
    \end{tabular}.
\end{center}

\paragraph{Generalized stabilizer codes.}
Stabilizer QECCs define the code space via a symmetry group of Pauli operators $\mathcal{S}$, called the stabilizer group \cite{gottesman1997stabilizer}. The code states are all the invariant states under $\mathcal{S}$, i.e., $S\ket{\psi}_\code = \ket{\psi}_\code$ for all $S\in\mathcal{S}$. In the context of gauge systems, if the gauge transformations form a group $G$, then the physical states are invariant under a (possibly projective) unitary representation $U(g)$ of $G$, i.e., $U(g)\ket{\psi}_\phys = \ket{\psi}_\phys$. Extending the identifications above, we thus interpret $G$ as a ``generalized stabilizer group'' of the gauge system when viewed as a QECC. We can then understand the error correction properties of the gauge system through its representation-theoretic structure (see \cref{fig:gauge_structure_QEC}). For stabilizer codes, this structure was elaborated in detail in \cite{CCHM_2024}. 

Due to the gauge symmetry, the kinematical Hilbert space $\HS_\kin$ decomposes into a direct sum of the charge sectors $\HS_\kin = \bigoplus_{\bm{q}}\HS_{\bm{q}}$ corresponding to isotypes of irreducible representations of $G$. We index the irreducible representations by the charges $\bm{q}$ corresponding to their highest weight, assuming for simplicity that $G$ is compact and its Lie algebra semi-simple.
In this decomposition, $U(g)$ acts on $\HS_{\bm{q}}$ as $U_{\bm{q}}(g)$.
The gauge-invariant states live in the zero-charge sector corresponding to the trivial representation, $\HS_\phys = \HS_{\bm{0}}$.  
A bounded operator $A\in\mathcal{B}(\HS_\kin)$ is gauge-invariant if it is invariant under conjugation, $U(g)AU(g)^\dagger = A$, and maps states in $\HS_\phys$ into $\HS_\phys$. By $A_{\bm{q}}\in\mathcal{B}(\HS_\kin)$ we denote operators which satisfy
\begin{equation}
    A_{\bm{q}}\ket{\psi}_\phys\in\HS_{\bm{q}}\quad\forall\ket{\psi}_\phys\in\HS_\phys.
\end{equation}
Then, $A_{\bm{q}}^\dagger$ maps states in $\HS_{\bm{q}}$ to $\HS_\phys$.

Now, consider a physical state $\ket{\psi}\in\HS_\phys$. The occurrence of an error described by an operator $E$, e.g., in a quantum simulation of the theory, will map $\ket{\psi}$ into an error state $E\ket{\psi}$, which may generally be supported over all of charge sectors $\HS_{\bm{q}}$. However, if we can determine through some measurement the charge $\bm{q}$, this projects the state onto one of the charge sectors, i.e., the post-measurement state is $\Pi_{\bm{q}}E\ket{\psi} \in\HS_{\bm{q}}$ (up to normalization), where $\Pi_{\bm{q}}$ is the orthogonal projector onto $\HS_{\bm{q}}$. This is analogous to the role of syndrome measurements in stabilizer codes (see \cref{ssec:QEC}).
We can complete an error correction protocol if we can choose an operator $A_{\bm{q}}$ for each charge sector $\HS_{\bm{q}}$, such that applying $A_{\bm{q}}^\dagger$ to the error state after measuring the charge $\bm{q}$ maps it back to $A_{\bm{q}}^\dagger\Pi_{\bm{q}}E\ket{\psi} \in \HS_\phys$. This corresponds to the recovery operation
\begin{equation}            \rho_{\text{err}}\mapsto\sum_{\bm{q}}A_{\bm{q}}^\dagger\Pi_{\bm{q}}\rho_{\text{err}}\Pi_{\bm{q}}A_{\bm{q}},
\end{equation}
which is a valid quantum channel if $\Pi_{\bm{q}}A_{\bm{q}}A_{\bm{q}}^\dagger\Pi_{\bm{q}} = \Pi_{\bm{q}}$. The error $E$ is appropriately corrected by this recovery if $A_{\bm{q}}^\dagger\Pi_{\bm{q}}E|_{\HS_\phys} \propto I|_{\HS_\phys}$, where $I|_{\HS_\phys}$ is the identity operator restricted to the subspace $\HS_\phys$. These conditions are satisfied for all errors $E$ which are linear combinations of the operators $A_{\bm{q}}$ if all $A_{\bm{q}}$ are unitary on $\HS_\kin$. Therefore, the choice of operators $A_{\bm{q}}$ in the channel above determines which errors $E$ will be correctable.
In the general case, the choice of such unitary operators $A_{\bm{q}}$ is not unique and their existence not guaranteed\footnote{For example, in finite dimensions, $\dim\HS_{\bm{q}}<\dim\HS_\phys$, then no unitary can map $\HS_\phys$ into $\HS_{\bm{q}}$.}.
Indeed, there are potentially none or many different correctable error sets for any given gauge system acting as a QECC, all based on this active recovery scheme.

We will exploit this idea to construct recovery operations and find correctable error sets in the case of lattice quantum electrodynamics, in which case $G$ is Abelian and consists of local $\operatorname{U}(1)$ gauge transformations.

\begin{figure}[h]
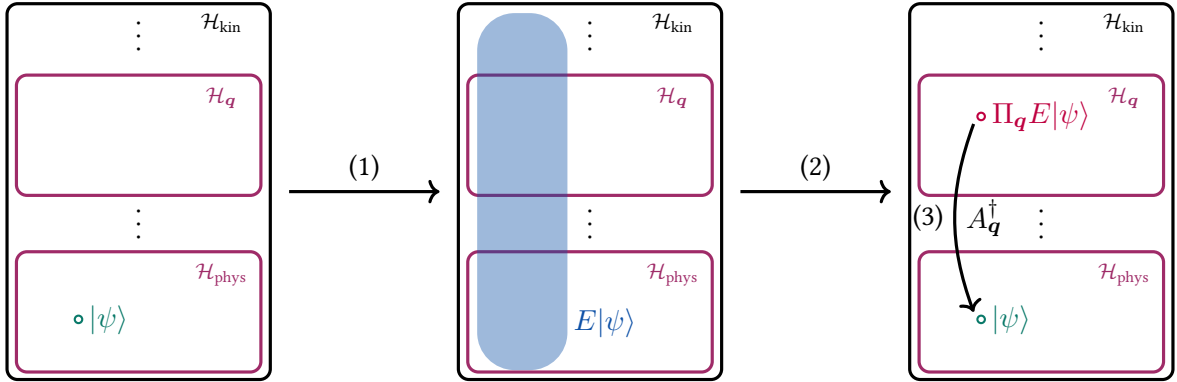

    \centering
    $\centertikz{
        \begin{scope}[shift={(-6,0)}]
            \node[box,minimum height=5cm,minimum width=3.5cm, draw= black] at (0,0) {}; 
            \node[box,minimum height=1.6cm,minimum width=3.25cm, draw= newpurple] at (0,-1.6) {};
            \node[box,minimum height=1.6cm,minimum width=3.25cm, draw= newpurple] at (0,0.75) {};
            \node[color=newpurple] at (1.1,1.25) {\scriptsize$\HS_{\bm{q}}$};
            \node[] at (1.1,2.25) {\scriptsize$\HS_\kin$};
            \node[color=newpurple] at (1.1,-1.1) {\scriptsize$\HS_\phys$};
            \node[rotate=90] at (0,-0.4) {$\dots$};
            \node[rotate=90] at (0,2.1) {$\dots$};
            \node[circ,color=newgreen] at (-0.8, -1.7) {};
            \node[color=newgreen] at (-0.4, -1.7) {$\ket{\psi}$};
        \end{scope}
        \draw[very thick, ->] (-4,0) -- node[above] {(1)} (-2,0);
        \begin{scope}[shift={(0,0)}]
            \node[box,minimum height=5cm,minimum width=3.5cm, draw= black] at (0,0) {}; 
            \node[box,minimum height=1.6cm,minimum width=3.25cm, draw= newpurple] at (0,-1.6) {};
            \node[box,minimum height=1.6cm,minimum width=3.25cm, draw= newpurple] at (0,0.75) {};
            \node[color=newpurple] at (1.1,1.25) {\scriptsize$\HS_{\bm{q}}$};
            \node[] at (1.1,2.25) {\scriptsize$\HS_\kin$};
            \node[color=newpurple] at (1.1,-1.1) {\scriptsize$\HS_\phys$};
            \node[rotate=90] at (0,-0.4) {$\dots$};
            \node[rotate=90] at (0,2.1) {$\dots$};
            \node[box,minimum height=4.75cm,minimum width=1.2cm,rounded corners=0.5cm, fill=myblue, opacity = 0.4] at (-0.9,0){};
            \node[color=myblue] at (0.2,-1.75){$E\ket{\psi}$};
        \end{scope}
        \draw[very thick, ->] (2,0) -- node[above] {(2)} (4,0);
        \begin{scope}[shift={(6,0)}]
            \node[box,minimum height=5cm,minimum width=3.5cm, draw= black] at (0,0) {}; 
            \node[box,minimum height=1.6cm,minimum width=3.25cm, draw= newpurple] at (0,-1.6) {};
            \node[box,minimum height=1.6cm,minimum width=3.25cm, draw= newpurple] at (0,0.75) {};
            \node[color=newpurple] at (1.1,1.25) {\scriptsize$\HS_{\bm{q}}$};
            \node[] at (1.1,2.25) {\scriptsize$\HS_\kin$};
            \node[color=newpurple] at (1.1,-1.1) {\scriptsize$\HS_\phys$};
            \node[rotate=90] at (0,-0.4) {$\dots$};
            \node[rotate=90] at (0,2.1) {$\dots$};
            \node[circ,color=newgreen] at (-0.8, -1.7) {};
            \node[color=newgreen] at (-0.4, -1.7) {$\ket{\psi}$};

            \node[circ,color=purple] (Pi) at (-0.8, 1.) {};
            \node[color=purple]  at (-0., 1.) {$\Pi_{\bm{q}}E\ket{\psi}$};
            \draw[very thick, ->] (-0.9,0.9) to [in=110, out=250] node[left] {(3)} node[right] {$A_{\bm{q}}^\dagger$} (-0.9, -1.6);
        \end{scope}

    }$
    \caption{Under the gauge symmetry, $\HS_\kin$ decomposes into a direct sum of charge sectors $\HS_{\bm{q}}$, where $\HS_\phys$ corresponds to the trivial representation. (1) An error $E$ maps a physical state $|\psi\rangle\in\HS_\phys$ to the error state $E|\psi\rangle$ which may have support spread across every charge sector. (2) A measurement collapses the state to one of the charge sectors, $\Pi_{\bm{q}}E|\psi\rangle\in\HS_{\bm{q}}$, and (3) applying an operator $A_{\bm{q}}^\dagger$ recovers the original state if $A_{\bm{q}}^\dagger\Pi_{\bm{q}}E|\psi\rangle\propto|\psi\rangle$. }
    \label{fig:gauge_structure_QEC}
\end{figure}

\paragraph{Perspective-neutral quantum reference frames and error correction.}
A useful tool to make the connection between gauge systems and QECCs explicit is the perspective-neutral construction of QRFs \cite{Vanrietvelde2020changeof, HSL_2021, delahamette_2021, Krumm_2021, H_hn_2022}. In gauge systems, perspective-neutral QRFs are subsystems endowed with a set of orientation states which transform covariantly under the action of the gauge group $G$.
Importantly for this work, such QRFs constitute gauge-fixable subsystems and identifying a QRF splits the kinematical space into physical and gauge-redundant parts.
The perspective-neutral construction thus provides a convenient setting in which we can analyze the relationship between redundant kinematical degrees of freedom and the error-correction properties of gauge systems.

Previous investigations on the connection between stabilizer codes and gauge systems also employed the perspective-neutral framework of QRFs \cite{sem_proj, CCHM_2024}. Notably, \cite[Theorem 4.13]{CCHM_2024} showed that there is a 1-to-1 relation between correctable error sets and ideal QRFs (of a specific type) for stabilizer codes. Moreover, it was shown that every ideal QRF supports a set of correctable errors consisting of gauge-fixing projectors \cite[Lemma 5.2]{CCHM_2024}.

This latter result can be extended from stabilizer codes to general perspective-neutral systems with non-ideal QRFs (see \cref{thm:correctable_gauge-fixing}). In particular, we consider a QRF $R$ of a gauge system $\HS_\kin = \HS_R\otimes \HS_S$ that transforms under the representation $U_R(g)\otimes U_S(g)$ of the compact gauge group $G$. The orientation states of $R$, $\{\ket{\phi(g)}_R\, | \, g\in G\}$ transform according to $U_R(h)\ket{\phi(g)}_R = \ket{\phi(hg)}_R$ and the gauge-fixing operators $\mathcal{P}^g\propto\ket{\phi(g)}\bra{\phi(g)}_R\otimes I_S$ fix the orientation of $R$ to $\ket{\phi(g)}_R$. We show that these gauge-fixing operators are correctable errors provided their orientation states are orthogonal. This generalization is crucial for determining the correctable error sets associated with QRFs that can be identified within gauge systems.

Furthermore, if $G$ is Abelian and the QRF $R$ ideal, \cref{prop:A_q_from_gf} connects the correctable gauge-fixing operators to the recovery protocol via charge measurements illustrated in \cref{fig:gauge_structure_QEC} by constructing a set of operators $\{A_{\bm{q}}\}_{\bm{q}}$ from gauge-fixing operators, inspired by the duality between Pauli errors and gauge-fixing errors worked out in \cite{CCHM_2024}. If $R$ is non-ideal, this result can still be applied if $G$ contains some subgroup $H$ for which $R$ becomes an ideal QRF. This leads to a recovery based on coarse-grained measurements of the charge sectors.

\paragraph{Lattice Quantum Electrodynamics.}
We study the connection between gauge systems and quantum error correction using lattice quantum electrodynamics (QED) as an example, working in the Kogut–Susskind Hamiltonian formulation \cite{KS1975} in temporal gauge. We denote the set of oriented {\em links} (or {\em edges}) of the lattice by $\mathcal{L}$ and the set of {\em vertices} (or {\em sites}) by $\mathcal{V}$. The links $l\in\mathcal{L}$ are assigned a Hilbert space $L^2\left(\operatorname{U}(1)\right)$ representing the gauge field degrees of freedom. Link states in $L^2\left(\operatorname{U}(1)\right)$ can be expressed in terms of the $\operatorname{U}(1)$-position basis $\{\ket{e^{i\theta}}_l\}_{\theta\in[0, 2\pi)}$ or the Fourier transformed electric flux basis $\{\ket{k}\}_{k\in\mathbb{Z}}$.
On each link, the electric flux operator $\epsilon_l\ket{k}_l = k\ket{k}_l$ is canonically conjugate to the $\operatorname{U}(1)$-position operator $U_l\ket{e^{i\theta}}_l = e^{i\theta}\ket{e^{i\theta}}_l$, i.e., $[\epsilon_l, U_l] = U_l$. 
In the pure gauge field sector, the net electric flux at any vertex $v\in\mathcal{V}$ is zero,
\begin{equation}
    \mathcal{C}_v = \sum_{l_{\text{out}}}\epsilon_{l_{\text{out}}} - \sum_{l_{\text{in}}}\epsilon_{l_{\text{in}}}\overset{\text{on }\HS_\phys}{=}0,
\end{equation}
where the sums run over all outgoing and incoming links at $v$, $l_{\rm out} = [v, *]$ and $l_{\rm in} = [*, v]$.
This constraint enforces Gauss' law and generates the group of gauge transformations $\mathcal{G}$.

Fermionic matter and its dynamics can be introduced in the model by adding field operators $\psi_v, \psi_v^\dagger$ on each site $v$ of the lattice. We employ the staggered fermionic field prescription \cite{KS1975}. The constraints of the theory are then modified to include the charge density $\rho_v$, defined in terms of the staggered fields, as $\mathcal{C}_v = \sum_{l_{\text{out}}}\epsilon_{l_{\text{out}}} - \sum_{l_{\text{in}}}\epsilon_{l_{\text{in}}} - \rho_v$, representing Gauss' law $\nabla\cdot E=\rho$ on the lattice.

\paragraph{QRFs for lattice QED.}
We construct two types of QRFs for lattice QED by applying the perspective-neutral framework.
For the pure gauge theory, an ideal QRF can be chosen as the space of links on a spanning tree (a loop-free subgraph connecting all vertices with a unique path) of the lattice (\cref{thm:LGT_QRF_pure}). Gauge-fixing these QRFs maps the physical degrees of freedom, which are encoded in closed loops on the lattice, to the links outside of the spanning tree. This is similar to the spanning tree QRFs introduced in \cite{AHS_25}, which were used to  classify relational entanglement entropies for lattice gauge theories. Spanning trees have commonly been used to gauge-fix lattice gauge theories without explicitly invoking QRFs, see for example \cite{GL_2010_textbook} for an introduction.

In the presence of staggered fermions, we can use the fermionic field as a QRF (\cref{thm:matter_qrf}). This frame can resolve and be used to gauge-fix the local phase of $\psi_v$. Due to the fermionic nature of the matter, there is only a qubit degree of freedom available to parametrize a $\operatorname{U}(1)$-action on each site of the lattice. Therefore, these constitute non-ideal QRFs with non-orthogonal orientation states.

\paragraph{Lattice QED as a generalized stabilizer code.}
According to the connection between gauge systems and QECCs outlined above, we can understand lattice QED as an error correction code. There are two types of such codes, the first coming from pure gauge lattice QED and the second from lattice QED with staggered fermions. These differ in the structure of the kinematical Hilbert space: While the former consists only of $L^2(\operatorname{U}(1))$ spaces, the latter also includes a qubit space $\mathbb{C}^2$ per site due to the fermions. Both types can correct subsets of gauge-violating errors.

For pure gauge lattice QED, applying \cref{thm:correctable_gauge-fixing} to the spanning tree QRFs described above yields correctable error sets which correspond to the gauge-fixing operators on a spanning tree $T$. Via \cref{prop:A_q_from_gf}, these are equivalent to the set of all the $U$-type operators supported on $T$ (\cref{prop:correctable_U_on_T_in_QED}),
\begin{equation}
    \left\{\bigotimes\nolimits_{l\in T}U_l^{m_l} \ \middle| \ \forall l\in T:m_l \in \mathbb{Z}\right\}.
\end{equation}
Physically, these operators act as shifts in electric flux on the spanning tree $T$.
We can also find other correctable error sets by following the logic outlined in \cref{ssec:QEC_and_gauge}. In the case of $\operatorname{U}(1)$ lattice QED, we can measure the charge of a state by measuring the Gauss constraints $\mathcal{C}_v$. With an appropriate choice of operators $A_{\bm{q}}$, we can find that the set of single electric flux shifts
\begin{equation}
    \{U_l^m \}_{m\in\mathbb{Z}, l\in\mathcal{L}}
\end{equation} 
can be corrected (\cref{prop:correctable_U_errors}). This is the $\operatorname{U}(1)$ version of a result from \cite{CLLL_2024}, who showed that for pure gauge systems with finite gauge groups, gauge-violating electric flux errors on any single link are correctable errors.

For staggered fermion lattice QED, the fermionic field QRF supports correctable error sets of orthogonal gauge-fixing operators. 
These errors can be combined to an equivalent set of local occupation number flips with a relative phase at arbitrary sites, which can be expressed in terms of the field operators $\psi_v = \ket{0}\bra{1}_v$ and $\psi_v^\dagger = \ket{1}\bra{0}_v$ as $A_v(\alpha_v) = e^{i\alpha_v}\psi_v + e^{-i\alpha_v}\psi_v^\dagger$. Equivalent to the gauge-fixing operators, we thus find the correctable error set (\cref{prop:correctable_matter_errors})
\begin{equation}
    \left\{\bigotimes\nolimits_{v \in \mathcal{V}} A_v(\alpha_v)^{r_v} \ \middle| \ \forall v:r_v \in \{0, 1\}\right\}.
\end{equation}
By a different recovery based on a coarse-grained measurement of the constraints, we can also correct the set
\begin{equation}
    \{U_l^m \}_{m\in\mathbb{Z}, l\in\mathcal{L}}\cup\{A_v(\alpha_v)\}_{v\in\mathcal{V}}.
\end{equation}
Thus, any single electric flux error on a link or flip in occupation number on a site is correctable.
This result (\cref{thm:correctable_errors_full_QED}) was previously shown in \cite{Rajput2023, spagnoli2024} for lattice QED with a truncated local $\mathbb{Z}_2$-gauge group. In particular, \cite[Theorem 1]{spagnoli2024} shows that using Gauss's law to correct errors can protect against any single flip in electric flux on the links or flip of fermion occupation number on the sites \cite{Rajput2023, spagnoli2024}. Compared to the $\mathbb{Z}_2$-version, the $\operatorname{U}(1)$-theory can correct any integer error in electric flux and the more general occupation number flips including relative phases.

\paragraph{The continuum counterparts.} Finally, we identify  the continuum analogues  of both types of QRFs and of the associated correctable errors in \cref{sec:continuum}. On the gauge-field side, spanning-tree QRFs can be related to contour-gauge constructions \cite{Anikin2025}, while the lattice $U$-type errors correspond to smeared Wilson line operators. On the matter side,  the fermionic non-ideal QRF admits a continuum analogue in the field-basis description of the Dirac field, with orientation states given formally by functional superpositions of eigenstates of the local mode-density operator weighted by a local
$U(1)$-phase parameter. The corresponding family of local fermionic error operators take the form of smeared phase-rotated Majorana type operators. We further discuss the challenges for
a full QECC construction of continuum QED, which will require new field-theoretic tools.

\section{Preliminaries}
\label{sec:intro_to_QECQRFGT}

In this section, we review the basics and relevant features for gauge theory, the perspective-neutral formalism of Quantum Reference Frames (QRFs), and Quantum Error Correction (QEC). The individual subsections are intended to provide a concise overview of each topic. For more detailed introductions, see \cref{app:QRFs} for the perspective-neutral formalism and \cref{app:QEC} for QEC.
Readers already familiar with these formalisms may wish to skip this section.

\subsection{Gauge Theory}

We first briefly review the Hamiltonian formalism of gauge systems and their quantization \cite{dirac1964, Hennex_Teitelboim_1992}.  A gauge theory can be understood as a constrained dynamical system whose variables do not represent independent physical degrees of freedom, but are restricted by constraints. In the canonical formalism, one starts from a classical phase space $P$ together with a Hamiltonian in the following form:
\begin{equation}
H_{\rm tot} = H  + \lambda^I C_I
\end{equation}
where $\lambda^I$ are Lagrange multipliers and $C_I$ are the constraints. The \emph{primary constraints} hold without imposing the equation of motion,  while requiring their preservation under time evolution $\{H, C_I\} = 0$ may generate further \emph{secondary constraints}. The primary and secondary constraints together identify the constraint surface $\Sigma \subset P$ within phase space. For an Abelian gauge theory, such as electrodynamics and linearized gravity, the constraints commute among themselves on the constraint surface, i.e., $\{C_I, C_J \} \approx 0$ \footnote{ We use $\approx$ to denote ``weak equality", i.e., equality on the constraint surface.}. For non-Abelian gauge theories, for instance $\operatorname{SU}(N)$ Yang-Mills theory, the constraints are closed under the Poisson bracket, $\{C_I, C_J\}\approx  f_{IJ}^K C_K$, with structure constants or functions $f_{IJ}$.
An (infinitesimal) gauge transformation of a phase-space function generated by a constraint $C_I$ is given by
\begin{equation}
\delta_\epsilon F = \epsilon^I \{F, C_I\}
\end{equation}
The points in phase space which are related by a gauge transformation are regarded as physically equivalent; they lie in the same gauge orbit generated by $C_I$ within $\Sigma$. Fixing a specific gauge then corresponds to choosing a representative along each such orbit. A \emph{physical observable} ${O}_{ \phys}$ is defined as a function on phase space that is invariant under gauge transformations
\begin{equation}\label{eq:classicalOph}
\{C_I, {O}_{ \phys}\} \approx 0. 
\end{equation}
Such functions are constant along the gauge orbits, meaning that their values are independent of which specific gauge one chooses.

For quantum gauge theory, another classification of constraints is particularly relevant.
Constraints whose Poisson brackets with all other constraints vanish on the constraint surface $\Sigma$ are called \emph{first class}, we denote them by $C_{\alpha}$. By Dirac's conjecture \cite{dirac1964}, they generate gauge transformations and encode redundancies of description rather than physical restrictions. In contrast, the \emph{second class constraints} have non-vanishing Poisson brackets on the constraint surface $\Sigma$. They arise when the constraints impose genuine restrictions on the phase space, for instance a gauge-fixing, and are treated by replacing the Poisson bracket with the so-called Dirac bracket prior to quantization. 

Upon quantization, the classical phase space  is replaced by a kinematical Hilbert space $\HS_{\kin}$, which carries the canonical operator algebra but still contains gauge redundancy. 
The physical Hilbert space $\HS_{\phys}$ consists of the states which are annihilated by all first-class constraints $\hat{C}_{\alpha}$,
\begin{equation}
\hat{C}_{\alpha} |\psi\rangle_{\phys} =0.
\end{equation}  
Parallel to \cref{eq:classicalOph}, the quantized physical observables commute with the constraints 
\begin{equation}
[\hat{C}_{\alpha}, \hat{O}_{\phys} ]=0.
\end{equation}
By exponentiation, the constraints $\hat{C}_\alpha$ generate the group of gauge transformations $G$, i.e., $e^{s\frac{i}{\hbar}\hat{C}_{\alpha}}\in G$.

For locally compact gauge groups $G$, gauge invariance can be imposed by group averaging through the Haar measure $\D g$. If $G$ is compact and we can normalize the Haar measure by the volume of $G$ to $\frac{\D g}
{|G|}$, then the projector onto the physical Hilbert space is 
\begin{equation}
\Pi_{\phys} = \frac{1}{|G|} \int_G \D g \ U(g),
\end{equation}
where we use $U(g)$ to denote the unitary representation of $G$ on $\HS_\kin$. A physical state satisfies $U(g)  |\psi\rangle_{\phys} =  |\psi\rangle_{\phys}$.
Particularly relevant for this work is the case of quantum electrodynamics, for which a discretization to a lattice allows the group of all gauge transformations to be locally compact.

\subsection{Quantum Reference Frames and the Perspective-Neutral Approach}
Quantum reference frames (QRFs) have been developed as a framework for describing quantum systems relationally and transformations between different perspectives on the same underlying scenario~\cite{Aharanov_Susskind_1, Aharanov_Susskind_2, Aharanov_1984, Bartlett_2007, Gour_2008, Gour_2009, Miyadera_2016, Loveridge2018, Carette_2025, Giacomini2019, de_la_Hamette_2020, Vanrietvelde2020changeof,HSL_2021, delahamette_2021, Krumm_2021, H_hn_2022}.
Early work on this idea dates back to Aharonov and Susskind in 1967 \cite{Aharanov_Susskind_1, Aharanov_Susskind_2}, which connected superselection rules with the presence or absence of a reference system, and investigated the idea of including a reference frame to make coherent superpositions across different charge superselection sectors physically possible. 
Later work analyzed quantum mechanical finite-mass objects as quantum reference frames, invoking for the first time the term “quantum reference frame” \cite{Aharanov_1984}.

Since then, a multitude of QRF frameworks have been developed.
 \textit{The quantum information approach} considers the accessible information in the absence of an external frame for states obtained by an incoherent average over the symmetry group \cite{Bartlett_2007, Gour_2008, Gour_2009}, and the related \emph{the extra particle approach} addresses the transformation between perspectives \cite{castroruiz2023, Garmier:2025soc}. 
 \textit{The operational approach} 
studies how the Positive Operator Valued Measure (POVMs) can be invariant under symmetries and constructs the corresponding notion of relational observables \cite{Miyadera_2016, Loveridge2018, Carette_2025}.
 \textit{The perspectival approach} studies how relational observables and the quantum states of other systems are described relative to a given quantum system, and how to transform between different frames \cite{Giacomini2019, de_la_Hamette_2020}.  \textit{The perspective-neutral approach} \cite{Vanrietvelde2020changeof,HSL_2021, delahamette_2021, Krumm_2021, H_hn_2022} is naturally connected to the Dirac quantization of gauge systems.   The physical states are defined through constraints and admit multiple equivalent representations related by gauge transformations, while the reduction map allows us to consistently reduce to each QRF perspective. This has been shown to be equivalent to the perspectival approach for ideal frames \cite{Vanrietvelde2020changeof}. In general, the different QRF frameworks are suited to different physical situations, and the precise relationships between them remain an open research topic \cite{DeVuyst:2025ezt, Castro-Ruiz:2025yvi,delahamettekabel_2026, doat_2025}.

The various approaches to QRFs have found applications ranging from quantum information theory \cite{Kitaev_2004, Bartlett_2006, Bartlett_2007, Palmer_2014}, uncertainty relations \cite{riera2024}, thermodynamics \cite{Hoehn:2023ehz} to gauge  theory \cite{araujoregado2025} and quantum spacetime \cite{Giacomini:2020ahk, delaHamette:2021iwx, Kabel_2023, kabel2024,  Freidel:2025ous}. It has been shown that correlations, entanglement \cite{Giacomini2019, cepollaro2024}, entropy \cite{devuyst2024, araujoregado2025, delaHamette:2026rtv}, the very
notion of the quantum subsystem \cite{Giacomini2019, Ali_Ahmad_2022, hoehn2023,Hoehn:2023ehz}, and localization of events \cite{Giacomini:2020ahk, Castro-Ruiz:2019nnl, kabel2024, Vilasini:2025qun} are all QRF-dependent concepts. Including a clock reference frame is essential in the recent understanding of the type reduction of von Neumann algebra in field theory \cite{Chandrasekaran:2022cip, Fewster_2024, DeVuyst:2024fxc}.   Non-ideal QRFs are particularly important, since quantum systems as QRFs do not come with infinite resources \cite{Garmier:2025soc}. 

This section is intended as an introduction to the perspective-neutral construction as presented in \cite{delahamette_2021}, for more details see \cref{app:QRFs}.

\paragraph{QRFs and orientation states.}
Let us consider a kinematical Hilbert space $\HS_{\text{kin}}$ with a physical subspace $\HS_\phys$ which is invariant under the action of a gauge group $G$, i.e., for all $\ket{\psi}_\phys \in \HS_\phys$ 
\begin{equation}
    U(g) \ket{\psi}_{\phys} = \ket{\psi}_{\phys}
\end{equation}
where $U(g)$ is a representation of $G$. The physical subspace $\HS_\phys$ is associated with an orthogonal projector which, for compact groups with normalized Haar measure $\frac{\D g}{|G|}$, equals
\begin{equation}
    \Pi_\phys \defeq \frac{1}{|G|}\int_G \D g \ U(g).
\end{equation}
Such an expression can be generalized for unimodular locally compact Lie groups as in \cite{delahamette_2021}, however, we will assume compactness to avoid issues with improper projectors and rigged Hilbert spaces.

Consider furthermore a factorization of the kinematic space $\HS_\kin$ into a reference frame $R$ and a remaining system $S$, $\HS_{\text{kin}}=\HS_R\otimes \HS_S$, such that $U(g)$ is of the form $U_{RS}(g) = U_R(g)\otimes U_S(g)$, where $U_R$ and $U_S$ are representation of the group\footnote{These representations may generally be projective. In this work, this will not be the case.}. 
The subsystem $R$ serves as a QRF if it parametrizes the action of $G$, thus, $R$ admits a set of {\it orientation states} $\{\ket{\phi(g)}_R  \, | \,  g\in G\}$ (called a coherent state system in \cite{delahamette_2021}) on which $G$ acts transitively, 
\begin{equation}
    U_R(g')\ket{\phi(g)}_R = \ket{\phi(g'g)}_R,
\end{equation} 
and which can resolve identity,
\begin{equation}
    \int_G \D g \ \ket{\phi(g)}\bra{\phi(g)}_R = cI_R
    \label{eq:resolution_of_I}
\end{equation}
for $c>0$. The orientation states thus span $\HS_R$ and can be obtained as the $G$-orbit of a seed state $\ket{\phi(e)}_R$ where $e\in G$ is unity.

If $G$ acts freely on the orientation states, i.e., $U_R(g')\ket{\phi(g)}_R= e^{i\theta}\ket{\phi(g)}_R$ if and only if $g' = e$, the orientation states are sufficient to completely parametrize the gauge group $G$, and the corresponding QRFs are called {\it complete.}
In addition, a QRF is {\it ideal} if $\braket{\phi(g)}{\phi(g')}_R = \delta(g, g')$ and {\it non-ideal} otherwise.

\paragraph{Reduction maps and gauge-fixing operators.}
For a given kinematical state, the physical state $\ket{\psi}_\phys =\Pi_\phys \ket{\psi}_\kin$ is a coherent superposition of the states in the $G$-orbit of $\ket{\psi}_\kin$, each corresponding to a specific orientation of $R$. Reduction maps allow to extract the state of $S$ as seen from the perspective of $R$ in a fixed orientation from $\ket{\psi}_\phys$. 

One way of defining such reduction maps is a generalization of the \textit{Page-Wootters} mechanism, developed in the context of quantum clocks \cite{PageWooters_1983}. These maps $\mathcal{R}_R^g:\HS_{\text{kin}}\rightarrow\HS_S$ are defined as
\begin{equation}
    \mathcal{R}_R^g = \sqrt{N}\left(\bra{\phi(g)}_R \otimes I_S\right) \Pi_\phys,
\end{equation}
where $N = \frac{|G|}{c}$, and $c$ arises from the resolution of the identity of orientation states \cref{eq:resolution_of_I}. The reduction map $\mathcal{R}_R^g$ simultaneously fixes the frame's orientation and removes the reference $R$ by taking the inner product with $\bra{\phi(g)}_R$.

Here we denoted the output space of the reduction map as $\HS_S$. Its image is $\mathcal{R}_R^g(\HS_\phys)=\HS_{S|R}^g$, the Hilbert space of the system $S$ from the perspective of $R$. While generally, this space is not equal to $\HS_S$, the two are isomorphic for ideal frames. In addition, we define the reduction map on all of $\HS_\kin$ by explicitly including the projector onto $\HS_\phys$.
With this in mind and for simplicity, we will often write $\mathcal{R}_R^g:\HS_\phys\to\HS_S$, commenting in case there is any ambiguity arising. 

Importantly, $\mathcal{R}_R^g$ is an isometry from $\HS_\phys$ into $\HS_S$, and the adjoint $(\mathcal{R}_R^g)^\dagger$ serves as its inverse \cite[Lemma 8]{delahamette_2021},
\begin{equation}
    (\mathcal{R}_R^g)^\dagger \mathcal{R}_R^g = \Pi_\phys.
\end{equation}
Hence, when restricted to $\HS_\phys$, $\mathcal{R}_R^g$ is unitary onto its image $\HS_{S|R}^g$ (see \cref{app:QRFs}, \cref{lemma:isometry}).

Alternatively, one can introduce the \textit{gauge-fixing operator}
\begin{equation}
    \mathcal{P}_R^g \defeq \sqrt{N}\left(\ket{\phi(g)}\bra{\phi(g)}_R\otimes I_S \right),
    \label{eq:gauge-fixing}
\end{equation}
which, up to normalization, is an orthogonal projector. Restricted to $\HS_\phys$, its image is 
\begin{equation}
    \mathcal{P}_R^g(\HS_\phys) = \operatorname{span}\{\ket{\phi(g)}_R\}\otimes\HS_{S|R}^{ g}\subset\HS_{\text{kin}},
\end{equation} 
the space of states with the orientation of $R$ gauge-fixed to $\ket{\phi(g)}_R$. As operators on $\HS_\phys$, $\mathcal{R}_R^g$ and $\mathcal{P}_R^g$ are unitarily equivalent via the  map $\ket{\psi}_{S|R}\mapsto \ket{\phi(g)}_R\otimes\ket{\psi}_{S|R}$. Therefore, gauge-fixing $\mathcal{P}_R^g:\HS_\phys\rightarrow\operatorname{span}\{\ket{\phi(g)}_R\}\otimes\HS_{S|R}^{g}$ is unitary.

If $R$ is a complete QRF, then $\mathcal{P}_R^g$ constitutes a complete gauge-fixing of physical states. Otherwise, this is akin to a partial gauge-fixing, leaving some redundancy which is due to the subgroup of gauge transformations which the frame can not resolve.

\subsection{Quantum Error Correction}
\label{ssec:QEC}
This section (and \cref{app:QEC}, in more detail) follows the standard literature. General introductions to quantum error correction can be found, for example, in \cite{Nielsen_Chuang_2010, preskill1999chapter7}.

A fundamental discovery for quantum information science was that computers based on quantum mechanics can solve certain problems drastically more efficiently than their classical counterparts. However, any realistic quantum system in a lab interacts with its environment, introducing errors into calculations and memory. To correct such errors, some form of redundancy needs to be introduced into the system. This can be achieved by encoding the quantum information into highly entangled states in a larger physical system.

Formally, a quantum error correction code (QECC) is defined by an isometric embedding of a \textit{logical} Hilbert space $\HS_{\log}$ into the \textit{physical }Hilbert space\footnote{Note that this is different from the physical Hilbert space in gauge theory.} $\HS_\physical$. The image of this embedding is the \textit{code subspace} $\HS_\code$, and we denote the associated orthogonal projector by $\Pi_\code$. 
In the common case where the quantum information is assumed to be carried by 
qubits, a QECC which encodes $k$ logical qubits into $n$ physical ones is referred to as an $[[n,k]]$-code.

A noise channel $\mathcal{N}$ with Kraus operators $\mathcal{E}=\{E_i\}_i$ can be corrected by a QECC if there exists a recovery operation $\mathcal{R}$ such that $\mathcal{R}\circ\mathcal{N}(\rho)=\rho$ for all $\rho\in S(\HS_\code)$. This requirement is equivalent to the {\it Knill-Laflamme conditions} \cite{KnillLaflamme97}
\begin{equation}
    \Pi_\code E_i^\dagger E_j \Pi_\code = c_{ij}\Pi_\code \quad \forall E_i, E_j\in\mathcal{E}.
    \label{eq:Knill_Laflamme}
\end{equation}
If these are satisfied, we call $\mathcal{E}$ a {\it correctable set of errors}. By linearity, any noise channel with Kraus operators that are linear combinations of the errors $E_i$ is also corrected by $\mathcal{R}$. In this sense, error sets which linearly span the same space are equivalent.

\paragraph{Stabilizer codes.}

The stabilizer formalism~\cite{gottesman1997stabilizer} describes a special class of QECCs, whose code space is defined as the joint $+1$-eigenspace of the {\it stabilizer group} $\mathcal{S}$, a subgroup of the $n$-qubit Pauli group $\mathcal{P}_n$. In other terms, 
\begin{equation}
    \ket{\psi}\in\HS_\code \; \iff \; S\ket{\psi} = \ket{\psi} \quad \forall S \in \mathcal{S}.
\end{equation}
Notice that the code space is non-empty only if $\mathcal{S}$ Abelian and $-I\not\in \mathcal{S}$. A stabilizer group $\mathcal{S}$ of a $[[n, k]]$-code has $n-k$ independent generators
\begin{equation}
    \mathcal{S} = \langle S_1, S_2, \dots, S_{n-k} \rangle,
\end{equation}
while the logical operators which map $\HS_\code$ to itself commute with any stabilizer $S \in \mathcal{S}$. Since the $n$-qubit Paulis $\{I, X, Y, Z\}^{\otimes n}$ span the space of all $n$-qubit matrices, any logical operator is a linear combination of the logical Pauli operators which are given by
\begin{equation}
    \Log(\HS_\code) = \{P\in\mathcal{P}_n \, | \, [P,S]=0 \ \forall S\in\mathcal{S}\}/\mathcal{S}.
\end{equation}
The quotient by the stabilizer group $\mathcal{S}$ ensures that all operators which act on $\HS_\code$ in the same way are identified in the same equivalence class.

As with the logical operators, any error operator can be decomposed into Pauli components. Therefore, if a code can correct the set of all Pauli operators up to weight $t$ (where the weight $w(P)=|\supp(P)$| measures the number of non-identity tensor factors of $P$), i.e.,
\begin{equation}
    \mathcal{E} = \{P\in\{I, X, Y, Z\}^{\otimes n} \, | \, w(P) \leq t\},
\end{equation}
it can correct {\it any} error acting on $t$ or fewer qubits. We then say that the code corrects $t$ errors.

The code distance is an important parameter of a QECC. For stabilizer codes, it is the minimum weight of any non-trivial logical operator,
\begin{equation}
    d = \min_{[I]\neq[P]\in\Log(\HS_\code)}w(P).
\end{equation}
Usually, we call an $[[n, k]]$-code with distance $d$ an $[[n, k, d]]$-code. A code with distance $d$ corrects $t$ errors, where $t\leq \frac{d-1}{2}$ \cite[Thm. 1]{KnillLaflamme00}.

Elements of the Pauli group either commute or anti-commute. Thus, for a Pauli error $E$ and a generator $S_i$ of $\mathcal{S}$, we have
\begin{equation}
    S_i E = s_i ES_i
\end{equation}
for $s_i=\pm 1$. For a fixed error $E$, the string $s = (s_1, s_2, \dots, s_{n-k})\subset\{\pm 1\}^{n-k}$, where each element corresponds to a generator of $\mathcal{S}$, is called \textit{syndrome}. Measuring all generators of $\mathcal{S}$, which necessarily commute, yields the string $s$ and is called a \textit{syndrome measurement}.

A Pauli error set $\mathcal{E}$ (including identity) is correctable if any $E_i, E_j \in\mathcal{E}$ either have different syndromes or $E_i^\dagger E_j \in\mathcal{S}$. This means that one can measure the syndrome on an error state and subsequently apply an operator $E_i^\dagger$ with matching syndrome to correct the error.

\section{Correctable Errors for Perspective-Neutral Systems}
\label{sec:QEC_in_Gauge}
In a previous work \cite{sem_proj}, we investigated a connection between gauge systems, QRFs and QECCs at the example of stabilizer codes \cite{sem_proj}, where a type of QRF was identified as certain subsystems whose erasure can be recovered by the code. Following a similar motivation, the authors in \cite{CCHM_2024} present different constructions for QRFs in stabilizer codes, showing that correctable error sets are in a 1-to-1 relation to ideal, non-local (with respect to the tensor product structure of physical qubits) QRFs \cite[Theorem 4.13]{CCHM_2024}. Notably, \cite[Lemma 5.2]{CCHM_2024} establishes that for any ideal QRF in a stabilizer code (and, in fact, in any system with a finite Abelian gauge group), the set of gauge-fixing projectors is a correctable error set.

Here, we consider the opposite direction to most of the discussion \cite{sem_proj, CCHM_2024}: Interpreting general gauge systems with a perspective-neutral structure as QECCs, can we also find correctable error sets tied to QRFs? Through a generalization of Lemma 5.2 from \cite{CCHM_2024} (\cref{thm:correctable_gauge-fixing}), this is indeed possible both for compact Abelian and non-Abelian gauge groups. In the Abelian case, this leads to a recovery scheme via charge-sector measurements.

\subsection{Reduction Maps as Encoding Isometries} 

Consider a choice of quantum reference frame (QRF) that induces a factorization of the kinematical Hilbert space as $\HS_\kin = \HS_R\otimes\HS_S$. We assume that the action of a compact gauge group $G$ on $\HS_\kin$ splits into a product $U_{RS}(g) = U_R(g)\otimes U_S(g)$, where $U_R(g)$ and $U_S(g)$ are unitary representations on $\HS_R$ and $\HS_S$ respectively. 
Let $\{\ket{\phi(g)}_R  \, | \,  g\in G\}$ be a set of orientation states of the QRF, and consider the associated Page-Wootters reduction maps
\begin{equation}
\mathcal{R}_R^{g}:\HS_\phys \rightarrow \HS_{S|R}^g.
\end{equation}

Recall that a QECC is defined by an encoding isometry $T:\HS_{\rm logical}\rightarrow \HS_\code\subset\HS_\physical$. In the QECC picture of the above perspective-neutral system, 
the role of such an encoding isometry is played by the inverse reduction maps $(\mathcal{R}_R^g)^\dagger$ as pointed out in the dictionary in \cite{CCHM_2024}.

Suppose now that $R$ is an ideal QRF, such that the orientation states $\{\ket{\phi(g)}_R\, |\, g\in G\}$ form an orthonormal basis of $\HS_R$. Then, the reduction maps $\mathcal{R}_R^g:\HS_\phys\rightarrow\HS_{S|R}^g$ are surjective, i.e., $\HS_{S|R}^g = \HS_S$, as follows immediately from the projector $\mathcal{R}_R^g(\mathcal{R}_R^g)^\dagger$ onto their image,
\begin{equation}
\begin{aligned}
    \mathcal{R}_R^g(\mathcal{R}_R^g)^\dagger &= {|G|}(\bra{\phi(g)}_R\otimes I_S)\Pi_\phys(\ket{\phi(g)}_R\otimes I_S)\\ &= \int \D h \ \braket{\phi(g)}{\phi(hg)}_R U_S(h) = I_S \, .
\end{aligned}
\end{equation}
Thus, for ideal QRFs, the encoded logical data in $\HS_\phys$ coming from $\HS_{{S|R}}^g$ takes up all of the residual system $\HS_S$; and the inverse reduction maps act as encoding isometries
\begin{equation}
\centertikz{
     \node at (0,0) {$(\mathcal{R}_R^g)^\dagger :\HS_{S|R}^g\cong\HS_S\;\to\;\HS_\phys$};
     \node at (0.,-0.5) {$\updownarrow$};
     \node at (1.4,-0.5) {$\updownarrow$};
      \node at (0.,-1.) {$\HS_{\rm logical}$};
     \node at (1.4,-1.) {$\HS_\code$};
}.
\end{equation}
The logical states in $\HS_S$ associated with differently oriented reduction maps $\mathcal{R}_R^g$ and $\mathcal{R}_R^h$ may be related by a rotation as captured by the unitary transformation $\mathcal{R}_R^h(\mathcal{R}_R^g)^\dagger = U_S(hg^{-1})$ on $\HS_S$. This essentially amounts to a relabeling of the code states.

For non-ideal QRFs $R$, the reduction maps are generally not surjective. Their images $\HS_{S|R}^{ g}\subset \HS_S$ may therefore be proper subspaces. In that case, akin to a partial gauge-fixing, the reduction maps $\mathcal{R}_R^g$ do not fully extract the logical information, but leave some redundancy. As a result, the intersection between the subspaces $\HS_{S|R}^{g}$ and $\HS_{S|R}^{h}$ in the non-ideal case may range from $\{0\}$ to complete coincidence of the two subspaces.
Even so, the adjoint $(\mathcal{R}_R^g)^\dagger:\HS_{S|R}^{g}\to \HS_\phys$ may still be interpreted as an encoding map for the information contained in $\HS_{S|R}^{g}$.

\subsection{Correctable Gauge-Fixing Operators}
For the QRF setup from before, consider the gauge-fixing operators from \cref{eq:gauge-fixing}
\begin{equation}
    \mathcal{P}_R^g = \sqrt{N}(\ket{\phi(g)}\bra{\phi(g)}_R\otimes I_S).
\end{equation}
Lemma 5.2 in \cite{CCHM_2024} states that the gauge-fixing operators associated with an ideal QRF for a stabilizer code constitute a correctable error set, a fact also asserted for finite Abelian gauge groups.
This statement can be generalized to any system with a compact, possibly non-Abelian gauge group.

\begin{restatable}[Correctable gauge-fixing operators]{theorem}{corrgaugefix}
\label{thm:correctable_gauge-fixing}
    Let $G$ be a compact gauge group, acting on $\HS_R\otimes\HS_S$ as a product of unitary representations of $G$, i.e., $U_R(g)\otimes U_S(g)$, and let $\{\ket{\phi(g)}_R \, | \,  g\in G\}$ be the orientation states of the QRF $R$. Let $\tilde{G}$ be a set of group elements with orthogonal orientation states,  i.e., $\braket{\phi(\tilde{g})}{\phi(\tilde{h})}_R = \delta(\tilde{g},\tilde{h})\ \forall \tilde{g}, \tilde{h} \in \tilde{G}$. Then, the set of gauge-fixing operators associated with $\tilde{G}$,
    \begin{equation}
        \mathcal{E}_{\mathcal{P},\tilde{G}} \defeq \{\mathcal{P}_R^{\tilde{g}}  \, | \,  \tilde{g}\in\tilde{G}\},
    \end{equation}
    satisfies the Knill-Laflamme conditions, i.e.,
    \begin{equation}
        \Pi_\phys (\mathcal{P}_R^{\tilde{g}})^\dagger\mathcal{P}_R^{\tilde{h}} \Pi_\phys = \delta(\tilde{g}, \tilde{h})\Pi_\phys \quad \forall \tilde{g}, \tilde{h} \in \tilde{G}.
    \end{equation}
\end{restatable}

\begin{proof}
    Let $\tilde{g}, \tilde{h}\in\tilde{G}$. Then, the Knill-Laflamme conditions follow directly from the the fact that the reduction map is an isometry on $\HS_\phys$. Indeed,
    \begin{equation}
    \begin{aligned}
        \Pi_\phys (\mathcal{P}_R^{\tilde{g}})^\dagger\mathcal{P}_R^{\tilde{h}} \Pi_\phys &= N\delta(\tilde{g}, \tilde{h})\Pi_\phys(\ket{\phi(\tilde{g})}\bra{\phi(\tilde{g})}_R\otimes I_S) \Pi_\phys \\
        &=\delta(\tilde{g}, \tilde{h}) (\mathcal{R}_R^{\tilde{g}})^\dagger\mathcal{R}_R^{\tilde{g}} 
        = \delta(\tilde{g}, \tilde{h})\Pi_\phys.
    \end{aligned}
    \qedhere
    \end{equation}
\end{proof} 
If $G$ is continuous, the Knill-Laflamme conditions likewise become continuous in the error index, and the prefactor $\delta(\tilde{g}, \tilde{h})$ is a Dirac delta. Since the Knill-Laflamme conditions also hold for continuous error indices and infinite-dimensional spaces \cite{BKK_2007}, the result implies that there exists a recovery operation for the error set $\mathcal{E}_{\mathcal{P},\tilde{G}}$ on all of $\HS_\phys$.
If the orientation states render $R$ an ideal QRF, then the full set of gauge-fixing operators $\mathcal{E}_\mathcal{P} \defeq \mathcal{E}_{\mathcal{P}, G}$ is correctable.

\paragraph{Understanding gauge-fixing errors.}
The result in \cref{thm:correctable_gauge-fixing} is not surprising as choosing a QRF amounts to selecting redundant degrees of freedom. This provides a subsystem on which certain errors may act without disturbing the logical data. 
Importantly, these errors need not be formulated explicitly in terms of gauge-fixing operators. Since linear combinations of correctable errors remain correctable, gauge-fixing operators can be used to generate a broader class of equivalent error sets and only their linear span is relevant. From this perspective, gauge-fixing errors are simply one representative within a family of physically equivalent correctable errors.
This leads to the following general observation.

\begin{remark}
\label{rmk:diagonal_operators_in_orientation_states}
    If the set of correctable gauge-fixing projectors $\mathcal{E}_{\mathcal{P}, \tilde{G}}$ is associated with an orthonormal basis of orientation states $\ket{\phi(\tilde{g})}_R$, then any operator which is diagonal in this basis is correctable.
\end{remark}

There are also scenarios in which gauge-fixing errors themselves are physically meaningful. Consider again the case where $\{\ket{\phi(\tilde{g})}_R  \, | \,  \tilde{g}\in\tilde{G}\}$ is an orthonormal basis. Then, a measurement in this basis acts on a state $\rho$ on $\HS_\phys$ as
\begin{equation}
    \rho \mapsto \int_{\tilde{G}} \D\tilde{g} \ \Pi_{\tilde{g}}\rho\Pi_{\tilde{g}},
\end{equation}
where the projectors $\Pi_{\tilde{g}} = \ket{\phi(\tilde{g})}\bra{\phi(\tilde{g})}_R$ are simply rescaled gauge-fixing operators. Since these are correctable, one can recover the original state after measuring $\HS_R$ in this basis.

\paragraph{Error correction via charge measurements for Abelian gauge groups.}

In \cite{CCHM_2024}, gauge-fixing errors were discussed as being dual (via Pontryagin duality) to standard Pauli errors for qubit stabilizer codes. This duality only holds for Abelian groups, and we do not attempt a full classification of the general case here. Nevertheless, an important take-away from the discussion is that correctable sets of gauge-fixing errors for stabilizer codes have equivalent dual Pauli error sets. A similar idea leads us to a connection between the correctable gauge-fixing operators and the recovery scheme via charge measurements outlined in~\cref{ssec:QEC_and_gauge}.

Consider thus a perspective-neutral setup $\HS_\kin = \HS_R\otimes\HS_S$ with a compact Abelian gauge group $G$. Then, we can decompose the kinematical space into isotypes $\HS_{\bm{q}}$ belonging to the irreducible representations $\chi_{\bm{q}}$ (the irreducible characters of $G$) labeled by their charges ${\bm{q}}$. 
We thus find $\HS_\kin = \bigoplus_{\bm{q}} \HS_{\bm{q}}$. The projectors onto the isotypes $\HS_{\bm{q}}$ are given by \cite{Fulton:2004uyc}
\begin{equation}
    \Pi_{\bm{q}} = \frac{1}{|G|}\int \D g \ \overline{\chi}_{\bm{q}}(g)U_{RS}(g),
\end{equation}
where the bar denotes complex conjugation.

\begin{restatable}{proposition}{corrAq}
\label{prop:A_q_from_gf}
    Let G be a compact Abelian gauge group, and let $R$ be an ideal QRF for $G$. Then, the correctable set of gauge-fixing operators $\mathcal{E}_{\mathcal{P}} =\{\mathcal{P}_R^g\}_{g\in G}$ yields a set of operators $\{A_{\bm{q}}\}_{\bm{q}}$ which are unitary on $\HS_\kin$ and such that $A_{\bm{q}}( \HS_\phys) \subset \HS_{\bm{q}}$, where
    \begin{equation}
        A_{\bm{q}} = \frac{1}{\sqrt{|G|}}\int\D g \ \overline{\chi}_{\bm{q}}(g)\mathcal{P}_R^g.
    \end{equation}
    The operators $\{A_{\bm{q}}\}_{\bm{q}}$ are again a correctable error set.
\end{restatable}
\begin{proof}
      The proof is a straight-forward calculation and can be found in \cref{app:proofs}.
\end{proof}

The operators $A_{\bm{q}}$ make the recovery  procedure described in \cref{ssec:QEC_and_gauge} possible. If the charge ${\bm{q}}$ can be measured, for instance, by a syndrome measurement in stabilizer codes or, as we will see later, by measuring the gauge constraint operators, then upon obtaining the outcome ${\bm{q}}$, we can apply $A_{\bm{q}}^\dagger$:
\begin{equation}
    \rho_{\rm err}\mapsto \sum_{\bm{q}} A_{\bm{q}}^\dagger \Pi_{\bm{q}} \rho_{\rm err}\Pi_{\bm{q}}A_{\bm{q}}.
\end{equation}
This defines a recovery channel for the error set $\{A_{\bm{q}}\}_{\bm{q}}$, and if $G$ is a stabilizer group, it recovers Eq. (162) from \cite{CCHM_2024}. Moreover, by the character orthogonality relations, we can invert \cref{eq:A_q_def} to $\mathcal{P}_R^g = \frac{1}{\sqrt{|G|}}\sum_{\bm{q}}\chi_{\bm{q}}(g)A_{\bm{q}}$. Therefore, the gauge-fixing operators are linear combinations of $A_{\bm{q}}$ and similarly correctable by this recovery.

\paragraph{Non-ideal QRFs and coarse-grained recovery.}
In certain cases we can generalize \cref{prop:A_q_from_gf} to non-ideal frames.  More precisely, if a non-ideal QRF becomes ideal upon restriction to a subgroup $H\subset G$, then the corresponding orthogonal orientation states not only give correctable gauge-fixing operators by \cref{thm:correctable_gauge-fixing}, but also determine the associated recovery via coarse-grained charge-sector measurements. Thus, whenever an orthogonal subgroup-orbit exists, non-ideal frames select both an exact correctable operator family and the appropriate coarse-grained charge-sector measurement.

In the setup of \cref{thm:correctable_gauge-fixing}, suppose thus that $\tilde{G}=H$ is a subgroup of the compact Abelian group $G$, associated with orientation states which form an {\em orthonormal basis} of $\HS_R$,
\begin{equation}
    \operatorname{span}\big\{\ket{\phi(h)}_R\, | \ h\in H\big\} = \HS_R, \quad \braket{\phi(h)}{\phi(h')} = \delta(h, h') \ \forall h, h'\in H.
\end{equation}
Then, $R$ becomes an ideal QRF for the subgroup $H$ and the restricted representation $\tilde{U}_{RS}=U_{RS}|_H$. 

We denote the characters of $H$ by $\chi_{\bm{r}}$, and $\HS_\kin$ decomposes into the direct sum of their isotypes $\HS_\kin = \bigoplus_{\bm{r}}\HS_{\bm{r}}$. Note that these isotypes are related to the ones of $G$ by
\begin{equation}
\label{eq:branching_rule}
    \HS_{\bm{r}} = \bigoplus_{\bm{q}\,:\,\chi_{\bm{q}}|_{H}=\chi_{\bm{r}}}\HS_{\bm{q}},
\end{equation}
in which the direct sum runs over all $\bm{q}$ such that the restriction of the $\chi_{\bm{q}}$ to $H$ equals $\chi_{\bm{r}}$. In other words, the $H$-charge sectors are obtained by coarse-graining the original $G$-charge sectors according to the restriction of characters to $H$ and contain precisely the states which transform as $U_{R}(h)\ket{\psi}_R = \chi_{\bm{r}}(h)\ket{\psi}_R$.
Applying \cref{prop:A_q_from_gf} to this situation, we find that\footnote{This expression is written as an integral over the Haar measure of $H$. If $H$ is discrete, we can replace the integral with a sum.}
\begin{equation}
    A_{\bm{r}} = \frac{1}{\sqrt{|H|}}\int_{H}\D h \ \overline{\chi}_{\bm{r}}(h)\mathcal{P}_R^{h}
\end{equation}
maps the zero-charge sector (with respect to $H$) into $\HS_{\bm{r}}$ and is unitary on $\HS_\kin$.  Since every $G$-invariant state is also $H$-invariant, the zero-charge sector of $H$ also contains $\HS_\phys$, and we get
\begin{equation}
    A_{\bm{r}}(\HS_\phys) \subset \bigoplus_{\bm{q}\,:\,\chi_{\bm{q}}|_{H}=\chi_{\bm{r}}}\HS_{\bm{q}}.
\end{equation}
Therefore, we can state the following.
\begin{proposition}
\label{prop:non_ideal_A_r}

    If $R$ is a QRF for $G$ and $H\subset G$ is a subgroup whose orientation states are an orthonormal basis of $\HS_{R}$, then the set of operators $\{A_{\bm{r}}\}_{\bm{r}}$ is correctable and can be used to implement a recovery operation via coarse-grained charge measurements,
    \begin{equation}
        \rho_{\rm err}\mapsto \sum_{\bm{r}} A_{\bm{r}}^\dagger \Pi_{\bm{r}} \rho_{\rm err}\Pi_{\bm{r}}A_{\bm{r}}.
    \end{equation}
    Here, the coarse-grained charge measurements project the state according to $\Pi_{\bm{r}} = \sum_{\bm{q}\,:\,\chi_{\bm{q}}|_{H}=\chi_{\bm{r}}}\Pi_{\bm{q}}$.
\end{proposition}
The same construction applies if the orthonormal family of orientation states is associated not with a subgroup itself but with a coset $\tilde{G} = g{H}$ of a subgroup  ${H}$ of $G$. In that case, the orientation states can be redefined with a new seed state $\ket{\phi'(e)}_R = \ket{\phi(g)}_R$ to
\begin{equation}
    \ket{\phi'(h)}_R = \ket{\phi(h{g})}_R \quad \forall h\in H,
\end{equation}
leading to the redefined gauge-fixing operators $\mathcal{P}_R'^h = \mathcal{P}_R^{h{g}}$.
The previous discussion then yields the correctable operators
\begin{equation}
    A_{\bm{r}} = \frac{1}{\sqrt{|H|}}\int_H \D h \ \overline{\chi}_{\bm{r}}(h)\mathcal{P}_R^{hg}.
\end{equation}

\section{Quantum Reference Frames in Lattice Quantum Electrodynamics}
\label{sec:QRFs_in_LQED}

Quantum electrodynamics, which describes the interactions between light and matter, is the simplest gauge theory of the Standard Model due to the Abelian structure of its gauge group $\operatorname{U}(1)$. The gauge field is the electromagnetic four-potential $A_\mu$, which, in absence of matter, is associated with the Lagrangian density
\begin{equation}
    \mathcal{L} = -\frac{1}{4}\ F_{\mu\nu}F^{\mu\nu},
\end{equation}
where $F_{\mu\nu} = \partial_\mu A_\nu - \partial_\nu A_\mu$ is the field strength tensor. 

Passing to the Hamiltonian formulation, the canonical momenta conjugate to $A_\mu$
are $P^\mu = \partial \mathcal{L}/\partial( \partial_0A_\mu)=F^{\mu 0}$. Since $F^{\mu\nu}$ is antisymmetric, one has $P^0=F^{00}=0$, which is the primary constraint of the theory. The spatial component of the canonical momentum is the electric field $P^i=F^{i0}= E^i$. Preservation in time of the primary constraint yields the secondary constraint $\partial_i P^i = 0$, which is Gauss's law $\vec{\nabla} \cdot {\vec{E}} = 0$ in the absence of charges.
 
At this stage, one may impose a gauge condition. For example, in \textit{temporal gauge}, the Hamiltonian takes the form
\begin{equation}
    H = \frac{1}{2}\int \D^3 \vec{x} \ \big(\vec{E}^2 (\vec{x}) + \vec{B}^2(\vec{x}) \big),
\label{eq:EM_hamiltonian}
\end{equation}
where the magnetic field is $\vec{B} = \vec{\nabla} \times \vec{A}$. Another common choice is \emph{Coulomb gauge}, with $\partial_i A^i = 0$.
To identify QRFs and interpret QED as a QECC, we now consider the lattice version of the theory in its Hamiltonian formulation \cite{KS1975,Kogut1979}.
After introducing lattice QED in \cref{ssec:LQED}, we show in \cref{ssec:spanning_tree_QRFs} that the pure-gauge sector admits complete ideal QRFs in the form of spanning trees. \Cref{ssec:QED_matter_QRF} then explains how the inclusion of charged matter changes this picture and leads to a new, non-ideal QRF built from fermionic matter.

\subsection{Introduction to Lattice QED}
\label{ssec:LQED}

We consider QED in temporal gauge. In its lattice formulation in $3+1$ dimensions, time remains continuous while the three spatial dimensions are discretized.

Specifically, let us consider a cubic lattice in three dimensions $\Gamma \subset (a\mathbb{Z})^3$ with lattice spacing $a$, vertex set $\mathcal{V}$ and links $\mathcal{L}$. The vertices (or sites) $v = (v_1, v_2, v_3)\in\mathcal{V}$ are located at positions\footnote{To distinguish the continuum from the lattice, we use arrows to indicate vectors in the former.} $x_v =(av_1, av_2, av_3)$. 
The links $l = [v, v+e_i]\in\mathcal{L}$ are positively oriented along the axes $e_1, e_2$ and $e_3$\footnote{The specific choice of orientation serves only as a convention and has no physical significance.}.  
The lattice can have different types of boundary conditions. For instance, it can be infinite, and hence translation-invariant; finite with periodic boundary conditions; or finite with either amputated or dangling boundary links. The latter two are referred to as {\it smooth boundaries} (for amputated links) or {\it rough boundaries}  (for dangling links). In what follows, we focus on the infinite lattice, the finite lattice with periodic boundary conditions, and the finite lattice with smooth boundaries. 

Each link is assigned a copy of  $L^2(\operatorname{U}(1))$, which encodes the degrees of freedom of the gauge field. In other words, each link corresponds to a planar quantum rotor with $\HS_\rot = L^2(\operatorname{U}(1))$. We discuss quantum rotors and their ties to error correction in \cref{sssec:rotor_codes}.
The total kinematical Hilbert space is 
\begin{equation}
    \HS_\text{kin} = \bigotimes_{l\in\mathcal{L}} L^2 \left(\operatorname{U}(1)\right) = \bigotimes_{l\in\mathcal{L}} \HS_\rot.
\end{equation}
On each link, the Hilbert space is spanned by the formal basis $\{\ket{e^{i\theta}}_l  \, | \,  \theta \in [0, 2\pi) \}$, normalized to $\braket{e^{i\theta}}{e^{i\phi}}_l = \delta(\theta-\phi)$. By a Fourier transform, we obtain the {\em electric flux basis} $\{\ket{k}_l \}_{k\in\mathbb{Z}}$, where 
\begin{equation}
    \ket{k}_l = \int_0^{2\pi} \frac{\D\theta}{\sqrt{2\pi}}e^{ik\theta}\ket{e^{i\theta}}_l.
\end{equation}
The conjugate position and momentum operators on $\HS_\rot$ are $U_l = \int_0^{2\pi}\D \theta \ e^{i\theta} \ket{e^{i\theta}}\bra{e^{i\theta}}_l$ (in the context of lattice QED, referred to as the {\it link operator}) and $\epsilon_l = \sum_{k\in\mathbb{Z}}k\ket{k}\bra{k}_l$ (the {\it electric field operator}). The link operator shifts the electric flux on a link, $U_l^m\ket{k}_l = \ket{k+m}_l$. In \cref{sssec:rotor_codes}, we define the angle shift operator $X_l(\lambda) = e^{-i\lambda\epsilon_l}$, acting as $X_l(\lambda)\ket{e^{i\theta}}_l = \ket{e^{i(\theta+\lambda)}}_l$, and find the braiding relation $U_lX_l(\lambda) = e^{i\lambda}X_l(\lambda)U_l$.

\paragraph{Plaquette operators and the Kogut-Susskind Hamiltonian.} To set up the Hamiltonian of lattice QED, we first need to introduce plaquettes in the cubic lattice\footnote{This construction works for any spatial dimension $D \geq 2$.}. We take these to be the smallest loops (oriented counter-clockwise with respect to $e_1, e_2, e_3$), consisting of the links 
\begin{equation}
\centertikz{
    \node (v) at (0,0) {\small$v$};
    \node (vi) at (2,0) {\small$v+e_i$};
    \node (vij) at (2,2) {\small$v+e_i+e_j$};
    \node (vj) at (0,2) {\small$v+e_j$};
    \draw[->] (v)--node[above] {\small$l_1$} (vi);
    \draw[->] (vi)--node[left] {\small$l_2$}(vij);
    \draw[->] (vij)--node[below] {\small$l_3$}(vj);
    \draw[->] (vj)--node[right] {\small$l_4$}(v);
    
}
    \label{eq:plaquette}
\end{equation} where $v$ is a vertex and $i \neq j$. We then write $p_{ij}(v) = (l_1, l_2, l_3, l_4)$ for this plaquette and $\mathcal{P}$ for the set of all such plaquettes in the lattice.

Taking products of the link operators along a plaquette with inverses whenever passing through a link in the negative direction results in the {\it plaquette operators}. For $p_{ij}(v)$ such that $l_3$ and $l_4$ are passed negatively, we have
\begin{equation}
    U_{ij}(v) \defeq U_{l_1}U_{l_2}U_{l_3}^\dagger U_{l_4}^\dagger = e^{i(\Theta_{l_1}+\Theta_{l_2}-\Theta_{l_3}-\Theta_{l_4})} \eqdef e^{i\Theta_{ij}(v)}
\end{equation}
where $\Theta_{ij}(v) \defeq \Theta_{l_1} + \Theta_{l_2} - \Theta_{l_3} - \Theta_{l_4} $ is the discrete curl.

The Kogut-Susskind Hamiltonian of the pure gauge system is given by
\begin{equation}
    H = \frac{g^2}{2a}\sum_{l\in\mathcal{L}}\epsilon_l^2 - \frac{1}{ag^2}\sum_{p_{ij}(v)\in\mathcal{P}} (U_{ij}(v) + U_{ij}(v)^\dagger)/2 = \frac{g^2}{2a}\sum_{l\in\mathcal{L}}\epsilon_l^2 - \frac{1}{ag^2}\sum_{p_{ij}(v)\in\mathcal{P}} \cos\Theta_{ij}(v),
    \label{eq:KS_ham_pg}
\end{equation}
where $g$ is the coupling constant. In \cref{app:continuum_limit}, we discuss the continuum limit of this expression and relate the operators $\epsilon_l$ to the electric field, $\Theta_l$ to the gauge field, and its discrete curl to the magnetic field.

\paragraph{Gauge transformations and Gauss's law.}
Let us now include a charged matter field $\psi(x)$ into the system. On the lattice, this is discretized to $\psi_v$ on the vertices $v$ as we will see in more detail below. 
In the continuum case, a gauge transformation corresponds to a transformation of the charged matter fields as $\psi(x)\mapsto e^{-i\lambda(x)}\psi(x)$ and of the gauge fields as $A_i(x)\mapsto A_i(x)+\frac{1}{g}\partial_i\lambda(x)$ for some function $\lambda(x)$. As shown in \cref{app:gauss_law}, the discrete analogue of a gauge transformation is
\begin{equation}
\label{eq:discrete_gauge_transformation_analogue}
    \psi_v\mapsto e^{-i\lambda_v}\psi_v \quad\text{and}\quad
    \Theta_{[v,v+e_i]} \mapsto  \Theta_{[v,v+e_i]} + \Delta_{v, v+e_i}\lambda,
\end{equation}
where the the function $\lambda(x)$ is discretized to $\lambda_v$ on $v\in \mathcal{V}$ and $\Delta_{v, v+e_i}\lambda \defeq \lambda_{v+e_i}-\lambda_{v}$. 
This local gauge transformation at $v$ is generated by
\begin{equation}
    \mathcal{C}_v = \sum_{i=1}^3(\epsilon_{[v, v+e_i]}-\epsilon_{[v-e_i, v]}) - \rho_v
    \label{eq:gauss_law_constraint}
\end{equation}
and acts on states as $e^{i\lambda_v\mathcal{C}_v}$ (see \cref{app:gauss_law}). Here, $\rho_v$ is a tentative density operator which generates phase rotations of the matter, which we introduce in the next paragraph.
Since $e^{i\lambda_v\mathcal{C}_v}$ leaves physical states invariant, those states satisfy the constraint $\mathcal{C}_v\ket{\psi}_\phys = 0$, recognizable as a discrete version of Gauss' law $\nabla\cdot E - \rho = 0$.
In the pure gauge theory without matter, the constraints are simply modified to $\mathcal{C}_v = \sum_{i=1}^3(\epsilon_{[v, v+e_i]}-\epsilon_{[v-e_i, v]})$ and represent $\nabla\cdot E = 0$.
The group of gauge transformations is
\begin{equation}
    \mathcal{G} = \left\{\prod\nolimits_{v\in\mathcal{V}} e^{i\lambda_v\mathcal{C}_v}  \, \middle| \,  \lambda_v \in [0, 2\pi) \ \forall v\in\mathcal{V}\right\}.
\end{equation}
The physical states are obtained from the kinematical Hilbert space via the projector
\begin{equation}
    \Pi_{\phys} = \prod_{v\in\mathcal{V}}\left(\int_0^{2\pi}\frac{d\lambda_v}{2\pi} e^{i\lambda_v\mathcal{C}_v} \right), \quad \Pi_{\phys}(\HS_{\text{kin}}) = \HS_\phys.
    \label{eq:gauss_projector}
\end{equation}
Note that the pure gauge Hamiltonian \labelcref{eq:KS_ham_pg} is gauge invariant as each individual component commutes with the constraints $\mathcal{C}_v$.

\paragraph{Fermionic matter: discretizing the Dirac field.} Let us now introduce a discretized version of the Dirac field, which, in QED, describes both electrons and positrons and couples to the electromagnetic field. 
Technical issues, such as fermion doubling \cite{Susskind_1997}, arise when naively discretizing the Dirac equation to a lattice. Different discretizations have been proposed to address these issues \cite{KARSTEN1981}. Here, we present and make use of the Kogut-Susskind prescription of {\it staggered} fermionic fields \cite{KS1975, Kogut1979}.

With each vertex, we associate a two-dimensional Hilbert space
$\HS_{\text{matter}} = \bigotimes_{v\in\mathcal{V}}\mathbb{C}^2$ equipped with fermionic creation and annihilation operators $\psi_v^\dagger$ and $\psi_v$ satisfying the anti-commutation relations
\begin{equation}
    \{\psi_v, \psi_{v'}^\dagger\} = \delta_{v, v'},\quad \{\psi_v, \psi_{v'}\} = \{\psi_v^\dagger, \psi_{v'}^\dagger\}=  0.
\end{equation}
We consider the number basis spanned by the eigenstates of $\psi_v^\dagger\psi_v$, $\ket{0}_v$ and $\ket{1}_v$, labeled by their eigenvalue. 
The joint description of particles and anti-particles in the staggered formulation is obtained by distinguishing between even and odd sites ($v=(v_1,v_2,v_3)$ is even if $(-1)^{|v|} = (-1)^{v_1+v_2+v_3} = 1$ and odd if $(-1)^{|v|} = -1$). On even sites, the state $\ket{1}_v$ is interpreted as a positively charged fermion while on odd sites, the state $\ket{0}_v$ represents a negatively charged anti-fermion. The charge density at a vertex $v$ is then defined accordingly,
\begin{equation}
    \rho_v = \psi_v^\dagger\psi_v - j_v = \begin{cases}
        \psi_v^\dagger\psi_v\quad v \text{\ even}\\
        \psi_v^\dagger\psi_v - 1\quad v \text{\ odd}
    \end{cases}.
\end{equation}
Here, $j_v = \frac{1}{2}(1-(-1)^{|v|})$ is the parity indicator function.

\paragraph{Hamiltonian of lattice QED with fermionic matter.}
Here we provide the total Hamiltonian of interacting lattice QED with fermionic matter following Kogut and Susskind \cite{KS1975}. Firstly, the Hamiltonian contains a gauge invariant mass term for the matter field, called staggered mass,
\begin{equation}
    \sum_{v\in\mathcal{V}}(-1)^{|v|}m\psi^\dagger_v\psi_v.
\end{equation}
On odd sites, an occupied state $\ket{1}_v$ (i.e., no anti-fermions) contributes $-m$ to the sum while a hole $\ket{0}_v$ (i.e., an anti-fermion) contributes $0$. 
Then, we introduce a hopping term describing the creation and annihilation of particle - anti-particle pairs. This is realized by the gauge-invariant operators $\psi_v^\dagger U_{[v,v+e_i]}\psi_{v+e_i}$ (or their hermitian conjugate) which conserve the global charge. Notice that to satisfy Gauss's law after this change of charge density on $v$ and $v+e_i$, the operator $U_l$ is needed to adjust the electric flux on $[v, v+e_i]$. This leads to the hopping term
\begin{equation}
    \frac{i}{2a}\sum_{v\in\mathcal{V}, \ i} \eta_{v,i}(\psi_v^\dagger U_{[v, v+e_i]}\psi_{v+e_i} - \psi_{v}^\dagger U_{[v-e_i, v]}^\dagger\psi_{v-e_i}),
    \label{eq:hopping_term}
\end{equation}
with the staggered sign factor $\eta_{v, i} = (-1)^{\sum_{j<i}v_j}$.

Finally, the full Kogut-Susskind Hamiltonian reads
\begin{equation}
    \begin{aligned}
        H_{m} = &\sum_{v\in\mathcal{V}}(-1)^{|v|}m\psi^\dagger_v\psi_v +\frac{i}{2a}\sum_{v\in\mathcal{V}, \ i} \eta_{v,i}(\psi_v^\dagger U_{[v, v+e_i]}\psi_{v+e_i} - \psi_{v}^\dagger U_{[v-e_i, v]}^\dagger\psi_{v-e_i})  \\
        & + \frac{g^2}{2a}\sum_{l\in\mathcal{L}}\epsilon_l^2 - \frac{1}{ag^2}\sum_{p_{ij}(v)\in\mathcal{P}} (U_{ij}(v) + U_{ij}(v)^\dagger)/2.
    \end{aligned}
    \label{eq:full_hamiltonian}
\end{equation}
One can show that this Hamiltonian does indeed yield a discrete version of the Klein-Gordon equation for the staggered field $\psi_v$ \cite{CPS_2025}.

\paragraph{Wilson loops and holonomies.}
The gauge-invariant hopping operators can be generalized across multiple links. In this case, the operators $\psi_v$ and $\psi_{v'}^\dagger$ need to be connected by a string of link operators along a path $\gamma$ connecting $v$ and $v'$. This string of link operators is referred to as a {\it Wilson line}, and (if $\gamma$ has links $\{l_1, l_2, \dots, l_n\}$) is given by
\begin{equation}
    W_\gamma = \prod_{l\in\gamma}U_l^{\sigma_{l}},
\label{eq:wilson_line}
\end{equation}
where $\sigma_l$ is the sign of the direction in which $\gamma$ passes $l$. As discussed in \cref{app:gauss_law}, this operator implements the parallel transport along $\gamma$ and adjusts the phase of a charge when transported from $v$ to $v'$ via $\psi_{v'}^\dagger W_\gamma \psi_v$.
The curvature of the gauge field is measured by Wilson lines on closed loops, which we refer to as {\it holonomy operators} or {\it Wilson loops}\footnote{In higher-dimensional cases, Wilson loops are also traced over; for $\operatorname{U}(1)$, this does not matter.}. We denote the holonomy operator around a closed loop $\gamma$ by
\begin{equation}
    H_\gamma = \prod_{l\in\gamma}U_l^{\sigma_{l}}.
    \label{eq:holonomy}
\end{equation}
We refer to the eigenvalues $h_\gamma$ of these operators as {\it holonomies}.

\subsection{QRFs for Pure Gauge Lattice QED}
\label{ssec:spanning_tree_QRFs}

In this section, we first focus on the pure gauge sector without the charged matter field.
We will denote a gauge transformation by the angles $\bm{\lambda} = \{\lambda_v\}_{v\in\mathcal{V}}\in[0, 2\pi)^{\times|\mathcal{V}|}$ as
\begin{equation}
    G(\bm{\lambda}) \defeq \prod_{v\in\mathcal{V}}e^{i\lambda_v\mathcal{C}_v} = \bigotimes_{l=[v, v']\in\mathcal{L}}X_l(\lambda_{v'}-\lambda_v) \in\mathcal{G},
\end{equation}
where  $X_l$ is the angle shift operator $X_l(\lambda) = e^{-i\lambda\epsilon_l}$.
To construct complete, ideal reference frames for $\mathcal{G}$, we need to find a set of links $R$ such that we can define appropriate orientation states in $\HS_R$ to parametrize the action of $\mathcal{G}$. In this way, $R$ will serve as a reference frame for the remaining links $S$.
The action splits into a product $G(\bm{\lambda})=G_R(\bm{\lambda})\otimes G_S(\bm{\lambda})$ of factors on $R$ and on $S$ in the obvious way, i.e.,
\begin{equation}
    G_R(\bm{\lambda}) = \bigotimes_{l = [v, v']\in R} X_l(\lambda_{v'}-\lambda_v),
\end{equation}
and similarly for $G_S(\bm{\lambda})$. The constraints also split up into a sum of terms with support on $R$ and $S$, $\mathcal{C}_{v, R} = \sum_{l_{\text{out}}\in R}\epsilon_{l_{\text{out}}}-\sum_{l_{\text{in}}\in R}\epsilon_{l_{\text{in}}}$ and similarly for $\mathcal{C}_{v, S}$.
For simplicity, we will also denote the sum of constraints associated with a vertex set $V\subset\mathcal{V}$ as $\mathcal{C}_V\defeq \sum_{v\in V}\mathcal{C}_v$,
and we write $\mathcal{C}_{V, R} = \sum_{v\in V}\mathcal{C}_{v, R}$ for the restriction of this sum to $R$.

In the pure gauge theory, not all constraints $\mathcal{C}_v$ are independent. Indeed,
\begin{equation}  \sum_{v\in\mathcal{V}}\mathcal{C}_v = 0
    \label{eq:sum_to_zero}
\end{equation}
holds when the lattice $\Gamma$ is either infinite, or finite with smooth or periodic boundary conditions, in which case there are $|\mathcal{V} |-1$ independent constraints. This means that a global gauge transformation (by the same angle on every site) is trivial. 
In other words, if $\lambda_v  - \eta_v = \alpha$ for all $v$, then $G(\bm{\lambda}) = G(\bm{\eta})$.

If $\HS_R$ admits an orthonormal basis such that every non-trivial transformation $G_R(\bm{\lambda})\neq I_R$ maps each basis element to a different one, then $R$, equipped with  this basis as orientation states, is a complete and ideal reference frame.  
Intuitively speaking, the basis elements keep track of the orientation induced by the gauge transformations and, since they are orthogonal, these orientations are perfectly distinguishable. 
As captured by the following theorem, such sets of links arise from spanning trees of the lattice. These are subgraphs of the lattice which span all of the vertices but contain no loops.

The following theorem is valid for the types of lattice $\Gamma$ for which~\cref{eq:sum_to_zero} holds.

\begin{restatable}[Spanning tree QRFs]{theorem}{spanningQRF}
\label{thm:LGT_QRF_pure}
    The gauge field spaces on the links of a spanning tree $T$ of $\Gamma$ form a QRF $R$ for $\mathcal{G} = \left\{G(\bm{\lambda})  \, | \,  \bm{\lambda}\in[0, 2\pi)^{\times |\mathcal{V}|}\right\}$ with the following properties:
    \begin{enumerate}
        \item $\HS_\kin = \HS_R\otimes\HS_S$, where $\HS_R = \bigotimes_{l\in R}\HS_\rot$ (similarly for $\HS_S$) and $\mathcal{G}$ factorizes to $G(\bm{\lambda}) = G_R(\bm{\lambda}) \otimes G_S(\bm{\lambda})$.
        \item There exists a set of constraints $\{\mathcal{C}_{V_l} \ | \ l \in R\}$ such that $\mathcal{C}_{V_l, R} = -\epsilon_l$ through which we can uniquely parametrize gauge transformations by $\bm{\lambda} = \{\lambda_l\}_{l\in R}$,
        \begin{equation}
            G'(\bm{\lambda}) \defeq \prod_{l\in R}e^{i\lambda_l\mathcal{C}_{V_l}}, \quad G'_R(\bm{\lambda})=\bigotimes_{l\in R}X_l(\lambda_l).
        \end{equation}
        \item The orientation states $\ket{\phi(\bm{\lambda})}_{R} = G'_{R}(\bm{\lambda})\left(\bigotimes_{l\in R}\ket{e^{i(\theta = 0)}}_l\right)$ provide a formal orthonormal basis of $\HS_R$ which transforms covariantly under gauge transformations, 
        \begin{equation}
            G'_{R}(\bm{\eta})\ket{\phi(\bm{\lambda})}_{R} = \ket{\phi(\bm{\lambda} + \bm{\eta})}_{R}, \quad \braket{\phi(\bm{\eta})}{\phi(\bm{\lambda})}_{R} = \prod_{l\in R}\delta(\eta_l -\lambda_l),
        \end{equation}
        making $R$ an ideal QRF.
    \end{enumerate}
\end{restatable}

\begin{proof}
See~\cref{app:proofs}.
\end{proof}

The completeness of the spanning-tree QRF is tied to the special property \cref{eq:sum_to_zero} of the  pure gauge  Gauss constraints on lattices without boundary flux contributions of dangling links. Then, only $|\mathcal{V}|-1$ independent gauge parameters need to be fixed. Since a spanning tree contains exactly $|\mathcal{V}|-1$ links, its link degrees of freedom suffice to resolve every non-trivial gauge transformation, making the associated QRF complete.

On a finite lattice, the operators $\mathcal{P}_R^{\bm{\lambda}} = (2\pi)^{\frac{|\mathcal{V}|-1}{2}}(\ket{\phi(\bm{\lambda})}\bra{\phi(\bm{\lambda})}_R \otimes I_S)$ constitute a complete gauge-fixing and remove any redundancy.
Gauge fixing via spanning trees (often referred to as maximal trees in this context) is common practice in the field of lattice gauge theory without explicitly invoking QRFs. A few examples are \cite{Creutz_1977,LWB_2000, bauer_2023, Mariani_2024, MO_2018, Banburski:2014cwa}, and a textbook introduction can be found in \cite{GL_2010_textbook}.
\Cref{thm:LGT_QRF_pure} illustrates this process by explicitly connecting spanning trees to QRFs for the gauge group. 
A similar construction of spanning trees as perspective-neutral QRFs was recently presented in \cite{AHS_25}, where the QRFs were used to calculate relational entanglement entropies.

\paragraph{The holonomy basis of $\HS_{\phys}$.}
The spanning tree QRFs, constructed in \cref{thm:LGT_QRF_pure}, define a split between the redundant and physical degrees of freedom on the lattice. 
Using the associated Page–Wootters reduction maps, we can extract the physical degrees of freedom encoded in $\HS_{\phys}$ explicitly.
The spanning tree QRFs are ideal, and gauge-fixing their orientation removes all the redundancy. The reduction maps are thus unitaries onto $\HS_S$,
\begin{equation}
\mathcal{R}^{\bm{\lambda}}_R:\HS_{\phys}\rightarrow  \HS_{S|R}^{\bm{\lambda}} =\HS_S,
\end{equation}
where $S = \mathcal{L}\setminus R$. On a finite $N_1\times N_2 \times N_3$ lattice, a spanning tree contains $N_1N_2N_3-1$ links and thus, $\HS_S$ consists of $K \defeq \mathcal{|L|} - N_1N_2N_3+1$ rotor spaces $\HS_\rot$. The system therefore encodes $K$ rotors, each assigned to a link in $S$ by $\mathcal{R}^{\bm{\lambda}}_R$.

This admits a simple geometric interpretation: Since $R$ is a spanning tree, every link $l \in S$ closes a unique loop $\gamma_l$ when added to the tree. Each non-tree link therefore labels one independent holonomy, and these links hence naturally parametrize the physical information that remains after gauge fixing on $R$.
We refer to these independent holonomies as the {\it fundamental holonomies}, associated with a {\it holonomy basis}, presented in \cite{Mariani_2024} for lattice gauge theories with finite groups.
The fundamental holonomy associated with $l = [v, v']\in S$ is defined as the holonomy on the loop $\gamma_l$, obtained by connecting $v'$ to $v$ on the spanning tree and closing the loop along $l$. We write $H_l \defeq H_{\gamma_l}$ for this operator. 
Any other holonomy operator can be written uniquely as a product of the fundamental holonomies: If $H_\gamma$ passes a set of links $L_\gamma\subset S$, each with orientation $\sigma_l = \pm 1$, then
\begin{equation}
    H_\gamma = \prod_{l\in L_\gamma} H_l^{\sigma_l}.
\end{equation}
The eigenstates of the fundamental holonomies form the holonomy basis of $\HS_\phys$.

\begin{restatable}[Holonomy basis]{proposition}{holbasis}
\label{prop:holonomy_basis}
    Consider the inverse reduction map\footnote{We choose to denote $\ket{e^{i\theta}}$ for $\theta = 0$ as $\ket{e^{i 0 }}$ to avoid confusion with the state $\ket{1}\equiv\ket{k=1}$ in the electric flux basis.}
    \begin{equation}
        (\mathcal{R}_R^{\bm{0}})^\dagger = (2\pi)^\frac{|\mathcal{V}|-1}{2}\Pi_\phys (\ket{e^{i0}, e^{i0}, \dots, e^{i0}}_R\otimes I_S) : \HS_S\rightarrow \HS_{\phys}.
    \end{equation}
    The states $(\mathcal{R}_R^{\bm{0}})^\dagger \bigotimes_{l\in S}\ket{e^{i\theta_l}}_l \in \HS_{\phys}$ are eigenstates of the fundamental holonomies with eigenvalue $h_l = e^{i\theta_l}$ for $H_l$. Accordingly labeled, the states $\ket{\{h_l\}} = (\mathcal{R}_R^{\bm{0}})^\dagger \bigotimes_{l\in S}\ket{e^{i\theta_l}}_l$ form the orthonormal \textit{holonomy basis}
    \begin{equation}
        \left\{\ket{\{h_l\}} \, | \, \forall l\in S:  h_l\in\operatorname{U}(1) \right\}
    \end{equation}
    of $\HS_{\phys}$.
\end{restatable}

\begin{proof}
    See \cref{app:proofs}.
\end{proof}

The role of the spanning-tree QRF can thus be understood as follows. 
Fixing the gauge on the spanning tree links fixes their contribution to the holonomies, and the only remaining dynamical degrees of freedom reside on the links outside the tree. These unfixed links thus fully determine the holonomies and carry all of the physical information. Since we can associate a fundamental holonomy to each of them, the holonomy basis yields an explicit parametrization of the encoded rotors. The map $(\mathcal{R}^{\bm{0}}_R)^\dagger$ concretely realizes this parametrization.

Associated with each of the $K$ encoded rotors for $l\in S$, there are $X_l(\lambda)$ and $U_l^m$ operators on $\HS_S$. These are encoded as gauge-invariant operators on the full lattice whose form we can see from \cref{prop:holonomy_basis}:
\begin{equation}
\label{eq:physical_operators_pure_gauge}
    H_l^m = (\mathcal{R}^{\bm{0}}_R)^\dagger U_l^m \mathcal{R}^{\bm{0}}_R \quad\text{and}\quad X_l(\lambda) = (\mathcal{R}^{\bm{0}}_R)^\dagger X_l(\lambda) \mathcal{R}^{\bm{0}}_R.
\end{equation}

\subsection{QRFs from Fermionic Matter}

\label{ssec:QED_matter_QRF}

So far, we discussed QRFs and their interpretation in the case of pure gauge lattice QED. In the following, we will include dynamical fermionic matter into the system.
As we will see, spanning trees still provide QRFs for this case; however, these QRFs are no longer complete. Moreover, we can also construct a QRF from the fermionic degrees of freedom.

\paragraph{The kinematical space of the full theory.} The kinematical Hilbert space of full lattice QED also includes the matter degrees of freedom on the sites,
\begin{equation}
    \HS_{\text{kin}} = \bigotimes_{v\in\mathcal{V}}\mathbb{C}^2\bigotimes_{l\in\mathcal{L}}\HS_\rot.
\end{equation}
The qubit spaces on the sites represent the fermions, and the computational basis indicates the occupation number, $\psi_v^\dagger\psi_v \ket{n}_v=n\ket{n}_v$ for $ n = 0, 1$. The constraints now also include the charge density. Importantly, this means that they are all independent and \cref{eq:sum_to_zero} is modified to
\begin{equation}
    \sum_{v\in\mathcal{V}}\mathcal{C}_v = -\sum_{v\in\mathcal{V}}\rho_v.
\end{equation}
We denote the constraints without the charge density as $\mathcal{C}_v^{\mathcal{L}}$, i.e., $\mathcal{C}_v = \mathcal{C}_v^{\mathcal{L}} - \rho_v$.

In this subsection, we focus on the situation where $\Gamma$ is finite such that $|\mathcal{V}|< \infty$ with either smooth or periodic boundary conditions.

\paragraph{Spanning tree QRFs for full lattice QED.}
A spanning tree still yields a QRF for the gauge group in the presence of charged matter. However, in contrast to \cref{thm:LGT_QRF_pure} for the pure gauge theory, this QRF does not suffice to resolve the full action of $\mathcal{G}$. Since all $|\mathcal{V}|$ constraints are independent, the constraints $\mathcal{C}_{V_l}$ with $l\in R$  from property 2 in \cref{thm:LGT_QRF_pure} cannot be used to uniquely parametrize gauge transformations. In particular, the orientation states $\ket{\phi(\bm{\lambda})}_R$ are invariant under the action of the global transformation generated by $\sum_{v\in\mathcal{V}}\mathcal{C}_v$. A spanning tree QRF $R$ is thus {\it incomplete} for lattice QED including charged matter and the operators $\mathcal{P}_R^{\bm{\lambda}}$ constitute a partial gauge-fixing.

\paragraph{Fermionic field as QRF.} Instead of using link degrees of freedom as the frame for $\mathcal{G}$, we can use the fermionic ones on the sites. A gauge transformation $e^{i\lambda_v\mathcal{C}_v}$ acts on the fermionic field via $e^{-i\lambda_v\rho_v}$ yielding a local phase $\psi_v \mapsto e^{i\lambda_v}\psi_v$. We can thus use the number states $\ket{0}_v$ and $\ket{1}_v$ to build a QRF for the local phase. This is reminiscent of a set of independently constrained 2-level quantum clocks on a lattice, using the relative phase between the two levels to parametrize a $\operatorname{U}(1)$ action.

\begin{restatable}[Fermionic field QRFs]{theorem}{matterfieldqrfs}
\label{thm:matter_qrf}
     The matter degrees of freedom yield a QRF $\tilde{R}$ such that $\HS_{\tilde{R}} = \HS_{\text{matter}} = \bigotimes_{v\in\mathcal{V}}\mathbb{C}^2$, and the remaining subsystem $\tilde{S}$ consists of all the link degrees of freedom. With respect to this structure, the gauge transformations split into a product
    \begin{equation}
        G(\bm{\lambda}) = 
        \underbrace{\bigotimes_{v\in\mathcal{V}}e^{-i\lambda_v \rho_v}}_{G_{\tilde{R}}(\bm{\lambda})} \otimes \underbrace{\bigotimes_{v\in\mathcal{V}}e^{i\lambda_v\mathcal{C}_v^{\mathcal{L}}}}_{G_{\tilde{S}}(\bm{\lambda})}.
    \end{equation}
    The orientation states
    \begin{equation}\label{eq:matter_qrf}
        \ket{\bm{\lambda}}_{\tilde{R}} \defeq \bigotimes_{v\in\mathcal{V}}\tfrac{1}{\sqrt{2}}(\ket{0}_v+e^{-i\lambda_v}\ket{1}_v), \quad \bm{\lambda}\in[0, 2\pi)^{\times\mathcal{|V|}}
    \end{equation}
    are acted upon transitively and freely by $\mathcal{G}$, making $\tilde{R}$ a complete QRF. 
\end{restatable}

\begin{proof}
See \cref{app:proofs}.
\end{proof}

This QRF is complete, but the different orientation states are not perfectly distinguishable; instead, their overlap is
\begin{equation}
    |\braket{\bm{\eta}}{\bm{\lambda}}_{\tilde{R}}|^2 = \prod_{v\in\mathcal{V}}\left(\frac{1}{2} + \frac{1}{2}\cos(\eta_v-\lambda_v)\right).
\end{equation}
Therefore, the fermionic field QRF $\tilde{R}$ is not ideal.
The unitary gauge-fixing operators from \cref{eq:gauge-fixing} for the QRF $\tilde{R}$ are
$\mathcal{P}_{\tilde{R}}^{\bm{\lambda}} = 2^{|\mathcal{V}|/2}(\ket{\bm{\lambda}}\bra{\bm{\lambda}}_{\tilde{R}}\otimes I_{\tilde{S}})$
and fix the local phase of the field to $\lambda_v$ at the vertex $v$.

\paragraph{The reduced subspace and gauge-invariant operators.}

We can discuss the physical information encoded in the staggered fermion model by using the reduction maps associated with the fermionic field QRF. Since these remove the space of the fermionic matter, we can express any physical state in terms of only the electric flux on the links.

Indeed, any state $\ket{\psi}\in\HS_\phys$ contains sufficient information on the links alone to reconstruct the fermionic field on the sites via Gauss' law. This is precisely why the fermionic field can be treated as a QRF in the first place.
\begin{proposition}
\label{prop:electric_basis}
    The reduced subspace $\HS_{\tilde{S}|\tilde{R}} = \mathcal{R}_{\tilde{R}}^{\bm{\lambda}}(\HS_\phys)$ is independent of the orientation $\bm{\lambda}$ and spanned by the electric basis states
    \begin{equation}
    \label{eq:electric_basis}
        \bigotimes_{l\in\mathcal{L}}\ket{k_l}_l \, ,\quad \forall v\in\mathcal{V}: \sum\nolimits_i k_{[v, v+e_i]} - k_{[v, v-e_i]} \in \{-j_v, 1-j_v\}\,.
    \end{equation}
    The associated projector is given by
    \begin{equation}
        \Pi_{\tilde{S}|\tilde{R}} = \mathcal{R}_{\tilde{R}}^{\bm{\lambda}}\circ (\mathcal{R}_{\tilde{R}}^{\bm{\lambda}})^\dagger = \prod_{v\in\mathcal{V}}\left(\Pi_{(-j_v)} + \Pi_{(1-j_v)}\right),
    \end{equation}
    where $\Pi_{q_v} = \int_0^{2\pi}\frac{\D\lambda}{2\pi} \ e^{i\lambda(\mathcal{C}_v^{\mathcal{L}}-q_v)}$ projects onto eigenstates of $\mathcal C_v^{\mathcal L}$ with eigenvalue $q_v$.
\end{proposition}
\begin{proof}
    The electric basis states spanning $\HS_{\tilde{S}|\tilde{R}}$ follow directly from the form of the projector $\Pi_{\tilde{S}|\tilde{R}}$, which we find by a direct calculation in \cref{app:staggered_fermion_lqed_operators}.
\end{proof}
The reduction maps associated with the fermionic field QRF are thus not surjective onto $\HS_{\tilde{S}}$. Instead, they are unitaries between $\HS_\phys$ and $\HS_{\tilde{S}|\tilde{R}}$, and we can express the physical states in terms of the basis $\ket{\{k_l\}} \defeq (\mathcal{R}_{\tilde{R}}^{\bm{\lambda}})^\dagger\bigotimes_{l\in\mathcal{L}}\ket{k_l}$ with $\bigotimes_{l\in\mathcal{L}}\ket{k_l}_l$ as in \cref{eq:electric_basis}. Note that $\epsilon_l\ket{\{k_{l'}\}} = k_{l'}\ket{\{k_{l'}\}}$, i.e., this basis labels the physical states in terms of their electric flux.
This means that the encoded physical information in $\HS_\phys$ can be understood purely in terms of the electric flux on the links. This electric flux is further constrained by the allowed eigenvalues of the reduced Gauss-law operators:
\begin{equation}
    \mathcal{C}_v^{\mathcal{L}} = \begin{cases}
        0, 1\quad &j_v = 0\\
        -1, 0 &j_v = 1
    \end{cases}\ .
\label{eq:reduced_background_charges}
\end{equation}
These eigenvalues correspond precisely to the allowed charge of staggered fermions on $v$ and reflect the constraint $\mathcal{C}_v = \mathcal{C}_v^{\mathcal{L}} - \rho_v = 0$ since $\rho_v$ has eigenvalues $\{-j_v, 1-j_v\}$.
In the reduced picture, the eigenvalue of $\mathcal C_v^{\mathcal L}$ encodes the presence or absence of a staggered background charge.
Thus, the reduced states look like link configurations with at most one background (anti-)charge per site, consistent with fermionic occupation constraints.

The gauge-invariant operators reduced to $\HS_{\tilde{S}|\tilde{R}}$ reflect this structure. The holonomy operators $H_\gamma$ and the $X_l(\lambda)$-operators act on the reduced space exactly as on $\HS_\phys$, i.e.,
\begin{equation}
    \mathcal{R}_{\tilde{R}}^{\bm{\lambda}}H_\gamma(\mathcal{R}_{\tilde{R}}^{\bm{\lambda}})^\dagger = H_\gamma \quad {\rm and} \quad \mathcal{R}_{\tilde{R}}^{\bm{\lambda}}X_l(\lambda')(\mathcal{R}_{\tilde{R}}^{\bm{\lambda}})^\dagger = X_l(\lambda').
\end{equation}
In addition, there are gauge-invariant operators involving the sites. As shown in \cref{app:staggered_fermion_lqed_operators}, the hopping terms reduce to
\begin{equation}
\label{eq:reduced_hopping}
\mathcal{R}_{\tilde{R}}^{\bm{\lambda}} \circ \psi_v^\dagger U_{[v, v']}\psi_{v'}\circ (\mathcal{R}_{\tilde{R}}^{\bm{\lambda}})^\dagger = U_{[v, v']}\circ\Pi_{(-j_v)}\Pi_{(1-j_{v'})}\Pi_{\tilde{S}|\tilde{R}}.
\end{equation}
The link operator changes the electric flux, adding and removing background (anti-)charges on the adjacent sites to the reduced state. 
The projectors $\Pi_{(-j_v)}$ and $\Pi_{(1-j_{v'})}$ ensure that the charge constraints \labelcref{eq:reduced_background_charges} remain satisfied upon this action. Their kernels represent the action of $\psi_{v'}\ket{0} _{v'}= 0$ and $\psi_v^\dagger\ket{1}_v = 0$ in the hopping term.
The reduced operator therefore reproduces the fermionic creation and annihilation operators as well as the link operator purely in terms of an electric flux change and background charge constraints on the reduced space. From this, we can read off the action of the hopping term on the basis $\ket{\{k_l\}}$:
\begin{equation}
    \psi_v^\dagger U_{[v, v']}\psi_{v'} \ket{\{k_l\}} = \begin{cases}
        \ket{\{k_l + \delta_{l, {[v, v']}}\}} & \rho_v = -j_v, \rho_{v'} = 1-j_{v'} \\
        0 &{\rm else},
    \end{cases}
\end{equation}
where $\rho_v = \sum\nolimits_i k_{[v, v+e_i]} - k_{[v, v-e_i]}$. Similarly, The more general hopping term $\psi_v^\dagger W_\gamma\psi_{v'}$ on for any two vertices $v, v'$ connected by a path $\gamma$ corresponds to the reduced operator
\begin{equation}
    W_\gamma\circ\Pi_{(-j_v)}\Pi_{(1-j_{v'})}\Pi_{\tilde{S}|\tilde{R}}.
\end{equation}

The projector $\psi_v^\dagger\psi_v = \ket{1}\bra{1}_v$ reduces to
\begin{equation}
\label{eq:reduced_number}
    \mathcal{R}_{\tilde{R}}^{\bm{\lambda}} \circ \psi_v^\dagger\psi_v\circ (\mathcal{R}_{\tilde{R}}^{\bm{\lambda}})^\dagger  = \Pi_{(1-j_v)}\Pi_{\tilde{S}|\tilde{R}}.
\end{equation}
The projector $\Pi_{(1-j_v)}$ enforces the correct charge sector on the site $v$. Analogously, $\psi_v\psi_v^\dagger = \ket{0}\bra{0}_v$ reduces to $\Pi_{-j_v}\Pi_{\tilde{S}|\tilde{R}}$.

\section{Lattice Quantum Electrodynamics as an Error Correction Code}
\label{sec:LQED_as_QECC}
By interpreting the gauge group as a generalized stabilizer group for lattice QED, we obtain two types QECCs corresponding to lattice QED with or without fermionic matter.
The one associated with pure gauge has the structure of a quantum rotor code \cite{Vuillot_2024} (see \cref{sssec:rotor_codes}), while in the case of fermionic matter the kinematical Hilbert space consists of both qubits and quantum rotors.
In \cref{sec:QEC_in_Gauge}, we identified QRFs for both types of QECC. Here, we use these QRFs as a tool to understand the QECCs given by lattice QED. 
Recalling the dictionary of \cite{CCHM_2024} and \cref{sec:QEC_in_Gauge}, the inverse reduction map of a QRF acts as encoding isometry of the associated QECC. Thus, choosing a QRF amounts to choosing an encoding of the logical data, and simultaneously defines a split between physical and redundant degrees of freedom. This perspective, together with \cref{thm:correctable_gauge-fixing}, allows us to identify sets of correctable gauge-violating errors. 

Using this approach, we characterize the error-correcting properties of the codes resulting from lattice QED. Because the theory is Abelian, gauge-fixing errors can be recovered via charge measurements (see \cref{prop:A_q_from_gf} and \cref{prop:non_ideal_A_r}) which correspond to measuring Gauss constraints. 
We also find correctable errors which are not connected to the QRFs through different choices of operators mapping from the non-trivial charge sectors to the code subspace. Our results recover and extend known features: For pure gauge systems with finite gauge groups, a set of correctable errors is given by all single link operators which change the generalized electric flux \cite[Section III.2]{CLLL_2024}. We find the same result for pure gauge $\operatorname{U}(1)$ lattice QED and extend it to case with staggered fermions. 
In \cite[Theorem 1]{spagnoli2024}, it was shown that $\mathbb{Z}_2$ lattice QED with staggered fermions, when treated as a QECC, can correct any single error in electric flux on the links or in the fermion occupation number on the sites. We show that the same is true in the $\operatorname{U}(1)$ theory, where electric flux is quantized to $\mathbb{Z}$ instead of $\mathbb{Z}_2$.

\subsection{U(1)-Charges and Constraint Measurements}

Before we examine the error correction code properties of lattice QED, let us briefly consider its charge structure. Following the outline in \cref{ssec:QEC_and_gauge} and \cref{prop:A_q_from_gf}, we interpret the recovery operations associated with the correctable error sets in the remainder of this section as measurements of the $\operatorname{U}(1)$-charges, or equivalently, as measurements of the Gauss’ law constraints.

\paragraph{Pure gauge lattice QED.}
The group of gauge transformations $\mathcal{G}$ of pure gauge lattice QED on a finite lattice is generated by $|\mathcal{V}|-1$ constraints $\mathcal{C}_v$ due to \cref{eq:sum_to_zero}. We may choose any subset of vertices $\mathcal{V}'\subset \mathcal{V}$ which contains all but one vertex. Then, $\mathcal{G}$ can naturally be seen as a representation of $\operatorname{U}(1)^{|\mathcal{V}|-1}$ on the kinematical Hilbert space parametrized by the constraints associated with $\mathcal{V}'$:
\begin{equation}
    (e^{i\lambda_v})_{v\in\mathcal{V}'}\mapsto G(\bm{\lambda}) = \prod_{v\in\mathcal{V}'}e^{i\lambda_v\mathcal{C}_v}.
\end{equation} 
The kinematical space decomposes into the direct sum of the charge sectors $\HS_{\bm{q}}$,
\begin{equation}
    \HS_\kin = \bigoplus_{\bm{q}}\HS_{\bm{q}}.
\end{equation}
The integer vector of charges $\bm{q}=(q_v)_{v\in\mathcal{V}'}\in\mathbb{Z}^{|\mathcal{V}|-1}$ labels the irreducible representations (i.e., the characters) of $\operatorname{U}(1)^{|\mathcal{V}|-1}$ and contains one component for each independent vertex transformation, $\chi_{\bm{q}}(\bm{\lambda}) = \prod_{v\in\mathcal{V}'}e^{iq_v\lambda_v}$. The charge sectors $\HS_{\bm{q}}$ are the isotypes of the characters $\chi_{\bm{q}}$ within $\HS_\kin$. Therefore, $\bm{q}$ here does not represent charge of physical matter, but dictates the eigenvalues of the constraints $\mathcal{C}_v$ on states in $\HS_{\bm{q}}$. We interpret non-zero components $q_v$ of $\bm{q}$ as indicating the presence of fixed electric {\it background charges} on $v$ on the lattice\footnote{Despite $\bm{q}$ having only $|\mathcal{V}|-1$ entries, it also determines the background charge on the remaining vertex since \cref{eq:sum_to_zero} implies that the total charge vanishes.}. The physical states live in $\HS_\phys = \HS_{\bm{0}}$.

Since $\mathcal{C}_v\ket{\psi} = q_v\ket{\psi}$ for $\ket{\psi}\in\HS_{\bm{q}}$, measuring the constraints $\mathcal{C}_v$ (for $v\in\mathcal{V}'$) for a given a state $\ket{\psi}$ yields the charges.
The orthogonal projector onto the corresponding subspace is given by
\begin{equation}
\label{eq:charge_projector}
    \Pi_{q_v} = \int_0^{2\pi}\frac{\D\lambda}{2\pi} e^{i\lambda(\mathcal{C}_v - q_v)}, 
\end{equation}
and a full measurement of the set $\{\mathcal{C}_v\}_{v\in\mathcal{V}'}$ with outcome $\bm{q}$ projects onto the charge sector $\HS_{\bm{q}}$ according to\footnote{This expression also recovers $\Pi_\phys$ from \cref{eq:gauss_projector}, despite the extra integral over $v\notin\mathcal{V}'$ there. This extra integral can be eliminated by rewriting the constraint using \cref{eq:sum_to_zero} and a change of variables.}
\begin{equation}
    \Pi_{\bm{q}} = \prod_{v\in\mathcal{V}'}\Pi_{q_v}.
\end{equation}

This charge-sector decomposition provides the basic mechanism for error correction. Errors map the state to a different outside the charge sector, and measuring the constraints obtaining outcome $\bm{q}$, the error state is projected to $\HS_{\bm{q}}$. The recovery can be done through an appropriate choice of unitary $A_{\bm{q}}$, which maps $\HS_\phys$ into $\HS_{\bm{q}}$. For the ideal spanning tree QRFs, an example of such set $\{A_q\}_q$ is given by \cref{prop:A_q_from_gf}.
An error set $\mathcal{E} = \{E_i\}_i$ can be corrected by this operation if for each $i$ and $\bm{q}$, $A_{\bm{q}}^\dagger\Pi_{\bm{q}}E_i|_{\HS_\phys}\propto I_{\HS_\phys}$.

We apply these ideas in \cref{ssec:pglqed_qecc} to find correctable error sets based on constraint measurements. 

\paragraph{Staggered fermion lattice QED.} 

For lattice QED with staggered fermionic matter, $\mathcal{G}$ is a representation of $\operatorname{U}(1)^{|\mathcal{V}|}$. The charge sectors are thus labeled by ${\bm{q}}\in\mathbb{Z}^{|\mathcal{V}|}$ and correspond to eigenvalues of constraints including the charge density $\mathcal{C}_v = \mathcal{C}_v^{\mathcal{L}} - \rho_v$, where $\mathcal{C}_v^{\mathcal{L}}$ is the component of the constraint acting on the links.
We can now follow the same discussion as in the pure gauge case to find the projectors $\Pi_{{\bm{q}}}$, and choose operators $\{A_{{\bm{q}}}\}_{\bm{q}}$ to obtain a recovery operation following a constraint measurement. 
We will, however, need to perform a coarse-grained version of the constraint measurement first if we want to correct error sets which affect the sites of the lattice.

\subsection{Pure Gauge Lattice QED as a QECC}
\label{ssec:pglqed_qecc}

In pure gauge lattice QED, all the relevant degrees of freedom are quantum rotors. The system thus has the structure of a quantum rotor code with $\HS_\physical = \HS_\kin = \bigotimes_{l\in\mathcal{L}}\HS_\rot$ (see \cref{sssec:rotor_codes} for a brief introduction to rotor codes). The stabilizer $\mathcal{G}$ is composed entirely of $X$-type operators,
\begin{equation}
    G(\bm{\lambda}) = \bigotimes_{v\in\mathcal{V}}X_v(\lambda_v), \quad \bm{\lambda} = (\lambda_v)_{v\in\mathcal{V}}.
\end{equation}

\subsubsection{The QECC Structure}

In \cref{sec:QRFs_in_LQED} we introduced the holonomy basis of $\HS_\phys$. Since we now identify $\HS_\phys$ with the code space $\HS_\code$, the holonomy operators become logical operators of the code. The logical $U$-type operators $\Log_U(\HS_\phys)$ on the code subspace are generated by the holonomy operators. The logical $X$-type operators $\Log_X(\HS_\phys)$ are generated by the operators $X_l(\lambda)$ modulo the stabilizer group $\mathcal{G}$. 

To make the encoding structure explicit, choose a spanning tree QRF $R$ of the lattice. We can parametrize the encoded $|\mathcal{V}|-1$ quantum rotors by the links in the residual system $S = \mathcal{L}\setminus R$, $S\ni l\mapsto i_l\in\{1, 2, \dots, |S|\}$. As in \cref{eq:physical_operators_pure_gauge}, we can make the identifications
\begin{equation}
        \bigotimes_{i_l = 1}^{|S|}\HS_{\rot, i_l} \cong \HS_\phys, \quad
        \begin{cases}
            X_{i_l}(\lambda) \mapsto X_l(\lambda) \in \Log_X(\HS_\phys) \\
        U_{i_l}^{m} \mapsto H_l^m \in \Log_U(\HS_\phys)
        \end{cases}.
\end{equation}

\Cref{prop:holonomy_basis} makes this isomorphism explicit via the inverse reduction map $(\mathcal{R}_R^{\bm{0}})^\dagger$. The logical Hilbert space is $\HS_S \cong \bigotimes_{i_l = 1}^{|S|}\HS_{\rot, i_l}$ and $(\mathcal{R}_R^{\bm{0}})^\dagger:\HS_S\rightarrow\HS_\phys$ plays the role of the encoding isometry.

Thus, pure gauge lattice QED can be viewed as a quantum rotor code in which physical states are gauge-invariant rotor configurations and logical information is stored in holonomies.

\paragraph{Errors and electric flux lines.} 

Ultimately, we are interested in finding the types of errors which are correctable in lattice QED. Since the operators $X(\lambda)U^m$ from a Hilbert-Schmidt basis on $\HS_\rot$, we consider errors in terms of these generalized Pauli operators as in \cite{Vuillot_2024}. 

Any $X_l(\lambda)$ is a logical/gauge-invariant operator and thus preserves Gauss' law. Consequently, an $X$-type error on $\HS_\phys$ is not correctable. In contrast, the gauge-violating $U$-type errors can be detected and potentially corrected. These $U$-type errors are the $\operatorname{U}(1)$-analogue of phase flip errors. In the language of lattice QED, they create electric flux lines \cite{KS1975} as we will see now. 

Recall that on the electric flux basis we have $U_l^m\ket{k}_l = \ket{k+m}_l$, increasing the electric flux on the link $l = [v, v']$ by $m$ units. Starting from a physical state $\ket{\psi}_\phys$, the error state $\ket{\psi'}=U_l\ket{\psi}_\phys$ violates Gauss' law at endpoints of the link,
\begin{equation}
    \mathcal{C}_v\ket{\psi'} = +1\ket{\psi'}, \quad \mathcal{C}_{v'}\ket{\psi'} = -1\ket{\psi'}.
\end{equation}
This state lies in a non-trivial charge sector of $\HS_\kin$ and the Gauss' law violation can be interpreted as a pair of opposite electric background charges on the lattice. More generally, a Wilson line $W_\gamma$ along a path $\gamma$ ``excites'' a pair of background charges at the end points of $\gamma$ by creating a line of electric flux along $\gamma$, where the electric flux is modified from the gauge-invariant configuration. These electric flux lines are thus unphysical configurations of the gauge field.

\subsubsection{Correctable Errors for Pure Gauge Lattice QED}
We now identify correctable sets of errors for pure gauge lattice QED within the QECC we constructed.

\paragraph{Correctable error sets associated with spanning tree QRFs.}
Let $R$ be a spanning tree QRF as in \cref{thm:LGT_QRF_pure}. The set of correctable gauge-fixing operators for this QRF consists of
\begin{equation}
    \mathcal{P}_R^{\bm{\lambda}} = (2\pi)^\frac{|\mathcal{V}|-1}{2}(\ket{\phi(\bm{\lambda})}\bra{\phi(\bm{\lambda})}_R \otimes I_S),
\end{equation}
where $\ket{\phi(\bm{\lambda})}_R = \bigotimes_{l\in R}\ket{e^{i\lambda_l}}_l$. Note that we chose the links in the spanning tree to parametrize $\mathcal{G}$, and thus get charges $q_l$ as the eigenvalues of the constraints $\mathcal{C}_{V_l}$ from \cref{thm:LGT_QRF_pure}.
Since $R$ is ideal and $\mathcal{G}$ is Abelian, we can use \cref{prop:A_q_from_gf} to find 
\begin{equation}
    A_{\bm{q}} = (2\pi)^{-\frac{|\mathcal{V}|-1}{2}}\int_0^{2\pi}\left(\prod\nolimits_{l\in R}\D \lambda_l\ e^{-iq_l\lambda_l}\right)\mathcal{P}_R^{\bm{\lambda}}.
\end{equation}
To find the explicit form of this operator, note that on a single link, we can write
\begin{equation}
    \int_0^{2\pi}\D\lambda_l \ e^{-iq\lambda_l}\ket{e^{i\lambda_l}}\bra{e^{i\lambda_l}}_l = U_l^{-q}.
\end{equation}
This readily generalizes to
\begin{equation}
     A_{\bm{q}}= (2\pi)^{-\frac{|\mathcal{V}|-1}{2}}\int_0^{2\pi} \left(\prod\nolimits_{l\in R}\D \lambda_l \ e^{-iq_l\lambda_l}\right) \mathcal{P}_R^{\bm{\lambda}} = \bigotimes_{l\in R}U_l^{-q_l}.
\end{equation}
We thus find the following statement.

\begin{proposition}[Correctable errors on spanning trees]
\label{prop:correctable_U_on_T_in_QED}
    Consider a spanning tree QRF $R$. Then, the set
    \begin{equation}
    \label{eq:correctable_errors_on_R}
        \left\{\bigotimes\nolimits_{l\in R}U_l^{m_l} \ \middle| \ \forall l\in R: m_l \in \mathbb{Z}\right\}
    \end{equation}
    of all $U$-type errors located on $R$ is correctable.
\end{proposition}

Following the discussion around \cref{prop:A_q_from_gf}, we can understand the recovery protocol for these errors in terms of a measurement of the constraints $\mathcal{C}_{V_l}$. Due to $[\epsilon_l, U_l] = U_l$ and $\mathcal{C}_{V_l, R} = -\epsilon_l$, we find
\begin{equation}  
\mathcal{C}_{V_l}U_l^m\ket{\psi} = -(U_l^m\mathcal{C}_{V_l} +mU_l^m)\ket{\psi} = -m U_l^m\ket{\psi}
\end{equation}
for $\ket{\psi}\in\HS_\phys$, as expected.
Thus, measuring the constraints $\mathcal{C}_{V_l}$ reveals the eigenvalues $m_l$, uniquely identifying the error $\bigotimes_{l\in R}U_{l}^{m_{l}}$ supported on $R$.

\paragraph{Any single $U$-type error is correctable.}
Spanning tree QRFs are supported on a subset of links, and therefore, the associated correctable error sets are restricted to these links. However, a different recovery operation allows correction of single link errors anywhere on the lattice.
To see this, consider a measurement of the local constraints $\{\mathcal{C}_{{v}}\}_{v\in\mathcal{V}}$ which projects onto the charge sectors by $\{\Pi_{\bm{q}}\}_{\bm{q}\in\mathbb{Z}^{|\mathcal{V}|}}$.
An error $U_l^m$ on $l = [v, v']$ shows up as $q_v=m,  q_{v'}=-m$ and can thus be uniquely determined from the error set $\{U_l^m  \, | \,  m\in\mathbb{Z}, l\in\mathcal{L}\}$. We therefore choose $A_{\bm{q}} = U_l^m$ whenever $\bm{q} = (\delta_{w, v}m - \delta_{w,v'}m)_{w\in\mathcal{V}}$ while remaining agnostic about errors which map to other charge sectors\footnote{As soon as two background charges are separated by more than one link, the choice of a physically motivated ``lowest weight'' error is generally ambiguous.}.

This shows the following result.

\begin{proposition}
\label{prop:correctable_U_errors}
    The set $\mathcal{E} = \{U_l^m  \, | \,  m\in\mathbb{Z}, l\in\mathcal{L}\}$ is correctable on $\HS_\phys$.
\end{proposition}

Note that this is consistent with the code $U$-distance (see \cref{sssec:rotor_codes}). If we let the weight of a $U$-type operator be the number of non-trivial factors, the $U$-distance of lattice QED is $d_{U} = 4$, where the minimum weight is attained by the plaquette operators.
Indeed, the code corrects any $t_U = \left\lfloor \frac{d_{U,w}-1}{2}\right\rfloor = 1$ $U$-error.

\paragraph{Visualizing the recovery protocols.}
The correctable error sets from \cref{prop:correctable_U_on_T_in_QED} and \cref{prop:correctable_U_errors} correspond to different recovery protocols upon measuring the constraints $\mathcal{C}_v$.

Given a spanning tree $T$, the associated correctable errors are products of Wilson line operators $W_\gamma$ with $\gamma\subset T$. As previously mentioned, these create electric flux lines. To illustrate these errors, we consider the ground state in the {\it strong-coupling limit} $g\gg 1$ where the electric term in \cref{eq:KS_ham_pg} dominates, i.e., $H\approx \frac{g^2}{2a}\sum_{l\in\mathcal{L}}\epsilon_l^2$. The ground state is the gauge-invariant vacuum configuration $\ket{\Omega_{\text{sc}}} =\bigotimes_{l\in\mathcal{L}}\ket{0}_l$.
A Wilson line operator $W_\gamma$ on a path from $v$ to $v'$ creates an electric flux line along $\gamma$,
\begin{equation}\label{eq:WilsonlineError}
    W_\gamma \ket{\Omega_{\text{sc}}} = \bigotimes_{l\in\mathcal{L}\setminus \gamma}\ket{0}_l \bigotimes_{l\in\gamma}\ket{\sigma_l}.
\end{equation}
As before, $\sigma_l = \pm1$ indicates whether $\gamma$ passes $l$ in the positive or negative direction. For this state, $\mathcal{C}_v = +1$ and $\mathcal{C}_{v'} = -1$, and so $\gamma$ connects an opposite pair of background charges.

The recovery protocol associated with the spanning tree $T$ is given by the following steps:
\begin{enumerate}
\itemsep0em 
    \item measure the constraints $\mathcal{C}_v$,
    \item identify pairs of negative and positive charges and connect them along paths $\gamma_i\subset T$, and
    \item apply the operator $\prod_iW_{\gamma_i}$.
\end{enumerate}
The operator $\prod_iW_{\gamma_i}$ is independent of the choice of how the pairs of negative and positive charges are connected. This removes all background charges and corrects any $U$-type error which is supported on $T$, i.e., the error set from \cref{prop:correctable_U_on_T_in_QED}. This is illustrated in \cref{fig:electric_flux_lines}.

\begin{figure}[h!]
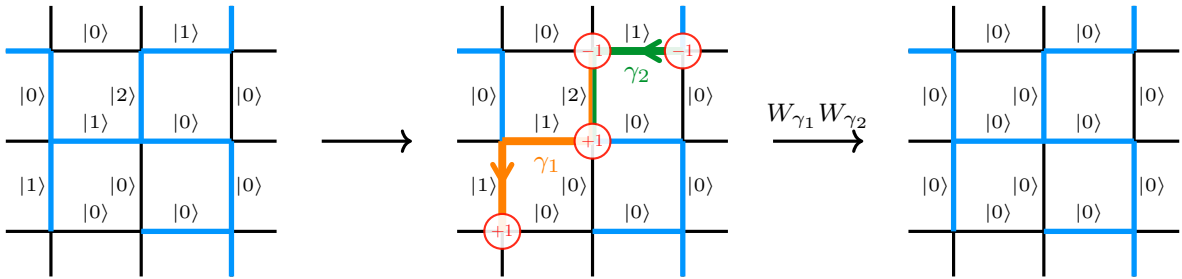

    \centering
    $\centertikz{
    \def\s{1}
    \def\L{2.5}
    \begin{scope}[scale=1.2,transform shape]
        \begin{scope}[shift = {(-5,0)}]
        \draw[very thick] (0,-0.5) -- (0,\L);
        \draw[very thick] (\s,-0.5) -- (\s,\L);
        \draw[very thick] (2*\s,-0.5) -- (2*\s,\L);
        \draw[very thick] (-0.5,0) -- (\L,0);
        \draw[very thick] (-0.5,\s) -- (\L,\s);
        \draw[very thick] (-0.5,2*\s) -- (\L,2*\s);
        \node at (0.5,0.2) {\tiny$\ket{0}$};
        \node at (0.5,0.2+\s) {\tiny$\ket{1}$};
        \node at (0.5,0.2+2*\s) {\tiny$\ket{0}$};

        \node at (0.5+\s,0.2) {\tiny$\ket{0}$};
        \node at (0.5+\s,0.2+\s) {\tiny$\ket{0}$};
        \node at (0.5+\s,0.2+2*\s) {\tiny$\ket{1}$};

        \node at (-0.2,0.5+\s) {\tiny$\ket{0}$};
        \node at (-0.2,0.5) {\tiny$\ket{1}$};

        \node at (-0.2+\s,0.5+\s) {\tiny$\ket{2}$};
        \node at (-0.2+\s,0.5) {\tiny$\ket{0}$};

        \node at (0.2+2*\s,0.5+\s) {\tiny$\ket{0}$};
        \node at (0.2+2*\s,0.5) {\tiny$\ket{0}$};

        \draw[line width = 2pt, color=lightblue] (0,0) -- (0,\L-0.5);
        \draw[line width = 2pt, color=lightblue] (0,\s) -- (\L-0.5,\s);
        \draw[line width = 2pt, color=lightblue] (\s,0) -- (\L-0.5,0);
        \draw[line width = 2pt, color=lightblue] (2*\s,\s) -- (2*\s,-0.5);
        \draw[line width = 2pt, color=lightblue] (\s,\s) -- (\s,2*\s);
        \draw[line width = 2pt, color=lightblue] (\s,2*\s) -- (2*\s,2*\s);
        \draw[line width = 2pt, color=lightblue] (2*\s,2*\s) -- (2*\s,2*\s+0.5);
        \draw[line width = 2pt, color=lightblue] (0,2*\s) -- (-0.5,2*\s);
    \end{scope}
    \draw[very thick, ->] (-2, \s) -- (-1,\s);
    \begin{scope}[shift = {(0,0)}]
        \draw[very thick] (0,-0.5) -- (0,\L);
        \draw[very thick] (\s,-0.5) -- (\s,\L);
        \draw[very thick] (2*\s,-0.5) -- (2*\s,\L);
        \draw[very thick] (-0.5,0) -- (\L,0);
        \draw[very thick] (-0.5,\s) -- (\L,\s);
        \draw[very thick] (-0.5,2*\s) -- (\L,2*\s);
        \node at (0.5,0.2) {\tiny$\ket{0}$};
        \node at (0.5,0.2+\s) {\tiny$\ket{1}$};
        \node at (0.5,0.2+2*\s) {\tiny$\ket{0}$};

        \node at (0.5+\s,0.2) {\tiny$\ket{0}$};
        \node at (0.5+\s,0.2+\s) {\tiny$\ket{0}$};
        \node at (0.5+\s,0.2+2*\s) {\tiny$\ket{1}$};

        \node at (-0.2,0.5+\s) {\tiny$\ket{0}$};
        \node at (-0.2,0.5) {\tiny$\ket{1}$};

        \node at (-0.2+\s,0.5+\s) {\tiny$\ket{2}$};
        \node at (-0.2+\s,0.5) {\tiny$\ket{0}$};

        \node at (0.2+2*\s,0.5+\s) {\tiny$\ket{0}$};
        \node at (0.2+2*\s,0.5) {\tiny$\ket{0}$};

        \draw[line width = 2pt, color=lightblue] (0,0) -- (0,\L-0.5);
        \draw[line width = 2pt, color=lightblue] (0,\s) -- (\L-0.5,\s);
        \draw[line width = 2pt, color=lightblue] (\s,0) -- (\L-0.5,0);
        \draw[line width = 2pt, color=lightblue] (2*\s,\s) -- (2*\s,-0.5);
        \draw[line width = 2pt, color=lightblue] (\s,\s) -- (\s,2*\s);
        \draw[line width = 2pt, color=lightblue] (\s,2*\s) -- (2*\s,2*\s);
        \draw[line width = 2pt, color=lightblue] (2*\s,2*\s) -- (2*\s,2*\s+0.5);
        \draw[line width = 2pt, color=lightblue] (0,2*\s) -- (-0.5,2*\s);
        \draw[line width = 1.5pt, color=mygreen] (\s+0.0221,2*\s) -- (\s+0.0221,\s);
    
    \draw[line width = 1.5pt, color=orange] (\s-0.0221,2*\s) -- (\s-0.0221,\s);
        \draw[line width = 3pt, color=mygreen,decoration={markings,
	mark= at position 0.5 with {\arrow[line width=2pt]{angle 60}}},postaction={decorate}] (2*\s,2*\s) --  node[below] {\scriptsize$\gamma_2$} (\s,2*\s);
        \draw[line width = 3pt, color=orange,decoration={markings,
	mark= at position 0.5 with {\arrow[line width=2pt]{angle 60}}},postaction={decorate}] (0,\s) -- (0,0);
        \draw[line width = 3pt, color=orange,] (0,\s) -- node[below] {\scriptsize$\gamma_1$} (\s,\s);

        \node[thick, circ, draw=myred, fill=white,scale=0.5, minimum size=22pt,opacity=0.9] at (0,0) {\color{red}{$+1$}};
        \node[thick, circ, draw=myred, fill=white,scale=0.5, minimum size=22pt,opacity=0.9] at (\s,\s) {\color{red}{$+1$}};
        \node[thick, circ, draw=myred, fill=white,scale=0.5, minimum size=22pt,opacity=0.9] at (\s,2*\s) {\color{red}{$-1$}};
        \node[thick, circ, draw=myred, fill=white,scale=0.5, minimum size=22pt,opacity=0.9] at (2*\s,2*\s) {\color{red}{$-1$}};
    \end{scope}
    \draw[very thick, ->] (3, \s) -- node[above] {\scriptsize$W_{\gamma_1}W_{\gamma_2}$}(4,\s);
    \begin{scope}[shift = {(5,0)}]
        \draw[very thick] (0,-0.5) -- (0,\L);
        \draw[very thick] (\s,-0.5) -- (\s,\L);
        \draw[very thick] (2*\s,-0.5) -- (2*\s,\L);
        \draw[very thick] (-0.5,0) -- (\L,0);
        \draw[very thick] (-0.5,\s) -- (\L,\s);
        \draw[very thick] (-0.5,2*\s) -- (\L,2*\s);
        \node at (0.5,0.2) {\tiny$\ket{0}$};
        \node at (0.5,0.2+\s) {\tiny$\ket{0}$};
        \node at (0.5,0.2+2*\s) {\tiny$\ket{0}$};

        \node at (0.5+\s,0.2) {\tiny$\ket{0}$};
        \node at (0.5+\s,0.2+\s) {\tiny$\ket{0}$};
        \node at (0.5+\s,0.2+2*\s) {\tiny$\ket{0}$};

        \node at (-0.2,0.5+\s) {\tiny$\ket{0}$};
        \node at (-0.2,0.5) {\tiny$\ket{0}$};

        \node at (-0.2+\s,0.5+\s) {\tiny$\ket{0}$};
        \node at (-0.2+\s,0.5) {\tiny$\ket{0}$};

        \node at (0.2+2*\s,0.5+\s) {\tiny$\ket{0}$};
        \node at (0.2+2*\s,0.5) {\tiny$\ket{0}$};

        \draw[line width = 2pt, color=lightblue] (0,0) -- (0,\L-0.5);
        \draw[line width = 2pt, color=lightblue] (0,\s) -- (\L-0.5,\s);
        \draw[line width = 2pt, color=lightblue] (\s,0) -- (\L-0.5,0);
        \draw[line width = 2pt, color=lightblue] (2*\s,\s) -- (2*\s,-0.5);
        \draw[line width = 2pt, color=lightblue] (\s,\s) -- (\s,2*\s);
        \draw[line width = 2pt, color=lightblue] (\s,2*\s) -- (2*\s,2*\s);
        \draw[line width = 2pt, color=lightblue] (2*\s,2*\s) -- (2*\s,2*\s+0.5);
        \draw[line width = 2pt, color=lightblue] (0,2*\s) -- (-0.5,2*\s);
        
    \end{scope}
    \end{scope}
    
    }
    $
    \caption{The electric flux on the lattice (here represented in 2 dimensions, with links oriented upwards/to the right) is excited from the strong-coupling ground state $\bigotimes_{l\in\mathcal{L}}\protect\ket{0}_l$ along the spanning tree (marked in blue) by an error in the set \labelcref{eq:correctable_errors_on_R}. A measurement of all Gauss' law constraints $\mathcal{C}_v$ reveals the background charges marked in red. Choosing a pair of negative and positive background charges and connecting them on the spanning tree yields paths $\gamma_1$ and $\gamma_2$. The error is then corrected by applying the Wilson line operators $W_{\gamma_1}W_{\gamma_2}$. Note that this product is independent of the choice of which negative charges to connect to which positive ones.}
    \label{fig:electric_flux_lines}
\end{figure}

To instead correct arbitrary single $U$-errors as in \cref{prop:correctable_U_errors}, the recovery is:
\begin{enumerate}
\itemsep0em
    \item measure the constraints $\mathcal{C}_v$,
    \item if two adjacent sites have background charges $q_v = m$ and $q_{v'} = -m$, apply $(U_l^m)^\dagger$ on the connecting link $l=[v, v']$.
\end{enumerate}
This is illustrated in \cref{fig:U_errors}.
Since this protocol is ambiguous as soon as there are background charges on more than two adjacent sites, it can only appropriately correct errors on single links.

\begin{figure}[h!]
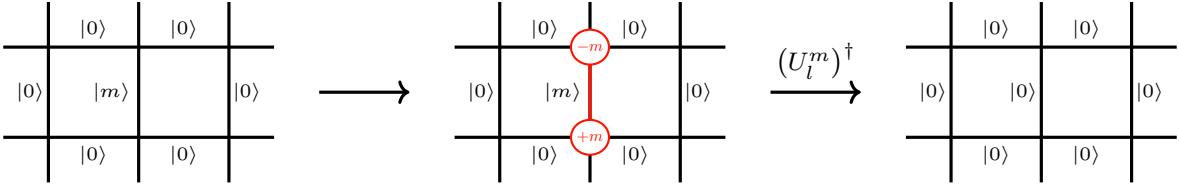

    \centering
        $\centertikz{
            \def\s{1}
    \def\L{2.5}
    \begin{scope}[scale=1.2,transform shape]
        \begin{scope}[shift = {(-5,0)}]
            \draw[very thick] (0,-0.5) -- (0,\L-\s);
            \draw[very thick] (\s,-0.5) -- (\s,\L-\s);
            \draw[very thick] (2*\s,-0.5) -- (2*\s,\L-\s);
            \draw[very thick] (-0.5,0) -- (\L,0);
            \draw[very thick] (-0.5,\s) -- (\L,\s);
            \node at (0.5,-0.2) {\tiny$\ket{0}$};
            \node at (0.5,0.2+\s) {\tiny$\ket{0}$};
            
            \node at (0.5+\s,-0.2) {\tiny$\ket{0}$};
            \node at (0.5+\s,0.2+\s) {\tiny$\ket{0}$};
    
            \node at (-0.2,0.5) {\tiny$\ket{0}$};
            \node at (-0.3+\s,0.5) {\tiny$\ket{m}$};
            \node at (0.2+2*\s,0.5) {\tiny$\ket{0}$};
        \end{scope}
    \draw[very thick, ->] (-2, 0.5*\s) -- (-1,0.5*\s);
    \begin{scope}[shift = {(0,0)}]
            \draw[very thick] (0,-0.5) -- (0,\L-\s);
            \draw[very thick] (\s,-0.5) -- (\s,\L-\s);
            \draw[very thick, color = myred] (\s,0) -- (\s,\s);
            \draw[very thick] (2*\s,-0.5) -- (2*\s,\L-\s);
            \draw[very thick] (-0.5,0) -- (\L,0);
            \draw[very thick] (-0.5,\s) -- (\L,\s);
            \node at (0.5,-0.2) {\tiny$\ket{0}$};
            \node at (0.5,0.2+\s) {\tiny$\ket{0}$};
            
            \node at (0.5+\s,-0.2) {\tiny$\ket{0}$};
            \node at (0.5+\s,0.2+\s) {\tiny$\ket{0}$};
    
            \node at (-0.2,0.5) {\tiny$\ket{0}$};
            \node at (-0.3+\s,0.5) {\tiny$\ket{m}$};
            \node at (0.2+2*\s,0.5) {\tiny$\ket{0}$};
            \node[thick, circ, draw = myred, fill=white,scale=0.5, inner sep=0.2pt, minimum size=22pt, font=\small] at (\s,\s) {\color{myred}{$-m$}};
            \node[thick, circ, draw = myred, fill=white,scale=0.5, minimum size=22pt, inner sep=0.2pt, font=\small] at (\s,0) {\color{myred}{$+m$}};
    \end{scope}
    \draw[very thick, ->] (3, 0.5*\s) -- node[above] {\scriptsize$\left(U_l^m\right)^\dagger$}(4,0.5*\s);
    \begin{scope}[shift = {(5,0)}]
        \draw[very thick] (0,-0.5) -- (0,\L-\s);
            \draw[very thick] (\s,-0.5) -- (\s,\L-\s);
            \draw[very thick] (2*\s,-0.5) -- (2*\s,\L-\s);
            \draw[very thick] (-0.5,0) -- (\L,0);
            \draw[very thick] (-0.5,\s) -- (\L,\s);
            \node at (0.5,-0.2) {\tiny$\ket{0}$};
            \node at (0.5,0.2+\s) {\tiny$\ket{0}$};
            
            \node at (0.5+\s,-0.2) {\tiny$\ket{0}$};
            \node at (0.5+\s,0.2+\s) {\tiny$\ket{0}$};
    
            \node at (-0.2,0.5) {\tiny$\ket{0}$};
            \node at (-0.2+\s,0.5) {\tiny$\ket{0}$};
            \node at (0.2+2*\s,0.5) {\tiny$\ket{0}$};
    \end{scope}
    \end{scope}
    
    }$
    \caption{An error $U_l^m$ acting on the strong-coupling ground state creates $m$ units of electric flux on $l = [v, v']$ (middle, oriented upwards). A measurement of the constraints results in $\mathcal{C}_v = m$ and $\mathcal{C}_{v'} = -m$ (marked in red) and $0$ elsewhere, and the error is corrected with $(U_l^m)^\dagger$.}
    \label{fig:U_errors}
\end{figure}

\subsection{Lattice QED with Fermionic Matter as a QECC}
\label{sec:fermionQECC}
Including staggered fermions into the model, the kinematical space contains both quantum rotors on links and qubits on sites. 
The resulting QECC is therefore no longer a pure rotor code and must be described more generally.
The stabilizer group is $\mathcal{G}$, generated by the constraints $\mathcal{C}_v = \mathcal{C}_v^{\mathcal{L}} - \rho_v$. Its elements take the form
\begin{equation}
    G(\bm{\lambda}) = \bigotimes_{v\in\mathcal{V}}e^{-i\lambda_v\rho_v}\otimes\bigotimes_{l = [v, v']\in\mathcal{L}}X_l(\lambda_{v'} - \lambda_v),
    \label{eq:local_stabilizer}
\end{equation}
where $\bm{\lambda} = \{\lambda_v\}_{v\in\mathcal{V}}$. Note that this is not a generalized Pauli stabilizer, because on the sites it acts as $e^{i\lambda_v \rho_v}= e^{i\lambda_v j_v}(\ket{0}\bra{0}_v + e^{-i\lambda_v}\ket{1}\bra{1}_v)$. This reduces to a Pauli operator only for $\lambda_v = 0$ (where it is $I_v$) or for $\lambda_v = \pi$ (where it is $(-1)^{|v|}Z_v$).

\subsubsection{The QECC Structure}

In contrast to the pure gauge case, the encoded logical degrees of freedom of this hybrid rotor-qubit code are no longer simply a collection of independent rotors. We have already found in \cref{ssec:QED_matter_QRF} that we can describe the gauge-invariant states purely in terms of the electric flux on all of their links through the reduction maps of the fermionic field QRF.

Since the reduced space $\HS_{\tilde{S}|\tilde{R}}\subset\HS_{\tilde{S}}$ is a proper subspace, the inverse reduction maps $(\mathcal{R}_{\tilde{R}}^{\bm{\lambda}})^\dagger$ are encoding maps only of the subspace $\HS_{\tilde{S}|\tilde{R}} \cong \HS_{\rm logical}$. As seen in \cref{prop:electric_basis}, this subspace has a basis given by the electric flux states $\bigotimes_{l\in\mathcal{L}}\ket{k_l}_l$ with eigenvalues $\mathcal{C}_v^{\mathcal{L}} \in \{-j_v, 1-j_v\}$. The logical operators of the QECC correspond to the gauge-invariant operators on $\HS_\phys$. As discussed in \cref{ssec:QED_matter_QRF}, these operators include the holonomy operators $H_\gamma$, the $X$-type operators $X_l(\lambda)$, the projectors $\psi_v^\dagger\psi_v, \, \psi_v\psi_v^\dagger$, and hopping terms $\psi_v^\dagger W_\gamma\psi_{v'}$. While these cannot be separated into $U$- and $X$-type logical operators on $\HS_\phys$, they reduce to operators that act on the basis $\bigotimes_{l\in\mathcal{L}}\ket{k_l}_l$ of $\HS_{\rm logical}$ either as $U$-, $X$-, or projector-type operators and products thereof. 

\paragraph{Errors on the links and on the sites.} Errors for staggered fermion lattice QED may occur on links and on sites. Only gauge-violating errors are detectable and thus potentially correctable.

For the errors that act on the links only, the discussion from the pure gauge case carries over: The $X$-type errors are gauge-invariant and thus not correctable, and only the $U$-type errors violate Gauss' law by creating electric flux lines. Correctable error sets of operators on the links thus sample from these $U$-type operators. 

On the sites, any operator can be expressed as linear combinations of $\psi_v$, $\psi_v^\dagger$, $\psi_v^\dagger\psi_v$ and $\psi_v\psi_v^\dagger$. Since the latter two operators are gauge-invariant, only errors that are linear combinations of the fermionic creation and annihilation operators violate Gauss' law. Such errors remove or add fermions onto the sites, creating unphysical charge configurations.

\subsubsection{Correctable Errors for Lattice QED with Fermionic Matter}
\label{sec:fermionError}
We now identify sets of correctable errors for lattice QED with staggered fermions. $U$-type errors are correctable in precisely the same manner as for pure gauge lattice QED.  
The fermionic field QRF yields sets of correctable errors acting on the sites via \cref{thm:correctable_gauge-fixing} and, due to \cref{prop:non_ideal_A_r}, a recovery operation based on coarse-grained charge-sector measurements. In addition, a coarse-grained measurement can also be used to correct an error set with errors acting on the sites or the links.

\paragraph{Correctable $U$-type errors.}
Both the correctable errors from \cref{prop:correctable_U_on_T_in_QED} and \cref{prop:correctable_U_errors} remain correctable. On the one hand, spanning trees provide QRFs for lattice QED if it also involves staggered fermions. On the other, staggered fermions do not affect the recovery of the error set \cref{prop:correctable_U_errors} based on measuring the constraints.

\paragraph{Correctable error sets associated to the fermionic field QRF.}
The fermionic QRF $\tilde{R}$ from \cref{thm:matter_qrf} yields a set of correctable errors acting on the sites by \cref{thm:correctable_gauge-fixing}. The orientation states 
\begin{equation}
    \ket{\bm{\lambda}}_{\tilde{R}} = \bigotimes_{v\in \mathcal{V}}\frac{1}{\sqrt{2}}(\ket{0}_v + e^{-i\lambda_v}\ket{1}_v) = \bigotimes_{v\in \mathcal{V}}\ket{\lambda_v}_v.
\end{equation}
are not all orthogonal, and so not the full set of gauge-fixing errors on the fermionic field QRF is correctable by the same recovery. However, choosing an angle $\alpha_v$ for every site $v$, the subset of orthogonal orientation states
\begin{equation}
    \left\{\ket{\bm{\lambda}}_{\tilde{R}} \, | \,  \forall v\in\mathcal{V}:\lambda_v\in\{\alpha_v, \alpha_v + \pi\}\right\}
\end{equation}
forms an orthonormal basis of $\HS_{\tilde{R}} = \bigotimes_{v\in\mathcal{V}}\mathbb{C}^2$.
The associated gauge-fixing operators are thus a correctable error set. These orientation states belong to the subset $\prod_{v\in\mathcal{V}}e^{i\alpha_v}\mathbb{Z}_2\subset\operatorname{U}(1)^{|\mathcal{V}|}$ which is a coset of the subgroup $\mathbb{Z}_2^{|\mathcal{V}|}$ (where $\mathbb{Z}_2 = \{\pm 1\}$ corresponds to the angles $0, \pi$). By \cref{prop:non_ideal_A_r}, we can thus find a set of correctable operators $A_{\bm{r}}$ which map $\HS_\phys$ into a coarse-grained charge sector $\HS_{\bm{r}}$.
Here, for $\bm{r}=(r_v)_{v\in\mathcal{V}}\in\{0, 1\}^{|\mathcal{V}|}$ and $\bm{\lambda}\in\{0, \pi\}^{|\mathcal{V}|}$, the restricted characters are $\chi_{\bm{r}}(\bm{\lambda})=\prod_{v\in\mathcal{V}}e^{-i\lambda_v r_v}$. We thus find the correctable operators
\begin{equation}
    A_{\bm{r}} = \frac{1}{2^{|\mathcal{V}|/2}}\sum_{\bm{\lambda}\in \{0, \pi\}^{|\mathcal{V}|}}\chi_{\bm{r}}(\bm{\lambda})\big(\ket{\bm{\lambda}+\bm{\alpha}}\bra{\bm{\lambda}+\bm{\alpha}}_{\tilde{R}}\otimes I_{\tilde{S}}\big).
\end{equation}
In each individual term, this operator on a single site is either $\ket{\alpha_v}\bra{\alpha_v}_v + \ket{\alpha_v +\pi}\bra{\alpha_v+\pi}_v = I_v$ or $\ket{\alpha_v}\bra{\alpha_v}_v - \ket{\alpha_v +\pi}\bra{\alpha_v+\pi}_v = e^{i\alpha_v}\psi_v + e^{-i\alpha_v}\psi_v^\dagger$. Defining $A_v(\alpha_v) \defeq e^{i\alpha_v}\psi_v + e^{-i\alpha_v}\psi_v^\dagger$,
we obtain
\begin{equation}
    A_{\bm{r}} = \bigotimes_{v\in\mathcal{V}}A_v(\alpha_v)^{r_v}.
\end{equation}
We thus find the following result.
\begin{proposition}
\label{prop:correctable_matter_errors}
    Let $\alpha_v\in[0, 2\pi)$ for all $v\in\mathcal{V}$. Then, the operators
    \begin{equation}
    \label{eq:correctable_A_errors}
        \bigotimes_{v\in\mathcal{V}} A_v(\alpha_v)^{r_v}, \quad r_v \in \{0, 1\}
    \end{equation}
    constitute a set of correctable errors.
\end{proposition}
For concreteness, let us fix $\alpha_{v} = 0$ for all $v$. Then, the operator $A_v(0) = \psi_v  +\psi_v^\dagger$ is simply a Pauli-$X$ in the number basis. This special case leads to the correctable error set 
    \begin{equation}
        \left\{\bigotimes\nolimits_{v\in\mathcal{V}}X_v^{r_v}\ \middle| \ \forall v\in\mathcal{V}:r_v\in\{0, 1\}\right\}.
    \end{equation}
Similarly, if $\alpha_v = \pi/2$, then $A_v(\pi/2) = -Y_v$. We could thus also choose to correct Pauli-$Y$ instead of Pauli-$X$ errors on some sites $v$.
Note however that $\psi_v$ and $\psi_v^\dagger$ by themselves are not correctable errors because they have non-trivial kernel on $\HS_\phys$ and do not satisfy the Knill-Laflamme conditions \labelcref{eq:Knill_Laflamme}. 

\paragraph{Recovery operation for the $A_v(\alpha_v)$-errors.} The recovery operation of the operators in \cref{eq:correctable_A_errors}, following \cref{sec:QEC_in_Gauge}, is based on a coarse-grained charge measurement. By construction in \cref{prop:non_ideal_A_r}, $A_{\bm{r}}$ maps $\HS_\phys$ into
\begin{equation}
    \HS_{\bm{r}} = \bigoplus_{\bm{q}\rightarrow\bm{r}}\HS_{\bm{q}}.
\end{equation}
Here, the direct sum goes over all $\bm{q}$ such that the characters satisfy $\chi_{\bm{r}}(\bm{\lambda}) = \chi_{\bm{q}}(\bm{\lambda}) $ for all $ \bm{\lambda}\in\{0, \pi\}^{|\mathcal{V}|}$, i.e., over all $\bm{q}$ such that $q_v$ is even if $r_v=0$ and $q_v$ is odd if $r_v = 1$.
The associated coarse-grained measurement thus projects onto these sectors.

More concretely, a single error $A_v(\alpha_v)$ maps a state $\ket{\psi}\in\HS_\phys$ to $A_v(\alpha_v)\ket{\psi}\in \HS_{\bm{q}}\oplus\HS_{-\bm{q}}$ where $q_{v'} = \delta_{vv'}$, i.e., into a superposition of charge $\pm1$ at $v$.
The coarse-grained measurement thus determines that this state is in $\HS_{\bm{r}}$ where $r_{v'} = \delta_{vv'}$, and we apply $A_v(\alpha_v)^\dagger$ to correct the error.

Finally, let us remark on the implications of the non-orthogonality of the orientation states of the fermionic field QRF for the correctable error sets in \cref{prop:correctable_matter_errors}. The recovery described above relies on a choice of the angle $\alpha_v$ on every site. This specifies two orthogonal states on $v$ and ultimately which operator the recovery corrects on each site. By a different choice of the angle on $v$, say $\beta_v$, a different operator $A_v(\beta_v)$ on $v$ becomes correctable. This does not mean that the operators $A_v(\alpha_v)$ for any angles and their products form a correctable error set, because the recovery operations are different for different angles $\alpha_v$.

\paragraph{Combining electric flux and $A_v(\alpha_v)$-error correction.}

We found in \cref{prop:correctable_U_errors} that any single $U_l^m$-error is correctable by a measurement of the constraints. This procedure can be combined with the one above to correct $A_v(\alpha_v)$-errors to correct both types of errors:
\begin{enumerate}
\itemsep0em
    \item perform a coarse-grained measurement with the projectors $\Pi_{\bm{r}}$,
    \item if the outcome corresponds to $\bm{r}$ with $r_{v'}=\delta_{vv'}$, apply $A_v(\alpha_v)^\dagger$,
    \item else, measure the charge sectors with projectors $\Pi_{\bm{q}}$, i.e., the constraints $\mathcal{C}_v$,
    \item if $\mathcal{C}_v=m$ and $\mathcal{C}_{v'}=-m$ on adjacent sites connected by $l = [v,v']$ and the constraints vanish elsewhere, apply $(U_{l}^m)^\dagger$.
\end{enumerate}
This proves the theorem below.

\begin{theorem}
\label{thm:correctable_errors_full_QED}
    Consider lattice QED including staggered fermions. On $\HS_\phys$, any single $U_l^m$- or $A_v(\alpha_v)$-error is correctable. That is, 
    \begin{equation}
    \label{eq:correctable_errors_full_QED}
        \left\{U_l^m  \, | \, l\in\mathcal{L}, m\in\mathbb{Z}\right\}\cup\left\{A_v(\alpha_v) \, | \,  v\in\mathcal{V}, \alpha_v\in[0, 2\pi)\right\}
    \end{equation}
    is a correctable error set.
\end{theorem}
Instead of relying on the recovery operation described above as a proof, one can check the Knill-Laflamme conditions for this set directly. 

Theorem 1 in \cite{spagnoli2024} is the analogue of this statement for lattice QED with a truncated local $\mathbb{Z}_2$-gauge group. In this case, both the link and site degrees of freedom are essentially qubits so that Gauss' law operators become qubit Pauli stabilizers. There are thus no constraints to measure and we have to resort to syndrome measurements of the Pauli stabilizers. On the sites, these act as Pauli-$Z$ and thus only anti-commute with $X_v$ (as opposed to $A_v(\alpha_v)$).
Also, there is only one possible electric flux error per link (represented by a qubit Pauli operator). This means that in the truncated case, any single electric flux error or single $X$-error on a site is correctable.

\paragraph{Visualizing the recovery protocols.}
The correctable error sets from \cref{prop:correctable_matter_errors} and from \cref{thm:correctable_errors_full_QED} are associated with recovery protocols that start with a coarse-grained charge-sector measurement and differ in the choice of which error to correct for a given outcome. 

To illustrate the errors on the fermionic field, we consider again the strong-coupling limit of the Hamiltonian \labelcref{eq:full_hamiltonian},
\begin{equation}
    \begin{aligned}
        H_{m} \approx &\sum_{v\in\mathcal{V}}(-1)^{|v|}m\psi^\dagger_v\psi_v +\frac{i}{2a}\sum_{v\in\mathcal{V}, \ i} \eta_{v,i}(\psi_v^\dagger U_{[v, v+e_i]}\psi_{v+e_i} - \psi_{v}^\dagger U_{[v-e_i, v]}^\dagger\psi_{v-e_i}) \\
        & + \frac{g^2}{2a}\sum_{l\in\mathcal{L}}\epsilon_l^2.
    \end{aligned}
\end{equation}
The ground state in this regime is $\ket{\Omega_{m,\text{sc}}} = \bigotimes_{l\in\mathcal{L}}\ket{0}_l \bigotimes_{v\in\mathcal{V}}\ket{j_v}_v$, where $j_v$ indicates the staggered vacuum configuration (since $\rho_v\ket{j_v}_v=0$).
Applying an error $A_v(\alpha_v)$ yields
\begin{equation}
    A_v(\alpha_v)\ket{\Omega_{m,\text{sc}}} = e^{(-1)^{|v|+1}i\alpha_v}\bigotimes_{l\in\mathcal{L}}\ket{0}_l \bigotimes_{v' \in\mathcal{V}}\ket{j_{v'} + (-1)^{|v|}\delta_{v,v'}}_{v'}.
\end{equation}
The occupation number on $v$ is thus flipped, creating a particle (or anti-particle) without adjusting the electric flux.
Similarly, an error $\bigotimes_{v\in V}A_v(\alpha_v)$ for $V\subset\mathcal{V}$ adds particles and antiparticles on all the sites in $V$ to the strong-coupling ground state. 
More generally, if the configuration is a superposition of occupation numbers, $A_v(\alpha_v)$ flips this number on $v$ and adds a relative phase, i.e., it exchanges occupied and unoccupied states on $v$ while adding a relative phase.

The recovery protocol of the errors in \cref{prop:correctable_matter_errors} starts by choosing an angle $\alpha_v$ per site (e.g., $\alpha_v = 0$ to correct Pauli $X$-errors). After the coarse-grained measurement results in the outcome $\bm{r}$, an operator $A_v(\alpha_v)^\dagger$ is applied to every site $v$ where $r_v = 1$ (see \cref{fig:X_errors}). 
In this way, any number of $A_v(\alpha_v)$-errors on the sites can be corrected.

To correct both single $A_v(\alpha_v)$- and $U_l^m$-errors, we again fix $\alpha_v$ for all sites and begin with the coarse-grained measurement. Compared to the previous recovery, we now interpret the outcome differently and follow the procedure outlined above \cref{thm:correctable_errors_full_QED}, as illustrated in \cref{fig:X_and_U_errors}.
The key difference is that matter errors produce a single-site violation, whereas electric flux errors produce a pair of opposite violations at adjacent sites. This difference is exploited in the recovery protocol.

\begin{figure}[h]
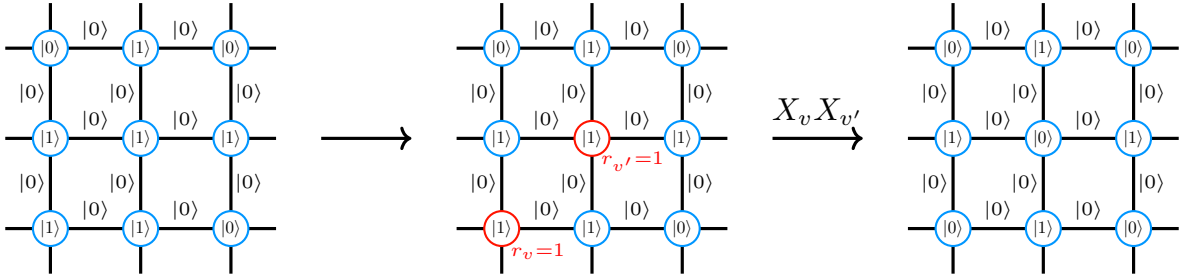

    \centering
    $\centertikz{
            \def\s{1}
    \def\L{2.5}
    \begin{scope}[scale=1.2,transform shape]
        \begin{scope}[shift = {(-5,0)}]
        \draw[very thick] (0,-0.5) -- (0,\L);
        \draw[very thick] (\s,-0.5) -- (\s,\L);
        \draw[very thick] (2*\s,-0.5) -- (2*\s,\L);
        \draw[very thick] (-0.5,0) -- (\L,0);
        \draw[very thick] (-0.5,\s) -- (\L,\s);
        \draw[very thick] (-0.5,2*\s) -- (\L,2*\s);
        \node at (0.5,0.2) {\tiny$\ket{0}$};
        \node at (0.5,0.2+\s) {\tiny$\ket{0}$};
        \node at (0.5,0.2+2*\s) {\tiny$\ket{0}$};

        \node at (0.5+\s,0.2) {\tiny$\ket{0}$};
        \node at (0.5+\s,0.2+\s) {\tiny$\ket{0}$};
        \node at (0.5+\s,0.2+2*\s) {\tiny$\ket{0}$};

        \node at (-0.2,0.5+\s) {\tiny$\ket{0}$};
        \node at (-0.2,0.5) {\tiny$\ket{0}$};

        \node at (-0.2+\s,0.5+\s) {\tiny$\ket{0}$};
        \node at (-0.2+\s,0.5) {\tiny$\ket{0}$};

        \node at (0.2+2*\s,0.5+\s) {\tiny$\ket{0}$};
        \node at (0.2+2*\s,0.5) {\tiny$\ket{0}$};
        \node[thick, circ, fill=white,scale=0.5, minimum size=22pt, draw=lightblue] at (0,0) {$\ket{1}$};
        \node[thick, circ, fill=white,scale=0.5, minimum size=22pt, draw=lightblue] at (0,\s) {$\ket{1}$};
        \node[thick, circ, fill=white,scale=0.5, minimum size=22pt, draw=lightblue] at (0,2*\s) {$\ket{0}$};

        \node[thick, circ, fill=white,scale=0.5, minimum size=22pt, draw=lightblue] at (\s,0) {$\ket{1}$};
        \node[thick, circ, fill=white,scale=0.5, minimum size=22pt, draw=lightblue] at (\s,\s) {$\ket{1}$};
        \node[thick, circ, fill=white,scale=0.5, minimum size=22pt, draw=lightblue] at (\s,2*\s) {$\ket{1}$};

        \node[thick, circ, fill=white,scale=0.5, minimum size=22pt, draw=lightblue] at (2*\s,0) {$\ket{0}$};
        \node[thick, circ, fill=white,scale=0.5, minimum size=22pt, draw=lightblue] at (2*\s,\s) {$\ket{1}$};
        \node[thick, circ, fill=white,scale=0.5, minimum size=22pt, draw=lightblue] at (2*\s,2*\s) {$\ket{0}$};

    \end{scope}
    \draw[very thick, ->] (-2, \s) -- (-1,\s);
    \begin{scope}[shift = {(0,0)}]
        \draw[very thick] (0,-0.5) -- (0,\L);
        \draw[very thick] (\s,-0.5) -- (\s,\L);
        \draw[very thick] (2*\s,-0.5) -- (2*\s,\L);
        \draw[very thick] (-0.5,0) -- (\L,0);
        \draw[very thick] (-0.5,\s) -- (\L,\s);
        \draw[very thick] (-0.5,2*\s) -- (\L,2*\s);
        \node at (0.5,0.2) {\tiny$\ket{0}$};
        \node at (0.5,0.2+\s) {\tiny$\ket{0}$};
        \node at (0.5,0.2+2*\s) {\tiny$\ket{0}$};

        \node at (0.5+\s,0.2) {\tiny$\ket{0}$};
        \node at (0.5+\s,0.2+\s) {\tiny$\ket{0}$};
        \node at (0.5+\s,0.2+2*\s) {\tiny$\ket{0}$};

        \node at (-0.2,0.5+\s) {\tiny$\ket{0}$};
        \node at (-0.2,0.5) {\tiny$\ket{0}$};

        \node at (-0.2+\s,0.5+\s) {\tiny$\ket{0}$};
        \node at (-0.2+\s,0.5) {\tiny$\ket{0}$};

        \node at (0.2+2*\s,0.5+\s) {\tiny$\ket{0}$};
        \node at (0.2+2*\s,0.5) {\tiny$\ket{0}$};

        \node[thick, circ, fill=white,scale=0.5, minimum size=22pt, draw=myred] at (0,0) {$\ket{1}$};
        \node at (0.4,-0.25) {\tiny\color{red}{$r_v\mathord{=}1$}};
        \node[thick, circ, fill=white,scale=0.5, minimum size=22pt, draw=lightblue] at (0,\s) {$\ket{1}$};
        \node[thick, circ, fill=white,scale=0.5, minimum size=22pt, draw=lightblue] at (0,2*\s) {$\ket{0}$};

        \node[thick, circ, fill=white,scale=0.5, minimum size=22pt, draw=lightblue] at (\s,0) {$\ket{1}$};
        \node[thick, circ, fill=white,scale=0.5, minimum size=22pt,draw=myred] at (\s,\s) {$\ket{1}$};
        \node at (\s+0.45,\s-0.25) {\tiny\color{red}{$r_{v'}\mathord{=}1$}};
        \node[thick, circ, fill=white,scale=0.5, minimum size=22pt, draw=lightblue] at (\s,2*\s) {$\ket{1}$};

        \node[thick, circ, fill=white,scale=0.5, minimum size=22pt, draw=lightblue] at (2*\s,0) {$\ket{0}$};
        \node[thick, circ, fill=white,scale=0.5, minimum size=22pt, draw=lightblue] at (2*\s,\s) {$\ket{1}$};
        \node[thick, circ, fill=white,scale=0.5, minimum size=22pt, draw=lightblue] at (2*\s,2*\s) {$\ket{0}$};
    \end{scope}
    \draw[very thick, ->] (3, \s) -- node[above] {\small$X_v X_{v'}$}(4,\s);
    \begin{scope}[shift = {(5,0)}]
        \draw[very thick] (0,-0.5) -- (0,\L);
        \draw[very thick] (\s,-0.5) -- (\s,\L);
        \draw[very thick] (2*\s,-0.5) -- (2*\s,\L);
        \draw[very thick] (-0.5,0) -- (\L,0);
        \draw[very thick] (-0.5,\s) -- (\L,\s);
        \draw[very thick] (-0.5,2*\s) -- (\L,2*\s);
        \node at (0.5,0.2) {\tiny$\ket{0}$};
        \node at (0.5,0.2+\s) {\tiny$\ket{0}$};
        \node at (0.5,0.2+2*\s) {\tiny$\ket{0}$};

        \node at (0.5+\s,0.2) {\tiny$\ket{0}$};
        \node at (0.5+\s,0.2+\s) {\tiny$\ket{0}$};
        \node at (0.5+\s,0.2+2*\s) {\tiny$\ket{0}$};

        \node at (-0.2,0.5+\s) {\tiny$\ket{0}$};
        \node at (-0.2,0.5) {\tiny$\ket{0}$};

        \node at (-0.2+\s,0.5+\s) {\tiny$\ket{0}$};
        \node at (-0.2+\s,0.5) {\tiny$\ket{0}$};

        \node at (0.2+2*\s,0.5+\s) {\tiny$\ket{0}$};
        \node at (0.2+2*\s,0.5) {\tiny$\ket{0}$};
        \node[thick, circ, fill=white,scale=0.5, minimum size=22pt, draw=lightblue] at (0,0) {$\ket{0}$};
        \node[thick, circ, fill=white,scale=0.5, minimum size=22pt, draw=lightblue] at (0,\s) {$\ket{1}$};
        \node[thick, circ, fill=white,scale=0.5, minimum size=22pt, draw=lightblue] at (0,2*\s) {$\ket{0}$};

        \node[thick, circ, fill=white,scale=0.5, minimum size=22pt, draw=lightblue] at (\s,0) {$\ket{1}$};
        \node[thick, circ, fill=white,scale=0.5, minimum size=22pt, draw=lightblue] at (\s,\s) {$\ket{0}$};
        \node[thick, circ, fill=white,scale=0.5, minimum size=22pt, draw=lightblue] at (\s,2*\s) {$\ket{1}$};

        \node[thick, circ, fill=white,scale=0.5, minimum size=22pt, draw=lightblue] at (2*\s,0) {$\ket{0}$};
        \node[thick, circ, fill=white,scale=0.5, minimum size=22pt, draw=lightblue] at (2*\s,\s) {$\ket{1}$};
        \node[thick, circ, fill=white,scale=0.5, minimum size=22pt, draw=lightblue] at (2*\s,2*\s) {$\ket{0}$};
    \end{scope}
    \end{scope}
    
    }$
    \caption{The fermionic field QRF lives on the sites of the lattice (marked in blue). We assume that the top left site is even. To illustrate \cref{prop:correctable_matter_errors}, two even sites $v$ (bottom left) and $v'$ (middle) are excited from the strong-coupling vacuum by a gauge-violating error $X_v X_{v'}$. A coarse-grained measurement reveals $r_v = r_{v'} = 1$ and $\bm{r}$ vanishing elsewhere (marked in red), and the error is appropriately corrected by acting on the lattice with $X_v X_{v'}$.}
    \label{fig:X_errors}
\end{figure}

\begin{figure}[h]
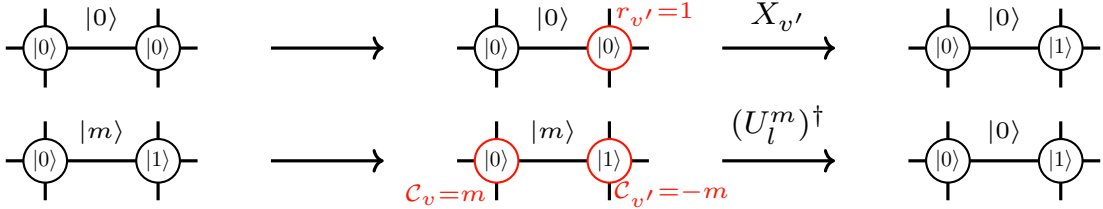

    \centering
   $\centertikz{
    \def\s{1}
    \def\L{2.5}
    \begin{scope}[scale=1.5,transform shape]
        \begin{scope}[shift = {(-4,0)}]
            \draw[very thick] (0,-0.35) -- (0,0.35);    
            \draw[very thick] (\s,-0.35) -- (\s,+0.35);
            \draw[very thick] (-0.35,0) -- node[above] {\tiny$\ket{m}$}  (\s+0.35,0);
            
            \draw[very thick] (0,-0.35+\s) -- (0,0.35+\s);    
            \draw[very thick] (\s,-0.35+\s) -- (\s,+0.35+\s);
            \draw[very thick] (-0.35,0+\s) -- node[above] {\tiny$\ket{0}$} (\s+0.35,0+\s);
            \node[thick, circ, fill=white,scale=0.5, minimum size=22pt] at (0,0) {$\ket{0}$};
            \node[thick, circ, fill=white,scale=0.5, minimum size=22pt] at (0,\s) {$\ket{0}$};
            \node[thick, circ, fill=white,scale=0.5, minimum size=22pt] at (\s,0) {$\ket{1}$};
            \node[thick, circ, fill=white,scale=0.5, minimum size=22pt] at (\s,\s) {$\ket{0}$};
        \end{scope}
    \draw[very thick, ->] (-2, \s) -- (-1,\s);
    \draw[very thick, ->] (-2, 0) -- (-1,0);
    \begin{scope}[shift = {(0,0)}]
            \draw[very thick] (0,-0.35) -- (0,0.35);    
            \draw[very thick] (\s,-0.35) -- (\s,+0.35);
            \draw[very thick] (-0.35,0) -- node[above] {\tiny$\ket{m}$}  (\s+0.35,0);
            
            \draw[very thick] (0,-0.35+\s) -- (0,0.35+\s);    
            \draw[very thick] (\s,-0.35+\s) -- (\s,+0.35+\s);
            \draw[very thick] (-0.35,0+\s) -- node[above] {\tiny$\ket{0}$} (\s+0.35,0+\s);
            \node[thick, circ, fill=white,scale=0.5, minimum size=22pt,draw=myred] at (0,0) {$\ket{0}$};
            \node[thick, circ, fill=white,scale=0.5, minimum size=22pt] at (0,\s) {$\ket{0}$};
            \node[thick, circ, fill=white,scale=0.5, minimum size=22pt,draw=myred] at (\s,0) {$\ket{1}$};
            
            \node[color=myred] at (\s+0.55,-0.3) {\tiny$\mathcal{C}_{v'} \mathord{=}\mathord{-}m$};
            \node[color=myred] at (\s+0.4,\s+0.3) {\tiny$r_{v'} \mathord{=}1$};
            \node[color=myred] at (-0.45,-0.3) {\tiny$\mathcal{C}_v \mathord{=}m$};
            \node[thick, circ, fill=white,scale=0.5, minimum size=22pt,draw=myred] at (\s,\s) {$\ket{0}$};
    \end{scope}
    \draw[very thick, ->] (2, \s) -- node[above] {\scriptsize$X_{v'}$} (3,\s);
    \draw[very thick, ->] (2, 0) -- node[above] {\scriptsize$(U_l^m)^\dagger$} (3,0);
    \begin{scope}[shift = {(4,0)}]
            \draw[very thick] (0,-0.35) -- (0,0.35);    
            \draw[very thick] (\s,-0.35) -- (\s,+0.35);
            \draw[very thick] (-0.35,0) -- node[above] {\tiny$\ket{0}$}  (\s+0.35,0);
            
            \draw[very thick] (0,-0.35+\s) -- (0,0.35+\s);    
            \draw[very thick] (\s,-0.35+\s) -- (\s,+0.35+\s);
            \draw[very thick] (-0.35,0+\s) -- node[above] {\tiny$\ket{0}$} (\s+0.35,0+\s);
            \node[thick, circ, fill=white,scale=0.5, minimum size=22pt] at (0,0) {$\ket{0}$};
            \node[thick, circ, fill=white,scale=0.5, minimum size=22pt] at (0,\s) {$\ket{0}$};
            \node[thick, circ, fill=white,scale=0.5, minimum size=22pt] at (\s,0) {$\ket{1}$};
            \node[thick, circ, fill=white,scale=0.5, minimum size=22pt] at (\s,\s) {$\ket{1}$};
    \end{scope}
    \end{scope}
   
   }$
    \caption{To illustrate \cref{thm:correctable_errors_full_QED}, we consider two different errors. (Top) A gauge-violating error excites the strong-coupling vacuum to $|0\rangle_v|0\rangle_l|0\rangle_{v'}\otimes|\phi\rangle_{\text{rest}}$, where we take $v$ to be an even site. Now, a coarse-grained measurement yields $r_{v'} = 1$ and $0$ elsewhere, and the error is corrected with $X_{v'}$. (Bottom) A different gauge-violating error excites the strong-coupling vacuum to $|0\rangle_v|m\rangle_l|1\rangle_{v'}\otimes|\phi\rangle_{\text{rest}}$ (where $v$ is an even site). A coarse-grained measurement results in $r_v = r_{v'} = 0, 1$, depending on the parity of $m$. Both outcomes do not correspond to a single $X_v$-error. Therefore, a subsequent constraint measurement shows $\mathcal{C}_v = m$, $\mathcal{C}_{v'} = -m$ (marked in red). This outcome is interpreted as the error $U_l^m$ and thus corrected with $(U_l^m)^\dagger$.}
    \label{fig:X_and_U_errors}
\end{figure}

\section{The Continuum Counterparts of the Quantum Reference Frames and Correctable Errors}
\label{sec:continuum}
To build further physical intuition for the QRFs and the correctable error sets identified in lattice QED, we now turn to their continuum counterparts. The goal of this section is not to cast continuum QED as a quantum error-correcting code; such a formulation would require different tools and a careful treatment of subtleties unique to continuum field theory. Rather, we aim to identify the continuum analogues of the QRFs and correctable error structures found in the lattice theory, which we expect to hold, in an appropriate sense, in the continuum limit.

We begin by noting a basic structural distinction between the lattice and continuum settings: in continuum QED, the group of gauge transformations is no longer locally compact. Consequently, the construction of QRFs in \cite{delahamette_2021} is not directly applicable, and one must instead employ field-theoretic techniques.
However, the physical intuition underlying QRFs remains clear. The spanning-tree QRF introduced in \cref{ssec:spanning_tree_QRFs} identifies a gauge-fixable subsystem (with respect to $\mathcal{H}_{\kin}$) of the gauge field by connecting any two vertices via a unique path. While there is no direct analogue of a spanning tree in the continuum, a conceptually similar construction does exist --- namely, the contour gauge \cite{Anikin2025}. The continuum limit of a Wilson line, i.e., a product of $U_l$ operators  $W_\gamma = \prod_{l\in\gamma}U_l^{\sigma_l} = e^{i\sum_{l\in\gamma}\sigma_l\Theta_l}$ (recall that $\sigma_l$ denotes wether $\gamma$ passes $l$ in its positive or negative orientation) is the open Wilson line:
\begin{equation}
    \hat{W}_\gamma = e^{ig\int_\gamma \D x^{\mu} \cdot \hat{A}_{\mu}}.
\end{equation}
Here, we will denote the operators in the continuum with a hat to distinguish from the lattice operators and to adhere to standard literature.
If we choose a root $v_0$ of a spanning tree, any site $v$ is connected to $v_0$ by a unique path $\gamma({v_0, v})$ along the tree.  Gauge-fixing the orientation of the tree to $\ket{\phi(\bm{0})}_R$ corresponds to setting all of the Wilson lines to identity, $W_{\gamma({v_0, v})} = I$, because $W_{\gamma({v_0, v})}\big(\ket{\phi(\bm{0})}_R\otimes\ket{\psi}_S\big) = \ket{\phi(\bm{0})}_R\otimes\ket{\psi}_S$.  This is similar to the contour gauge in the continuum, where we fix a path $\gamma(\vec{x}_0, \vec{x})$ for any point $\vec{x}$, connecting it to a fixed point $\vec{x}_0$ and subsequently gauge-fix the Wilson line $\hat{W}_{\gamma{(\vec{x}_0, \vec{x})}}=I$.

The operators in the correctable set $\{U_l^m\}_{m\in\mathbb{Z},l\in\mathcal{L}}$ are exponentials of $\Theta_l$ which corresponds to the vector potential (see \cref{app:continuum_limit}). 
Since the $l$ goes to zero in the continuum limit, to avoid singular operator, we can introduce smeared functions. This  is commonly used to define a proper notion of local operators, as local field operators are formally understood as operator-valued distribution rather than point-wise defined operators \cite{Fewster:2019ixc}.   More precisely, we choose a smooth compactly supported test vector field with components $l^i_{\sigma}(\vec{x})$  approximating an oriented link $l$ of the lattice case for each small region $\sigma \in \Sigma$ on the spacial hypersurface $\Sigma$. The size of $\sigma$  is set by the experimental resolution, representing the finite spatial resolution of the observable.
The continuum limit of the correctable set  is then 
\begin{equation}\label{eq:fieldUerror}
\{U_l^m\}_{l\in\mathcal{L}} \rightarrow \big\{  e^{i m g \int_{\sigma} \D^3\vec{x}   \hat{A}_i(\vec{x}) l^i_{\sigma}(\vec{x}) }    \big\}_{\sigma \in \Sigma}.
\end{equation}
The construction extends straightforwardly to the correctable error sets generated by products of Wilson line-type operators.

To discuss the continuum limit of fermionic field QRF and the associated correctable error sets, we must first understand the continuum limit of the staggered fermion construction. This is not trivial precisely because of the staggered nature of the lattice fermions. In $3+1$ dimensions, the continuum limit yields two species of fermions, commonly referred to as {\it tastes}. 
This can be formulated in the taste basis, in which the 8 degrees of freedom located at the corners of a cube are grouped into a matrix field $\Psi_w$ for all $w=2v\in\mathcal{V}$  (see \cref{app:staggered_continuum}). In the continuum limit, this construction recovers the Hamiltonian of two independent Dirac fields \cite{CPS_2025}.
A single Dirac field in the continuum admits the expansion in two locally independent modes,
\begin{equation}
\hat{\psi}(\vec{x}) = \sum_{s=a,b} \int\frac{\D^3\vec{p}}{(2\pi)^3 \sqrt{2E_p}} \left(\hat{b}_s(\vec{p})u_s(\vec{p}) e^{-i \vec{p}\cdot \vec{x}} + \hat{d}^{\dagger}_s(\vec{p})v_s(\vec{p}) e^{i \vec{p}\cdot \vec{x}}  \right), 
\end{equation}
where $s=a,b$ labels the two spin polarization modes. Here $\hat{b}^{\dagger}_s, \hat{b}_s$ are the creation/annihilation operators for a particle with polarization s while $\hat{d}^{\dagger}_s, \hat{d}_s$ are  the creation/annihilation operators for an anti-particle;  $u_s, v_s$ are the associated Dirac spinors. In the continuum Dirac field, the analog of the lattice site occupation operator is the local charge density $\hat{j}^0(\vec{x}) = \hat{\psi}^{\dagger}(\vec{x})  \hat{\psi}(\vec{x})$. Using projectors $P_s$ onto the two independent spin modes,  we define the corresponding local mode density operator as 
\begin{equation}
\hat{n}_s(\vec{x}) =  \hat{\psi}^{\dagger}(\vec{x}) P_s \hat{\psi}(\vec{x}), \qquad \sum_s\hat{n}_s (\vec{x})= \hat{j}^0(\vec{x})
\end{equation}
Let $\{|n_s\rangle \}$  denote the field basis  on an equal-time hypersurface  \cite{Hatfield:1992rz}, which diagonalizes the local mode density operator $\hat{n}_s(\vec{x})$ as: 
\begin{equation}
\hat{n}_s(\vec{x}) |n_s\rangle = n_s(\vec{x})|n_s\rangle\, .
\end{equation}
Note that the local density $n_s(\vec{x})$ is not point-wise binary occupation variable as in the lattice case. 
The orientation state analogous  to the one for the lattice fermionic field QRF in \cref{eq:matter_qrf} can be written formally as
\begin{equation}\label{eq:fermionfieldQRF}
|\lambda_s\rangle = \int \mathcal{D} [n_s]\ e^{-i \int \D^3\vec{x} \lambda(\vec{x}) \hat{n}_s(\vec{x})} |n_s\rangle,\quad s= a,b,
\end{equation}
where $\lambda(\vec{x}) \in [0,2\pi)$ is the local $\operatorname{U}(1)$ phase. The local gauge transformation acts with the same phase on both modes, and it is easy to check that $\operatorname{U}(1)$ acts transitively and freely on these orientation states through the transformation $e^{-\frac{i}{g}\int \D \vec{x} \ \lambda(\vec{x})\hat{n}_s(\vec{x})}$. Accordingly, the two independent spin modes $s= a,b$ may be regarded as two QRFs which, together with the electromagnetic field and the Gauss constraint, form the perspective-neutral description. The properties of these non-ideal frames, the corresponding reduced observables in the continuum field-theoretic setting, as well as perspectival transformations between these two non-ideal QRFs are not the focus of the present paper, and we will leave them to future investigation. 

To identify the continuum analogue of the fermionic correctable error set, it is useful to use the taste-basis description reviewed in \cref{app:staggered_continuum}. The fermionic correctable errors found on the lattice correspond to local flips of the staggered occupation number accompanied by a phase. In \cref{app:staggered_continuum}, we show that for each coarse lattice site $w = 2v' \in\mathcal{V}$, this leads to the family
\begin{equation}\label{eq:errormatrix}
    e^{i\theta_w}{\Psi}_w + e^{-i\theta_w}{\Psi}_w^\dagger,\quad \theta_w \in [0, 2\pi)
\end{equation}
where ${\Psi}_w$ is Dirac matrix field whose columns encode two Dirac fermions. The operator in \cref{eq:errormatrix} is essentially a combination of the operators $A_v(\alpha_v)$ (from \cref{prop:correctable_matter_errors}) on the individual sites $v$ on the cube around $w$ with appropriate angles $\alpha_v$.    In the continuum limit, this yields  $e^{i\theta(\vec{x})}\Psi(\vec{x})+ e^{-i\theta(\vec{x})}\Psi(\vec{x})^\dagger$, where $\Psi = (\psi_1, \psi_2)$ contains the two independent Dirac fields. This is a local Hermitian linear combination of Dirac field annihilation and creation operators, analogous to (a phase-rotated) Majorana Fermion operator \cite{Weinberg:1995mt}.  

Since the field operators $\hat{\Psi} (\vec{x})$ are operator-valued distributions, the local fermionic error operators are more properly defined in the smeared form \cite{Fewster:2019ixc}. 
More precisely, we choose a smooth compactly supported test functions $f_i(\vec{x})$ for each small spatial region $\sigma_i$.
This is also the operationally relevant notion of locality, as any realistic disturbance or measurement has finite spatial resolution. The corresponding local error operator associated with  $\sigma_i$ is then
\begin{equation}\label{eq:fermionfielderror}
\hat{F}_{\sigma_i} := \int_{\sigma_i} \D^3\vec{x} f_i(\vec{x}) \left(e^{i\theta_i}\hat{\Psi}(\vec{x})+ e^{-i\theta_i}\hat{\Psi}^\dagger(\vec{x})\right).
\end{equation}
The smeared operator $\hat{F}_{\sigma_i}$ therefore describes a local fermionic error supported in the small region $\sigma_i$, while finite products of such operators over different small regions $\sigma_i$ provide the continuum counterpart of multi-site lattice errors:
\begin{equation}
    \left\{\prod\nolimits_{i\in I} \hat{F}_{\sigma_i} \ \middle| \ \forall i: \theta_i\in [0, 2\pi) \right\}.
\end{equation}
While the smearing of field operators addresses a first step towards defining local error operators  in the continuum, a full QECC treatment of continuum QED raises several important challenges. Such treatment would require a notion of subsystem structure that no longer relies on the tensor-product decomposition of the Kinematical Hilbert space available on the lattice \footnote{More precisely, in  quantum field theory one does not in general have a natural tensor-product factorisation into strictly local subsystems. Under the \emph{split property}, for two suitably nested spacetime regions one can insert an intermediate type I factor $\mathcal{N}$ between the corresponding  von Neumann algebras of observables  \cite{fewster2016splitpropertyquantumfield, Doplicher:1984zz}. This yields an approximate subsystem decomposition, with a separating small ``buffer zone'' between the regions.}. In the continuum field theory, locality  is instead encoded by operator algebras associated with spacetime subregions \cite{Haag:1992hx,fewster2016splitpropertyquantumfield}.  These algebras are typically of type III \cite{Fredenhagen:1984dc} and do not give rise to local Hilbert-space tensor factors in the same way as lattice systems. In addition, one would need a field-theoretic formulation of syndrome measurements through local measurement schemes in quantum field theory \cite{Fewster:2018qbm}, together with a corresponding notion of recovery map, possibly formulated directly in the operator-algebraic level.

\section{Discussion and Outlook}
\label{sec:conclusion}
In this work, we developed a quantum error-correction construction of lattice quantum electrodynamics using the perspective-neutral framework of quantum reference frames. Our aim was not only to explore the structural analogy between gauge symmetry and quantum error correction \cite{sem_proj, CCHM_2024} (namely, that the physical information in a gauge theory can be viewed as logical information, redundantly encoded in a larger space through gauge constraints) but also to make precise how gauge symmetry defines this encoding of physical information, and how QRFs can be used to identify correctable errors and construct recovery operations in lattice gauge theory. As a result, lattice QED emerges as a quantum error-correcting code beyond the stabilizer setting.

At a general level, \cref{thm:correctable_gauge-fixing} extends a result of \cite{CCHM_2024} which identifies orthogonal gauge-fixing operators as correctable errors to a broader class of gauge systems with compact gauge groups, including cases where such operators arise from non-ideal QRFs.
For Abelian gauge groups, this result can be further refined using group-theoretical methods. In particular, \cref{prop:A_q_from_gf} relates these gauge-fixing errors to explicit recovery procedures via charge-sector measurements for ideal QRFs. \Cref{prop:non_ideal_A_r} extends this to the non-ideal setting: When the QRF contains an orthonormal family of orientation states associated with a subgroup, it determines not only a correctable set, but also the corresponding recovery map via coarse-grained charge-sector measurements.

Applied to lattice QED, this yields two distinct  QECC structures:
the quantum rotor code for pure gauge sector, and  a hybrid rotor-qubit code when we include fermionic matter.
For both of these codes, we constructed QRFs and obtained the reduced description relative to the chosen frames, thereby making explicit how the physical information is represented in the corresponding encoding.
In the pure-gauge sector, the spanning tree QRFs provide ideal frames and lead to correctable families of electric-flux errors. In the presence of  fermions, we obtained a non-ideal QRF from the fermionic matter field which gave rise to correctable local occupation-flip errors with a relative phases. In both cases, these correctable errors arose as linear combinations of orthogonal gauge-fixing operators via \cref{prop:A_q_from_gf} and \cref{prop:non_ideal_A_r}.

The role of QRFs in this QECC construction can be viewed as a way of resolving the degeneracy of the error syndrome.
In general, an error in a gauge system may violate the gauge constraints, but this  violation alone is insufficient to determine the error uniquely; as in standard QEC, the error syndrome  is generically degenerate.
For Abelian gauge groups, an ideal QRF removes this ambiguity by restricting to errors supported on the chosen frame. Thereby it identifies a class of errors for which the syndrome, obtained from charge-sector measurements, becomes sufficient for recovery. In lattice QED,
this is especially transparent  for non-local errors such as open Wilson lines in \cref{eq:WilsonlineError}: Different lines with the same endpoints map to the same charge sector, while a spanning-tree QRF singles out a unique line connecting any two points.
For the non-ideal fermionic field QRF, the local occupation-flip errors in \cref{sec:fermionError} map to coarse-grained charge sectors. 
After a coarse-grained measurement identifies the affected sites, the remaining degeneracy is removed by fixing the relative phase between annihilation and creation operators,  specified by a choice of orthogonal orientation states of the non-ideal QRF. 
In this sense, QRFs make explicit how gauge redundancy can be turned into an error-correcting structure. 

Beyond the QRF-related correctable error sets, charge-sector measurements identify further correctable families in lattice QED (\cref{prop:correctable_U_errors} and \cref{thm:correctable_errors_full_QED}), and thereby connected our work to previous results for lattice QED with a truncated gauge group in \cite{Rajput2023, spagnoli2024}. More broadly, recovery protocols based on charge-sector measurements provide a concrete operational procedure of error correction from gauge symmetry, with direct relevance for fault-tolerant quantum simulations of lattice gauge theories \cite{Stryker_2019, GL_2023, Rajput2023, spagnoli2024, CLLL_2024, yao_2025, pato_2026}.

Several interesting directions and questions arise from our construction and results:

\begin{itemize}
    \item \textbf{Non-Abelian theories:} An important next step is to extend the present analysis beyond Abelian theories. In our construction, the Abelian structure of lattice QED was crucial for relating gauge-fixing errors to recovery via charge-sector measurements in \cref{prop:A_q_from_gf,prop:non_ideal_A_r}. Generalizing this connection to non-Abelian gauge groups would further tie QRFs to recovery protocols by charge-sector measurements for a broad range of gauge systems. Determining the error correction properties of non-Abelian gauge theories is also of practical interest \cite{yao_2025}.
    \item \textbf{Non-local QRFs for lattice QED:} The QRFs we found in \cref{sec:QRFs_in_LQED} are local with respect to the tensor product structure of the kinematical space. Consequently, the correctable error sets found through their gauge-fixing operators consist of operators supported on these local QRFs. Constructing non-local refactorizations into frame and system would lead to correctable sets of errors potentially spread over all the lattice. The 1-to-1 relation between non-local ideal QRFs and correctable error sets for stabilizer codes \cite{CCHM_2024} suggests that also the results in \cref{prop:correctable_U_errors} and \cref{thm:correctable_errors_full_QED} may be connected to such non-local QRFs.
    \item \textbf{Continuum gauge field theory:} \Cref{sec:continuum} has already suggested the corresponding continuum counterparts of QRFs and correctable errors. A full QECC treatment of continuum QED, however,  will require genuinely new tools. In particular, since the group of gauge transformation of continuum QED is no longer locally compact, the construction in \cite{delahamette_2021} is not directly applicable and therefore requires a new  field-theoretic formulation of quantum reference frames.  Moreover, one must formulate the subsystem structure without relying on the tensor-product decomposition available on the lattice, but instead work directly with local operator algebras; correspondingly, syndrome extraction and recovery would be formulated directly in the operator-algebraic setting.  A better understanding of these steps will also open the way to analogous constructions for linearized gravity, whose four first-class constraints generate linearized diffeomorphism gauge symmetry. 
    \item \textbf{Further structural correspondence between QECC and gauge theories:} Our work suggests that gauge theories may support a richer notion of information robustness than has been uncovered so far.  One particularly intriguing question is whether intrinsic code properties, such as the code distance,  admit direct physical counterparts in gauge theory. In QECCs, the code distance characterizes the smallest size of an operator that can act non-trivially on the code subspace. In an earlier project \cite{sem_proj}, we identified a type of QRF for stabilizer codes that is recoverable from erasure, suggesting that operational notions of recoverability may already reflect structural features of the underlying gauge description. One interesting application of this could be in the correspondence between QEC and AdS/CFT \cite{Almheiri_2015, HaPPY_2015, PP_2017, Harlow_2017, Kibe2022}, where the concept of erasure of subsystems is central.
    Generally, exploring gauge systems with the tools of QEC could sharpen the meaning of gauge symmetry as an information-theoretic structure.
\end{itemize}

\section*{Acknowledgments}
We would especially thank Joe Renes for detailed discussion and insightful feedback. We would like to acknowledge Renato Renner, Ognyan Oreshkov,  Aidan Chatwin-Davies, Philipp Höhn, Anne-Catherine de la Hamette, Sebastian Garmier, Bei Zeng and Muxin Han for discussions.  We thank  Quantum Information Theory (QIT) group at ETH Zürich for insightful questions during E.R. Master thesis presentation. 
E.R. and C.F. acknowledge the support of the ETH Zürich Quantum Center and NCCR SwissMAP.
L.Q.C is funded by Österreichischer Wissenschaftsfonds FWF ESPRIT fellowship with grant DOI: 10.55776/ESP390. L.Q.C thanks the Quantum Information Theory (QIT) group at ETH Zürich for  hospitality during her visit, with where part of this work was carried out. She acknowledges support from the NCCR SwissMAP for the ETH visits.
This work was made possible through the support of the WithOut SpaceTime (WOST) project, supported by Grant ID\#~63683 from the John Templeton Foundation (JTF). The opinions expressed in this work are those of the author(s) and do not necessarily reflect the views of the John Templeton Foundation. 
This work also contributes to the European Union COST Action CA23115: \emph{Relativistic Quantum Information}, funded  by COST (European Cooperation in Science and Technology).

\sloppy
\printbibliography
\appendix
\crefalias{section}{appendix}
\crefalias{subsection}{appendix}
\newpage 

\section{Details on Perspective-Neutral Quantum Reference Frames}
\label{app:QRFs}

\paragraph{Perspective-neutral QRFs.}

For completeness, we will restate the definition of perspective-neutral QRFs as presented in \cite{delahamette_2021}. Recall from the main text that the starting point of the formalism is a system $\HS_\kin = \HS_R\otimes\HS_S$ which carries a representation $U_R(g)\otimes U_S(g)$ of a gauge group $G$.

\begin{definition}[Quantum reference frames]
    The subsystem $R$ is a QRF for the unimodular gauge group $G$ if it admits a set of \textit{orientation states} $\{\ket{\phi(g)}_R  \, | \,  g\in G\}$ on which $G$ acts transitively, 
    \begin{equation}
        U_R(g')\ket{\phi(g)}_R = \ket{\phi(g'g)}_R,
    \end{equation} and which can resolve identity,
    \begin{equation}
        \int_G \D g \ \ket{\phi(g)}\bra{\phi(g)}_R = cI_R
    \end{equation}
    for $c>0$.
\end{definition}

\paragraph{$\HS_\phys$ as a perspective-neutral space.}

Why is $\HS_\phys$ the central playground of the perspective-neutral construction? Demanding $G$-invariance of physical states in $\HS_\phys\subset\HS_{\text{kin}}$ means that these states are associated with full $G$-orbits of kinematical states. This can be seen from the projector $\Pi_\phys$,
\begin{equation}
    \ket{\psi}_\phys = \Pi_\phys\ket{\psi}_\kin = \frac{1}{|G|}\int_G \D g \ U_{RS}(g)\ket{\psi}_\kin.
\end{equation}
The key aspect of the construction is to associate the states in this orbit with the same physical state with respect to a different frame orientation. Additionally, the coherent group average removes any external frame information and what survives in $\HS_\phys$ are exactly the relational degrees of freedom. Due to these reasons, $\HS_\phys$ contains all the information about the physical states of the system with respect to any internal frame in any orientation. In other words, $\HS_\phys$ is a {\it perspective-neutral space}, and the physical states are often referred to as {\it relational}.

We can extract a specific perspective from such a perspective-neutral state by jumping into a QRF in a certain orientation. On a technical level, this requires reduction maps which act as isometries from $\HS_\phys$ into $\HS_S$ and allow recovery of the perspective-neutral states from the fixed frame-orientation ones. 

There are multiple constructions of such reduction maps. The one relevant in this work, as introduced in the main text, amounts to fixing the state of the frame to a certain orientation as a generalization of the Page-Wootters mechanism \cite{PageWooters_1983}. 

\begin{definition}[Reduction map]
\label{def:PW_reduction}
    The Page-Wootters type reduction map $\mathcal{R}_R^g:\HS_{\text{kin}}\rightarrow\HS_S$ is defined as
    \begin{equation}
        \mathcal{R}_R^g = \sqrt{N}\left(\bra{\phi(g)}_R \otimes I_S\right) \Pi_\phys,
    \end{equation}
    with normalization constant $N = \frac{|G|}{c}$.
\end{definition}
As a map restricted to $\HS_\phys$, $\mathcal{R}_R^g$ can be written without the projector $\Pi_\code$. The adjoint on $\HS_\phys$ contains the projector and is given by
\begin{equation}
    (\mathcal{R}_R^g)^\dagger = \sqrt{N}\Pi_\phys\left(\ket{\phi(g)}_R \otimes I_S\right).
\end{equation}

The procedure essentially amounts to fixing the state of the frame to a certain orientation. Intuitively, given a kinematical state $\ket{\psi}_\kin$, the physical state
\begin{equation}
    \Pi_\phys \ket{\psi}_\kin
\end{equation}
is a coherent superposition of all the states in the $G$-orbit of $\ket{\psi}_\kin$ associated with different frame orientations. Fixing the orientation of the frame thus singles out the corresponding state of the system. For an illustration, see \cref{fig:gauge-fixing_map}.

\begin{lemma}{\cite[Lemma 8]{delahamette_2021}}
\label{lemma:isometry}
    The reduction map $\mathcal{R}_R^g$ of a QRF $R$ is an isometry on $\HS_\phys$, i.e.,
    \begin{equation}
        (\mathcal{R}_R^g)^\dagger \mathcal{R}_R^g = \Pi_\phys.
    \end{equation}
\end{lemma}

\begin{proof}
    We will be making use of the left- and right-invariance of the measure which implies
    \begin{align}
        U_{RS}(h)\Pi_\phys &= \frac{1}{|G|}\int_G \D g \ U_{RS}(hg) \\
        &= \frac{1}{|G|}\int_G \D h^{-1}g \ U_{RS}(g) \\
        &= \Pi_\phys
        \label{eq:left_right_invariance_projector}
    \end{align}
    and similarly $\Pi_\phys U_{RS}(h) = \Pi_\phys$.
    The statement can now be proved by an explicit calculation:
    \begin{align}
        (\mathcal{R}_R^g)^\dagger \mathcal{R}_R^g &= N \Pi_\phys(\ket{\phi(g)}\bra{\phi(g)}_R\otimes I_S )\Pi_\phys \\
        &= N \Pi_\phys U_{RS}(hg^{-1})(\ket{\phi(g)}\bra{\phi(g)}_R\otimes I_S )U_{RS}^\dagger(hg^{-1})\Pi_\phys \\
        &= N \frac{1}{|G|}\int_G \D h \left(\Pi_\phys U_{RS}(hg^{-1})(\ket{\phi(g)}\bra{\phi(g)}_R\otimes I_S )U_{RS}^\dagger(hg^{-1})\Pi_\phys \right) \\
        &= \frac{1}{c}\int_G \D h \left(\Pi_\phys \ket{\phi(h)}\bra{\phi(h)}_R\otimes I_S )\Pi_\phys \right) \\
        &= \Pi_\phys (I_R\otimes I_S )\Pi_\phys \\
        &= \Pi_\phys.
    \end{align}
\end{proof}

\begin{figure} [h]
    \centering
    \includegraphics[width=0.9\linewidth]{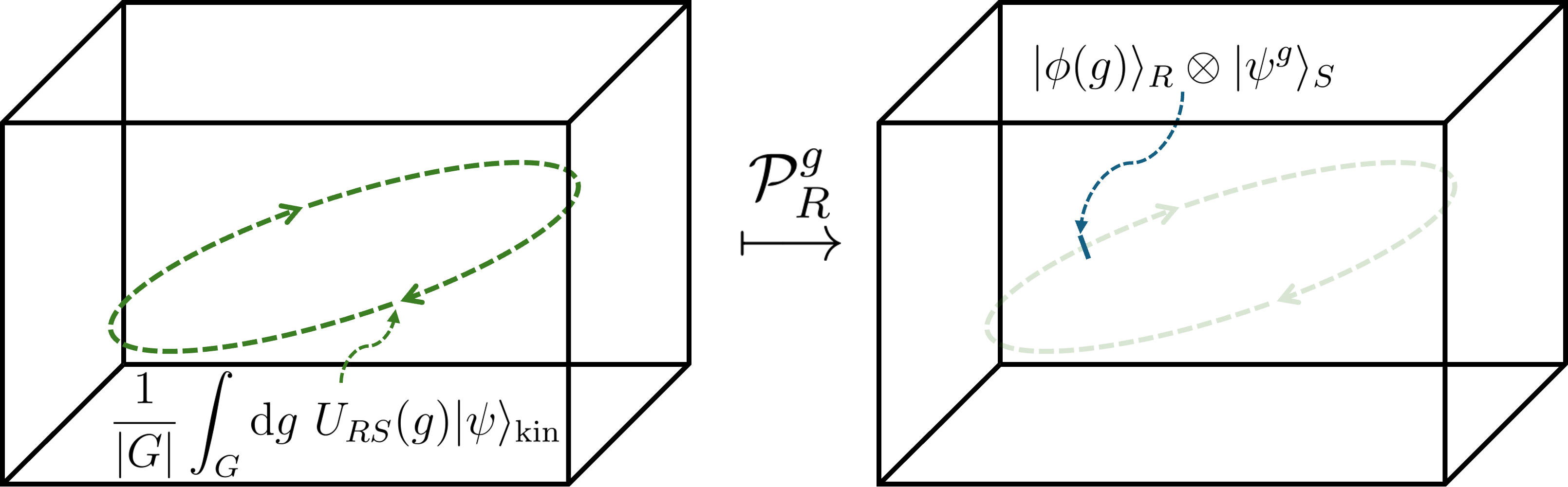}
    \caption{Physical states correspond to a coherent average over the $G$-orbit of kinematical states $|\psi\rangle_\kin$ (left). In the case of an ideal QRF, this orbit consists of states $|\phi(g)\rangle_R\otimes|\psi^g\rangle_S$ for all $g\in G$ and the gauge-fixing operator $\mathcal{P}^g_R = \sqrt{N}(|{\phi(g)}\rangle\langle{\phi(g)}|_R\otimes I_S)$ singles out the one corresponding to the orientation $g\in G$ (right). Accordingly, the reduction map extracts the state $|\psi^g\rangle_S$.}
    \label{fig:gauge-fixing_map}
\end{figure}

\section{Details on Quantum Error Correction}
\label{app:QEC}

\subsection{Quantum Error Correction Codes}

The goal of quantum error correction is to protect quantum information against possible errors, through redundantly encoding \textit{logical} quantum information in a subspace of a larger Hilbert space, called the \textit{code subspace}.
Such an encoding is constructed so that for a specified noise channel which takes code states outside the code subspace, there exists an operation to recover the logical information.
Generally, we define a quantum error correction code as follows:

\begin{definition}[Quantum error correction code]
    A quantum error correction code (QECC) is characterized by an \textit{encoding isometry} $T:\HS_{\log} \to\HS_{\text{physical}}$ from the logical space $\HS_{\log}$ to the physical space $\HS_{\text{physical}}$. Its image is the \textit{code subspace}  $\HS_\code \defeq T(\HS_{\log}) \subseteq \HS_{\phys}$, and the associated orthogonal projector is denoted as $\Pi_\code \defeq T T^{\dagger} : \HS_{\text{physical}} \mapsto \HS_\code$.
\end{definition}

Commonly, QECCs encode the logical information of a number $k$ of logical qubits into $n$ physical qubits which represent the physical device the data is stored on. These codes, for which $\HS_{\log}\cong (\mathbb{C}^2)^{\otimes k}$ and $\HS_{\physical}\cong (\mathbb{C}^2)^{\otimes n}$, are referred to as $[[n,k]]$-QECCs.

\paragraph{Knill-Laflamme conditions.}
Given a noise channel\footnote{One could also include measurements into the error model, in which case the noise operation does not necessarily preserve the trace. The equality condition can the be relaxed to proportionality.} $\mathcal{N}$, the code protects against the corresponding errors if there exists a recovery channel $\mathcal{R}: S(\HS_\physical)\to S(\HS_\physical)$ such that 
\begin{equation}
    \mathcal{R}\circ\mathcal{N}(\rho)=\rho \quad \forall\rho\in S(\HS_\code).
\end{equation}
In this case, the logical state can be perfectly recovered after the action of the noise.
Knill and Laflamme \cite{KnillLaflamme97} formulated the following necessary and sufficient conditions for error correction.
\begin{theorem}[Knill-Laflamme conditions]
\label{thm:KnillLaflamme}
    For a given a noise channel $\mathcal{N}$ with Kraus operators $\mathcal{E}=\{E_i\}_i$, i.e., $\mathcal{N}(\rho) = \sum_iE_i\rho E_i^\dagger$ for $\rho\in S(\HS_{\text{physical}})$,
    there exists a recovery channel $\mathcal{R}$ if and only if
    \begin{equation}
        \Pi_\code E_i^\dagger E_j \Pi_\code = c_{ij}\Pi_\code
    \end{equation}
    where $c_{ij}$ are entries in a hermitian matrix. We say that $\mathcal{E}$ is a correctable set of errors on $\HS_\code$.
\end{theorem}
By linearity, if these conditions are met for a set of Kraus operators $\mathcal{E}=\{E_i\}_i$, then the recovery operation $\mathcal{R}$ corrects any noise channel whose associated Kraus operators are linear combinations of $E_i$. 

\paragraph{The Pauli group.}
The Pauli operators on qubits are given by the matrices
\begin{equation}
    X = \begin{bmatrix}
        0 & 1 \\
        1 & 0
    \end{bmatrix}, \quad Z = \begin{bmatrix}
        1 & 0 \\
        0 & -1
    \end{bmatrix}, \quad Y = iXZ = \begin{bmatrix}
        0 & -i \\
        i & 0
    \end{bmatrix}.
\end{equation}
These operators are particularly significant for QEC because $\{I, X, Y, Z\}$ spans the set of all $\mathbb{C}^{2\times 2}$-matrices as a basis over $\mathbb{C}$. This extends to $\{I, X, Y, Z\}^{\otimes n}$, which spans all $n$-qubit matrices. 

Under multiplication, the $n$-qubit Pauli operators yield the {\it $n$-qubit Pauli group}
\begin{equation}
    \mathcal{P}_n = \left\{\alpha P_1\otimes P_2\otimes\dots\otimes P_n \, | \,  \alpha \in\{\pm1, \pm i\}, P_i \in \{X, Y, Z\}\right\}.
\end{equation}
The {\it weight} of a Pauli operator $P\in\mathcal{P}_n$ is $w(P) = |\supp(P)|$, where $\supp(P)$ are the non-identity tensor factors in $P$. 

\paragraph{Code distance.}
Since all errors can be decomposed into their Pauli contributions, if a code can correct the set of all Pauli operators up to weight $t$
\begin{equation}
    \mathcal{E} = \{P\in\{I, X, Y, Z\}^{\otimes n} \, | \, w(P) \leq t\},
\end{equation}
it can correct {\it any} error acting on $t$ or fewer qubits.

An important parameter of a QECC is the \textit{code distance} $d$, defined as the minimum weight of any Pauli operator $P$ such that
\begin{equation}
\label{eq:distance}
    \Pi_\code P \Pi_\code \not\propto \Pi_\code.
\end{equation}
The operator $P$ irrevocably changes a code state, i.e., there exist orthogonal states $\ket{\psi}, \ket{\phi}\in\HS_\code$ such that $\bra{\phi}P\ket{\psi}\neq0$.

Usually, $d$ is also included into the code parameters: A $[[n, k, d]]$-code encodes $k$ logical into $n$ physical qubits with a distance of $d$. Such a code can correct $t$ errors if $t\leq \frac{d-1}{2}$ \cite[Thm. 1]{KnillLaflamme00}.

\subsection{The Stabilizer Formalism}
An important family of QECCs are stabilizer codes \cite{gottesman1997stabilizer}. These codes define the code subspace by a symmetry group, the so-called stabilizer group. Owing to its algebraic structure, logical operations as well as error detection and correction can be described in a convenient way.

\begin{definition}[Stabilizer code]
    \label{def:stabilizer_code}
    A stabilizer code is a non-trivial subspace $\HS_\code \subset \HS_\physical$ which admits a Pauli subgroup $\mathcal{S}\subset \mathcal{P}_n$, called the \textit{stabilizer group}, such that
    \begin{equation}
       \ket{\psi}\in\HS_\code \; \iff \; S\ket{\psi} = \ket{\psi} \quad \forall S \in \mathcal{S}.
    \end{equation}
\end{definition}

The stabilizer group is necessarily Abelian and does not contain $-I$. Conversely, these two facts are sufficient to define a code space $\HS_\code$ as in the definition above. 

For a $[[n,k,d]]$-stabilizer code, the stabilizer group has $n-k$ independent generators,
\begin{equation}
    \mathcal{S} = \langle S_1, S_2, \dots, S_{n-k} \rangle
\end{equation}

\paragraph{Logical operators of a stabilizer code.} The logical operators of a QECC map between the code states while leaving the code space invariant. In the case of stabilizer codes, these logical operators are precisely the ones which commute with any stabilizer $S \in \mathcal{S}$. They can be decomposed into their Pauli components, which motivates the following definition.

\begin{definition}[Logical Pauli operators]
\label{def:logical_paulis}
    The \textit{logical Pauli operators} of a stabilizer code are
    \begin{equation}
        \Log(\HS_\code) = N(\mathcal{S})/\mathcal{S}.
    \end{equation}
    $N(\mathcal{S})$ denotes the normalizer of $\mathcal{S}$ in $\mathcal{P}_n$, i.e., $N(\mathcal{S}) = \{ P \in \mathcal{P}_n  \, | \,  P\mathcal{S} =\mathcal{S} P \}$.
    A logical operator $[P]\in\Log(\HS_\code)$ is thus an equivalence class of operators whose restriction onto the code space coincide.
\end{definition}
Since any two Pauli operators commute up to a sign, the normalizer coincides with the centralizer of $\mathcal{S}$, i.e., the subgroup of Pauli operators which commute with all $S\in\mathcal{S}$. 

For stabilizer codes, the code distance has a particularly simple form: it is the minimum weight of any (non-trivial) representative of a logical Pauli operator. That is,
\begin{equation}
    d = \min_{P\in N(\mathcal{S})\setminus\mathcal{S}} w(P).
\end{equation}

\paragraph{Syndrome measurement.} 
For stabilizer codes, error detection and correction can be described in a concise way by error syndrome measurements. Given a set of generators $S_i$ of $\mathcal{S}$, a Pauli error operator $E$ has the syndrome $s = (s_1, s_2, \dots, s_{n-k})\subset\{\pm 1\}^{n-k}$ if for $\ket{\psi}\in\HS_\code$
\begin{equation}
    S_i E \ket{\psi} = s_i E\ket{\psi}.
\end{equation}
Since any Pauli operators commute up to a sign, $s_i = \pm 1$. Non-trivial $s_i$ indicate that $E$ does not commute with the generator $S_i$. There are $2^{n-k}$ different syndromes, each corresponding to a different error subspace. A {\it syndrome measurement} consists of a measurement of each generator. 

A Pauli error set $\mathcal{E}$ is correctable if for each syndrome $s$ there exists an operator $E(s)$ such that applying $E(s)^\dagger$ to the error state after measuring $s$ recovers the original state. The condition is satisfied by $\mathcal{E}$ (which includes identity) if any $E_i, E_j \in \mathcal{E}$ either have different syndromes or $E_i^\dagger E_j \in \mathcal{S}$, in which case one can choose $E(s)$ as any $E_i\in\mathcal{E}$ which has the error syndrome $s$.
This protocol also corrects linear combinations of Pauli errors: The syndrome measurement collapses an error state $\sum_i\alpha_iE_i\ket{\psi}$ onto one of the $E_i\ket{\psi}$.

\subsection{Quantum Rotor Codes}
\label{sssec:rotor_codes}

So far, the systems we considered for error correction were all finite-dimensional. However, physical realizations of qubits (and qudits) on quantum platforms often involve infinitely many degrees of freedom of both the continuous and discrete variety. It is thus of practical interest to study error correction with infinite-dimensional systems. On one hand, one can consider bosonic encodings \cite{Terhal_2020} of qubits in subspaces of such systems, for example in a quantum oscillator \cite{GKP_2001}, in a molecule \cite{ACP_2020} or in a planar rotor (a $\operatorname{U}(1)$-system) \cite{RKSE_2010}. On the other hand, one can explore infinite-dimensional codes that encode infinite-dimensional systems (as opposed to finite ones). This idea is relevant in the context of covariant codes which are invariant under continuous symmetries \cite{Faist_2020, HNPS_2021}.

In this regard, Vuillot, Ciani and Terhal introduced a formalism of homological quantum rotor codes \cite{Vuillot_2024}. These codes generalize the construction of Calderbank-Shor-Steane (CSS) codes \cite{Calderbank_1996, Steane_1996}, a type of stabilizer codes, to quantum rotor systems.

\paragraph{Quantum rotors.}
The Hilbert space of a quantum planar rotor is $\HS_\rot = L^2(\operatorname{U}(1))$, the space of square-integrable functions on $\operatorname{U}(1) = \{e^{i\theta} \, | \,  \theta\in[0, 2\pi)\}$. The angular position $\theta$ of the rotor parametrizes the formal basis $\{\ket{e^{i\theta}}  \, | \,  \theta \in [0, 2\pi) \}$, normalized to $\braket{e^{i\theta}}{e^{i\phi}} = \delta(\theta-\phi)$. One can also express functions in $L^2(\operatorname{U}(1))$ with their Fourier coefficients in the Fourier basis $\{\ket{k} \}_{k\in\mathbb{Z}}$, where 
\begin{equation}
    \ket{k} = \int_0^{2\pi} \frac{\D\theta}{\sqrt{2\pi}}e^{ik\theta}\ket{e^{i\theta}}.
\end{equation} 

On $\HS_\rot$, the angular operator $\Theta$ and its conjugate momentum $\epsilon$ are of the form
\begin{equation}
    \Theta = \int_0^{2\pi} \D\theta \ \theta \ket{e^{i\theta}}\bra{e^{i\theta}}, \quad \epsilon = \sum_{k\in\mathbb{Z}} k \ket{k}\bra{k}.
\end{equation}
To avoid domain issues due to $\Theta$ not preserving periodicity, we can work with its exponential version $U = e^{i\Theta}$. This operator shifts angular momentum,
\begin{equation}
    U \ket{k} = \int_0^{2\pi} \frac{\D\theta}{\sqrt{2\pi}}e^{i(k+1)\theta}\ket{e^{i\theta}}
    = \ket{k+1}.
\end{equation} 
The angular momentum operator $\epsilon$ generates rotations and its exponential $X(\lambda) = e^{-i\lambda\epsilon}$ thus acts as an angle shift operator,
\begin{equation}
    X(\lambda) \ket{e^{i\theta}} =  \frac{1}{\sqrt{2\pi}}\sum_{k\in\mathbb{Z}}e^{-ik(\theta+\lambda)}\ket{k}
    =\ket{e^{i(\theta + \lambda)}}.   
\end{equation}
The operators $\epsilon$ and $U$ satisfy the canonical commutation relation $[\epsilon, U] = U$.
In exponential form, this leads to the $\operatorname{U}(1)$ braiding relation
\begin{equation}
\label{eq:U(1)_braiding}
    UX(\lambda) = e^{i\lambda}X(\lambda)U.
\end{equation}

\paragraph{The quantum rotor as a qudit limit for $D\rightarrow \infty$.} 
\label{app:rotor_qudit_limit}
The planar rotor can be seen as a continuous limit of a $D$-dimensional qudit as described in \cite{Albert_2017}. The Hilbert space of a qudit is $\mathbb{C}^D$, and in terms of the computational basis $\{\ket{j}\}_{j\in\mathbb{Z}_D}$, we can generalize the $X$ and $Z$-operators on qubits to
\begin{equation}
    X_D = \sum_{j=0}^{D-1}\ket{j+1\mod D}\bra{j}, \quad Z_D = \sum_{j=0}^{D-1}\omega_D^j\ket{j}\bra{j}
\end{equation}
where $\omega_D = e^{\frac{2\pi i}{D}}$.
We can interpret the computational basis as labelling the cyclic position states on a discrete circle. On the circle, we define the angle operator
\begin{equation}
    \phi\ket{j} = \frac{2\pi j}{D}\ket{j}.
\end{equation}
In the Fourier transformed basis $\ket{k} = \frac{1}{\sqrt{2\pi}}\sum_je^{i\frac{2\pi j}{D}k}\ket{j}$, we define the momentum operator
\begin{equation}
    \mu\ket{k} = k{\ket{k}}.
\end{equation}
The qudit Pauli operators are exponentials of these two, $X_D = e^{-i\frac{2\pi}{D}\mu}$ and $Z_D = e^{i\phi}$ and implement translations and momentum shifts, respectively. As we let the states on the circle go from discrete to continuous, the computational basis becomes the $\operatorname{U}(1)$-basis $\{\ket{e^{i\theta}}\}_{\theta\in[0, 2\pi)}$ and the discrete angle operator transitions towards the continuous one, $\phi \rightarrow \Theta$. Thus, we also have $Z_D \rightarrow U$. The momentum basis becomes the Fourier basis $\{\ket{k}\}_{k\in\mathbb{Z}}$ and we see that $\mu\rightarrow\epsilon$. Lastly, $X_D$ is replaced by the continuous translation operator $X(\lambda) = e^{-i\lambda\epsilon}$. The qudit relation $Z_DX_D = e^{\frac{2\pi i}{D}}X_DZ_D$ translates to the $\operatorname{U}(1)$ braiding relation \labelcref{eq:U(1)_braiding}. To summarize, we find
\begin{center}
    \begin{tabular}{c c c}
        $\mathbb{C}^D$ & $\overset{D\rightarrow\infty}{\longrightarrow}$ & $L^2(\operatorname{U}(1))$\\
        $\{\ket{j}\}_{j\in\mathbb{Z}_D}$ & $\rightarrow$ & $\{\ket{e^{i\theta}}\}_{\theta\in[0, 2\pi)}$\\
        $\ket{k} = \frac{1}{\sqrt{2\pi}}\sum_je^{i\frac{2\pi j}{D}k}\ket{j}$ & $\rightarrow$ & $\ket{k} =\int_0^{2\pi} \frac{\D\theta}{\sqrt{2\pi}}e^{ik\theta}\ket{e^{i\theta}}$\\
        $\phi$&$ \rightarrow$ & $\Theta$\\
        $\mu$ & $\rightarrow$ & $\epsilon$ \\
        $Z_D$ & $\rightarrow$ & $U$ \\
        $X_D$ & $\rightarrow$ & $X(\lambda)$
    \end{tabular}.
\end{center}

Note that it is also possible to take a limit where the spectrum of $Z$ becomes continuous and $X$ stays discrete. The roles are swapped by a Fourier transformation and this limit is equivalent to the one outlined above.

\paragraph{Homological rotor codes.}

Qubit (and qudit) CSS codes are stabilizer codes with the special property that the stabilizers are either pure $X$-type or pure $Z$-type operators. The generators are usually indicated using matrices. This is also the case for the rotor codes\footnote{In \cite{Vuillot_2024}, the opposite rotor limit to the one outlined above is taken. To follow our convention, we simply swap the roles of $X$ and $U$.} in \cite{Vuillot_2024}.

The stabilizers of the code space $\HS_\code\subset\HS_\physical = \HS_\rot^{\otimes n}$ are determined by full-rank matrices $H_X \in \mathbb{Z}^{r_x\times n}, H_U \in \mathbb{Z}^{r_u\times n}$ which satisfy $H_U H_X^T = 0$. The rows of these matrices indicate the $X$-type and $U$-type stabilizer generators
\begin{equation}
    S_{X,i} = \bigotimes_{j = 1}^{n}X(H_{X, ij}), \quad S_{U, i} = \bigotimes_{j = 1}^nU^{H_{U, ij}}.
\end{equation}
These generate the $X$-type stabilizer group $\mathcal{S}_X$ and the $U$-type stabilizer group $\mathcal{S}_U$. The full stabilizer group can be written as
\begin{equation}
    \mathcal{S} = \left\{S_X(\bm{\lambda})S_U(\bm{m})  \, | \,  \bm{\lambda}\in [0, 2\pi)^{r_x}, \bm{m}\in\mathbb{Z}^{r_u}\right\} = \mathcal{S}_X\times\mathcal{S}_U,
\end{equation}
where $S_X(\bm{\lambda}) = \prod_{i=1}^{r_x} S_{X,i}^{\lambda_i}$ and similarly $S_U(\bm{m}) = \prod_{i = 1}^{r_u} S_{U,i}^{m_i}$, and the condition $H_UH_X^T = 0$ ensures that $\mathcal{S}$ is Abelian.

The logical $U$-operators are
\begin{equation}
    \Log_U(\HS_\code) = \left.\left\{ A = \bigotimes\nolimits_{j=1}^n U_j^{m_j} \ \middle| \ [A, S_X] = 0 \ \forall S_X\in\mathcal{S}_X\right\} \middle/ \mathcal{S}_U\right.
\end{equation}
i.e., the quotient by $\mathcal{S}_U$ of the $U$-type operators $\bigotimes_{j=1}^n U^{m_j}$ which commute with $\mathcal{S}_X$. In general, $\HS_\code$ consists of $k$ logical rotors and $k'$ qudits of respective dimension $d_i$, such that $\Log_U(\HS_\code) \cong \mathbb{Z}^k\oplus\mathbb{Z}_{d_1}\oplus \dots \oplus\mathbb{Z}_{d_{k'}}$. The logical $X$-operators $\Log_X(\HS_\code)$ also match this decomposition. In \cite{Vuillot_2024}, this structure of the encoded data is derived from the homological properties of the code due to the matrices $H_U$ and $ H_X$.

\paragraph{Code distance.}
The $U$-distance of the code is the minimum weight of any representative of a non-trivial logical $U$-operator
\begin{equation}
    d_{U,w} = \min \{w_U(P) \, | \,  \ [I]\neq[P]\in \Log_U(\HS_\code) \}.
    \label{eq:U_distance}
\end{equation}
The weight function $w_U$ can be chosen differently depending on the context. Counting non-identity tensor factors, we get
\begin{equation}
    w_U(U^{m_1}\otimes U^{m_2}\otimes \dots \otimes U^{m_n}) = |\{m_i \neq 0\}|,
    \label{eq:weight_function_1}
\end{equation}
which generalizes the usual definition for the weight of qudit Pauli operators to the $\operatorname{U}(1)$-case. In \cite{Vuillot_2024}, the weight function
\begin{equation}
    \tilde{w}_U(U^{m_1}\otimes U^{m_2}\otimes \dots \otimes U^{m_n}) = \sum_{i=1}^n |m_i|
    \label{eq:weight_function_2}
\end{equation}
was used to account for a noise model which leads to an error $U^m$ with probability $p^m$.

The definition of the $X$-distance requires a more careful treatment due to the continuous nature of the parameter $\lambda$. Since this is not important for the present work, we refer to \cite{Vuillot_2024} for details.

\paragraph{Syndrome measurement.}
For an error set $\mathcal{E}$ consisting of strings of $U$- and $X$-type operators, we can find error syndromes using \cref{eq:U(1)_braiding}. For a generator $S_{X}$ of $\mathcal{S}_X$ of the form $S_{X} = e^{-\sum_jl_j\epsilon_j}$, the error state $U_j^m\ket{\psi}_\code$ is an eigenstate with eigenvalue $e^{-il_jm}$,
\begin{equation}
    S_{X}U_j^m\ket{\psi}_\code = e^{-il_jm}U_j^m\ket{\psi}_\code.
\end{equation}
For a generator $S_{U}$ of $\mathcal{S}_U$ of the form $S_{U} = \bigotimes_{j}U_j^{k_j}$, similarly, $X_j(\lambda)\ket{\psi}_\code$ is an eigenstate with eigenvalue $e^{ik_j\lambda}$,
\begin{equation}
    S_{U}X_j(\lambda)\ket{\psi}_\code = e^{ik_j\lambda}X_j(\lambda)\ket{\psi}_\code.
\end{equation}

Since the stabilizers are not self-adjoint, the syndrome needs to be measured through modified operators. For the $X$-type stabilizers, one can measure $\sum_{j=1}^n l_j\epsilon_j$ instead of $S_X = e^{-\sum_jl_j\epsilon_j}$, which results in an eigenvalue of $m$ instead of $e^{-il_jm}$. For the $U$-type stabilizers, a suitable choice is given by
\begin{equation}
    (S_{U,i} + S_{U,i}^\dagger)/2 \quad \text{and} \quad (S_{U,i} - S_{U,i}^\dagger)/(2i),
\end{equation}
such that the eigenvalues $\cos(l\lambda)$ and $\sin(-l\lambda)$ can be used to determine the syndrome $e^{-ilm}$.

In the $\operatorname{U}(1)$ case, as for qubits, the error set $\mathcal{E}$ satisfies the Knill-Laflamme conditions if two error operators $E_i, E_j \in \mathcal{E}$ either have different syndromes or combine to a stabilizer $E_i^\dagger E_j\in\mathcal{S}$. Even with infinite-dimensional spaces and potentially continuous error indices, this is sufficient for $\mathcal{E}$ being correctable on all of $\HS_\code$ (for details, see \cite[Section V]{BKK_2009_QEC_inf_dim}).

\section{Details on Lattice Quantum Electrodynamics}

\subsection{The Continuum Limit of the Pure Gauge Hamiltonian}
\label{app:continuum_limit}

Recall that the Kogut-Susskind Hamiltonian of the pure gauge system \labelcref{eq:KS_ham_pg} is
\begin{equation}
    H = \frac{g^2}{2a}\sum_{l\in\mathcal{L}}\epsilon_l^2 - \frac{1}{ag^2}\sum_{p_{ij}(v)\in\mathcal{P}} \cos\Theta_{ij}(v),
    \label{eq:KS_ham_pg_1}
\end{equation}
This form is related to QED in the following way: The angular operators and momentum operators at the link $l=[v,v+e_i]$ take the role of the $i$-th component of respectively the electromagnetic potential $A_i$ and the electric field $E_i$ in the position $x_v$ associated with the vertex $v$, i.e.,
\begin{equation}
\label{eq:discreteoperatorstofields}
    \Theta_{[v, v+e_i]} = agA_i(x_v), \quad \epsilon_{[v, v+e_i]} = \frac{a^2}{g}E_i(x_v).
\end{equation}
The discrete curl of the angular operators is identified as the magnetic field perpendicular to it, i.e.,
\begin{equation}
    \Theta_{jk}(v) = a^2gB_i(x_v), \quad  \Theta_{ki}(v) = a^2gB_j(x_v), \quad  \Theta_{ij}(v) = a^2gB_k(x_v) .
\end{equation}

To see that both the Hamiltonian \labelcref{eq:KS_ham_pg_1} and these identifications actually describe QED, we consider the continuum limit \cite{KS1975}. We thus let $a\rightarrow 0$ and replace the sum by an integral $\sum \rightarrow \int\frac{\D^3 \vec{x}}{a^3}$. Then, we obtain
\begin{equation}
    \frac{g^2}{2a}\sum_{l\in\mathcal{L}}  \epsilon^2_l \rightarrow \frac{1}{2}\int \D^3 \vec{x} \ \vec{E}^2(x),
\end{equation}
and
\begin{equation}
    -\frac{1}{ag^2}\sum_{p_{ij}(v)\in\mathcal{P}}\cos\Theta_{ij}(v) \rightarrow \frac{1}{2}\int \D^3 \vec{x} \ \vec{B}^2(\vec{x}) + \textup{const}.
\end{equation}
Up to an additive constant, we recover the Hamiltonian of the electromagnetic field.

\subsection{Gauss's law and Gauge Transformations}
\label{app:gauss_law}

\paragraph{Gauss's law generates gauge transformations.}
On the lattice, the gauge transformation $\psi(\vec{x})\mapsto e^{-i\lambda(\vec{x})}\psi(\vec{x})$ translates straightforwardly to $\psi_v \mapsto e^{-i\lambda_v}\psi_v$.
By \cref{eq:discreteoperatorstofields}, discretizing $A_i(\vec{x}) \mapsto A_i(\vec{x}) + \frac{1}{g}\partial_i\lambda(\vec{x})$ yields
\begin{equation}
    \frac{1}{ag}\Theta_{[v, v+e_i]}\mapsto \frac{1}{ag}\Theta_{[v, v+e_i]} + \frac{1}{g}\frac{\Delta_{v, v+e_i}\lambda}{a},
\end{equation}
which corresponds to \cref{eq:discrete_gauge_transformation_analogue}.

These transformations are indeed generated by the constraint $\mathcal{C}_v$ in \cref{eq:gauss_law_constraint}.
Due to the anti-commutation relation $\{\psi_v, \psi_v^\dagger\} = 1$ and $\psi_v\psi_v = 0$, we find
\begin{equation}
    [\rho_v, \psi_v] = [\psi_v^\dagger\psi_v, \psi_v] = -\psi_v.
\end{equation}
This exponentiates to the relation
\begin{equation}
    e^{-i\lambda_v\mathcal{C}_v}\psi_ve^{i\lambda_v\mathcal{C}_v} = e^{i\lambda_v\rho_v}\psi_v e^{-i\lambda_v\rho_v} = e^{-i\lambda_v}\psi_v.
\end{equation}
The gauge transformation $\prod_{v'\in\mathcal{V}} e^{i\lambda_{v'}\mathcal{C}_{v'}}$ acts on a link state $\ket{e^{i\theta}}_{[v, v+e_i]}$ as 
\begin{equation}
    \prod_{v'\in\mathcal{V}} e^{i\lambda_{v'}\mathcal{C}_{v'}}\ket{e^{i\theta}}_{[v, v+e_i]} = \ket{e^{i(\theta + \Delta_{v, v+e_i}\lambda)}}_{[v, v+e_i]}.
\end{equation}
Since this is an eigenstate of $\Theta_{[v, v+e_i]}$, we find
\begin{equation}
    \prod_{v'\in\mathcal{V}} e^{-i\lambda_{v'}\mathcal{C}_{v'}} \circ \Theta_{[v, v+e_i]} \circ \prod_{v'\in\mathcal{V}} e^{i\lambda_{v'}\mathcal{C}_{v'}} = \Theta_{[v, v+e_i]} + \Delta_{v, v+e_i}\lambda.
\end{equation}

\paragraph{The geometrical picture.}
The gauge structure on the lattice can be interpreted geometrically \cite{Kogut1979}: The local phase of matter fields on each site of the lattice is well-defined only with respect to a reference frame setting the unit-phase. Since one cannot make sense of the phase on two sites without relating their local frames, the link operators $U_{[v, v']}$ indicate the difference of the frames on their adjacent sites $v$ and $v'$. In other words, the gauge field ``mediates'' the description of such degrees of freedom by containing the information about the relative orientation of the local frames. The local references are arbitrary, thus each vertex $v$'s frame can be independently rotated by an angle $\lambda_v$ without changing the physical content, leading to a local $\operatorname{U}(1)$-gauge transformation $\psi_v\mapsto e^{-i\lambda_v}\psi_v$. Accordingly, the relative angle $\Theta_{[v,v']}$  is transformed by the difference $\Delta_{v, v'} \lambda = \lambda_{v'}-\lambda_v$, i.e., $\Theta_{[v,v']} \mapsto \Theta_{[v,v']} + \Delta_{v, v'}\lambda$ (see \cref{fig:gauge-transformation_image}). 

To transport a charge between sites, the phase has to be adjusted by the relative orientation indicated by the links along the traveled path. In this way, the link operators $U_l$ function as a parallel transport between the adjacent sites, a discrete version of the usual notion in differential geometry. The parallel transport along a path $\gamma$ is given by the Wilson line $W_\gamma$ introduced in \cref{eq:wilson_line}

\begin{figure}[h]
    \centering
    \includegraphics[width=0.6\linewidth]{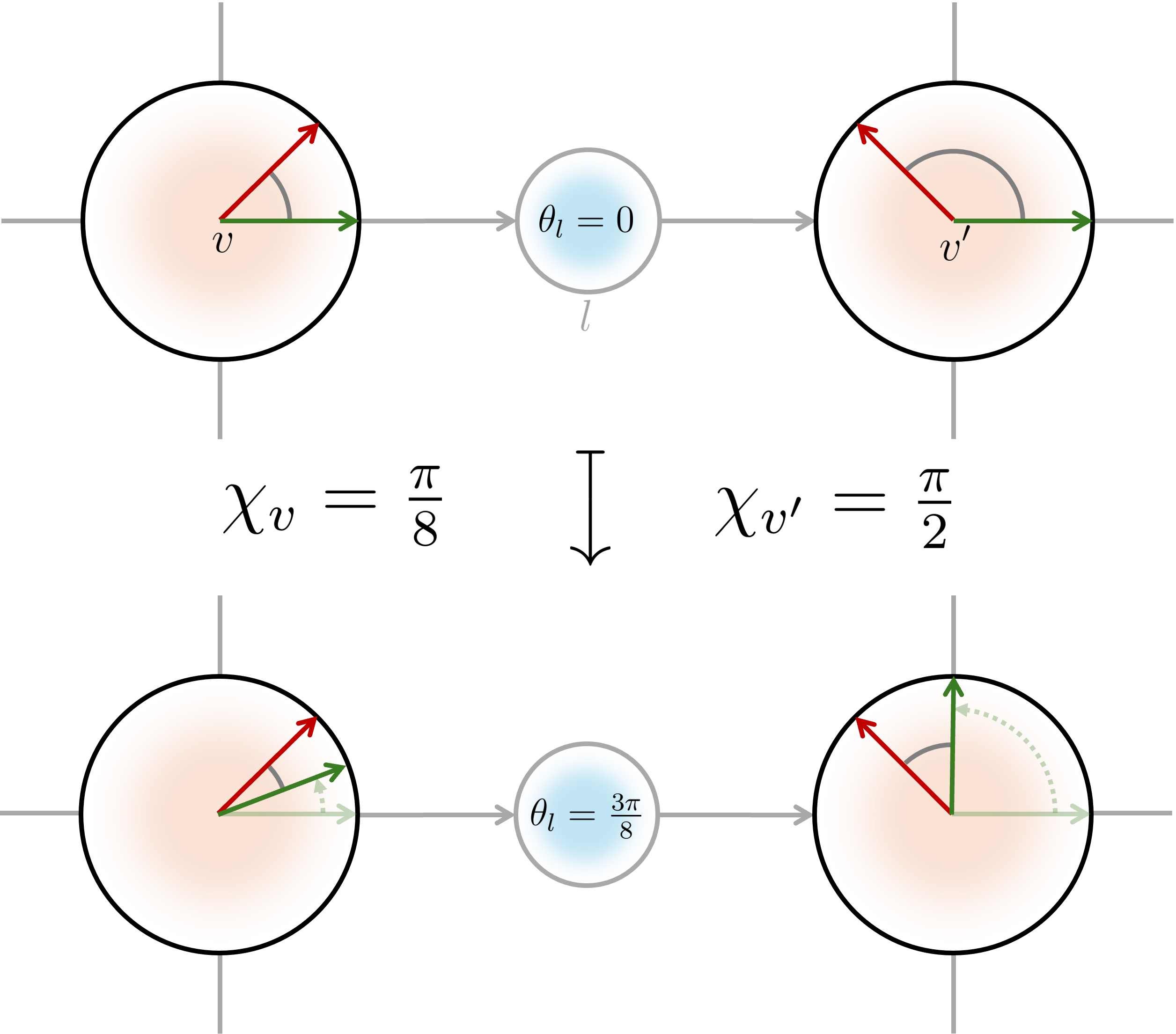}
    \caption{On the sites $v$, $v'$, the matter field operators have a phase degree of freedom (red arrow), which we picture to be indicated with respect to a reference frame (green arrow). The angle $\theta_l$ of the link $l$ measures the relative angle between the reference frames on $v$ and $v'$. Locally rotating the frames on $v, v'$ by $\lambda_v, \lambda_{v'}$ changes the phase of the field by $-\lambda_v, -\lambda_{v'}$ and the angle on $l$ by $\lambda_{v'} - \lambda_v$. }
    \label{fig:gauge-transformation_image}
\end{figure}

\subsection{Staggered Fermions in the Continuum.}
\label{app:staggered_continuum}
When naively discretizing a fermionic field on the lattice, the problem of fermion doubling arises \cite{Susskind_1997}: In the continuum limit, the model describes $2^d$ fermion {\em tastes} (where $d$ is the spatial lattice dimension). Staggered fermions reduce this issue to a degeneracy of $2^{\frac{d-1}{2}}$, and for $d=3$ produce only $2$ tastes. To see how exactly these arise, we will consider the taste basis, following \cite[App. B]{CPS_2025}.

The taste basis describes two Dirac fermions by the $8$ degrees of freedom that live on the corners of a lattice cube. These corners are parametrized by $\{2v + A \, | \,  v\in\mathcal{V}, \ A\in \{0, 1\}^3\}$. Now, for any $w = 2v \in \mathcal{V}$ we define a Dirac matrix field $\Psi_{w}$ by
\begin{equation}
    \Psi_w = \frac{1}{\sqrt{2}}\sum_{A\in\{0, 1\}^3}\alpha^{A}\psi_{w+A},
\end{equation}
where $\alpha^{A}=\alpha_1^{A_1}\alpha_2^{A_2}\alpha_3^{A_3}$. The matrices $\alpha_i= \gamma_0\gamma_i$ are products of the $4\times 4$ gamma matrices $\gamma_0 = \sigma_3\otimes I$, $\gamma_i = i\sigma_2\otimes\sigma_j$, where $\sigma_{1, 2, 3}=X,Y,Z$ denote the Pauli operators. One can show that rewriting the Hamiltonian \labelcref{eq:full_hamiltonian} in terms of $\Psi_v$ recovers the Hamiltonian of two Dirac fields in the continuum limit. 

The two Dirac fermions are contained in the columns of $\Psi_w$. Due to the form of the $\alpha_i$-matrices, $\Psi_w$ is of the block-form
\begin{equation}
    \Psi_w = I\otimes M_w + \sigma_1\otimes N_w,
\end{equation}
and the upper and lower components of the two spinors live in $M_w$ and $N_w$, respectively. Therefore, one can use the first two indices $a = 1, 2$ of $\Psi_{w,ba}$ for the two tastes, while $b$ is the spin-index. 

In order to relate the matrix field operators $\Psi_w, \Psi_w^\dagger$ to the correctable errors in \cref{prop:correctable_matter_errors}, note first that $(\alpha^A)^\dagger = (-1)^{A_1A_2 + A_1A_3 + A_2A_3}\alpha^A$. Therefore, we find that
\begin{equation}
    e^{i\theta_w}\Psi_w + e^{-i\theta_w}\Psi_w^\dagger = \frac{1}{\sqrt{2}}\sum_{A\in\{0, 1\}^3}\alpha^{A}(e^{i\theta_w}\psi_{w+A} +(-1)^{A_1A_2 + A_1A_3 + A_2A_3}e^{-i\theta_w} \psi_{w+A}^\dagger).
\end{equation}
Each entry in the matrix is thus an operator of the form $e^{i\theta_w}\psi_{w+A}+e^{-i\theta_w}\psi_{w+A}$ if $A$ has at least two zero-entries, and $-i\left(e^{i(\theta_w + \pi/2)}\psi_{w+A}+e^{-i(\theta_w + \pi/2)}\psi_{w+A}\right) = e^{i\theta_w}\psi_{w+A}-e^{-i\theta_w}\psi_{w+A}$ else. Therefore, we can find a correctable error set on the staggered fermions as in \cref{prop:correctable_matter_errors} that gives rise to the correctable error set in terms of the taste basis,
\begin{equation}
    \left\{{\bigotimes\nolimits_{w = 2v}}(e^{i\theta_w}\Psi_w + e^{-i\theta_w}\Psi_w^\dagger) \ \middle| \ \forall w: \theta_w \in [0, 2\pi) \right\}.
\end{equation}

\section{Reduced Subspace and Gauge-Invariant Operators of Staggered Fermion Lattice QED}
\label{app:staggered_fermion_lqed_operators}
Consider the fermionic field QRF $\tilde{R}$ from \cref{thm:matter_qrf}. To understand the physical degrees of freedom encoded in $\HS_\phys$, we can consider the image of the reduction maps $\mathcal{R}_{\tilde{R}}^{\bm{\lambda}}$. Since these are isometries, their image is unitarily equivalent to $\HS_\phys$. To prove \cref{prop:electric_basis}, we calculate
\begin{align}
    \Pi_{\tilde{S}|\tilde{R}}^{\bm{\lambda}} &= \mathcal{R}_{\tilde{R}}^{\bm{\lambda}}\circ (\mathcal{R}_{\tilde{R}}^{\bm{\lambda}})^\dagger \\
    &=2^{|\mathcal{V}|}\left(\bra{\bm{\lambda}}_{\tilde{R}}\otimes I_{\tilde{S}}\right)\Pi_\phys \left(\ket{\bm{\lambda}}_{\tilde{R}}\otimes I_{\tilde{S}}\right) \\
    &= 2^{|\mathcal{V}|}\left(\bra{\bm{\lambda}}_{\tilde{R}}\otimes I_{\tilde{S}}\right)\prod_{v\in\mathcal{V}}\left(\int_0^{2\pi}\frac{\D\eta_v}{2\pi}e^{i\eta_v(\mathcal{C}_v^{\mathcal{L}}-\psi_v^\dagger\psi_v +j_v)}\right) \left(\ket{\bm{\lambda}}_{\tilde{R}}\otimes I_{\tilde{S}}\right) \\
    &= 2^{|\mathcal{V}|}\prod_{v\in\mathcal{V}}\bra{\lambda_v}_v\left(\int_0^{2\pi}\frac{\D\eta_v}{2\pi}e^{i\eta_v(\mathcal{C}_v^{\mathcal{L}}-\psi_v^\dagger\psi_v +j_v)}\right)\ket{\lambda_v}_v  \\
    &=  2^{|\mathcal{V}|}\prod_{v\in\mathcal{V}}\left(\int_0^{2\pi}\frac{\D \eta_v}{2\pi}e^{i\eta_v(\mathcal{C}_v^{\mathcal{L}} +j_v)}\braket{\lambda_v}{\lambda_v + \eta_v} \right) \\
    &= 2^{|\mathcal{V}|}\prod_{v\in\mathcal{V}}\left(\int_0^{2\pi}\frac{\D \eta_v}{2\pi}e^{i\eta_v(\mathcal{C}_v^{\mathcal{L}} +j_v)}\frac{1+e^{-i\eta_v}}{2} \right) \\
    &= \prod_{v\in\mathcal{V}}\left(\int_0^{2\pi}\frac{\D \eta_v}{2\pi}\left(e^{i\eta_v(\mathcal{C}_v^{\mathcal{L}} +j_v)}+e^{i\eta_v(\mathcal{C}_v^{\mathcal{L}} - 1+j_v)}\right) \right)\\
    &= 
    \prod_{v\in\mathcal{V}}\left(\Pi_{(-j_v)} + \Pi_{(1-j_v)}\right).
\end{align}
Since this is independent of $\bm{\lambda}$, we denote the projector by $\Pi_{\tilde{S}|\tilde{R}}$.
The calculation behind \cref{eq:reduced_hopping} proceeds similarly:
\begin{align}
    \mathcal{R}_{\tilde{R}}^{\bm{\lambda}}\circ\psi_v^\dagger U_{[v, v']}\psi_{v'} \circ (\mathcal{R}_{\tilde{R}}^{\bm{\lambda}})^\dagger &= 2^{|\mathcal{V}|}\left(\bra{\bm{\lambda}}_{\tilde{R}}\otimes I_{\tilde{S}}\right)\psi_v^\dagger U_{[v, v']}\psi_{v'}\Pi_\phys \left(\ket{\bm{\lambda}}_{\tilde{R}}\otimes I_{\tilde{S}}\right)\\
    &=2^{|\mathcal{V}|}U_{[v, v']}\left(\bra{\bm{\lambda}}_{\tilde{R}}\otimes I_{\tilde{S}}\right)\psi_v^\dagger\psi_{v'}\Pi_\phys \left(\ket{\bm{\lambda}}_{\tilde{R}}\otimes I_{\tilde{S}}\right)\\
    &=2^{|\mathcal{V}|}U_{[v, v']}\prod_{w\in\mathcal{V}}\bra{\lambda_w}_w\psi_v^\dagger\psi_{v'}\left(\int_0^{2\pi}\frac{\D\eta_w}{2\pi}e^{i\eta_w(\mathcal{C}_w^{\mathcal{L}}-\psi_w^\dagger\psi_w +j_w)}\right)\ket{\lambda_w}_w\\
    &= 2^{|\mathcal{V}|}U_{[v, v']}\int_0^{2\pi}\frac{\D\eta_v}{2\pi}e^{i\eta_v(\mathcal{C}_v^{\mathcal{L}}+j_v)}\braket{0}{\lambda_v+\eta_v}\\
    & \qquad \times\int_0^{2\pi}\frac{\D\eta_{v'}}{2\pi}e^{i\eta_{v'}(\mathcal{C}_{v'}^{\mathcal{L}} +j_{v'})}e^{i\lambda_{v'}}\braket{1}{\lambda_{v'}+\eta_{v'}}\\
    &\qquad\times\prod_{w\neq v, v'}\left(\int_0^{2\pi}\frac{\D\eta_w}{2\pi}e^{i\eta_w(\mathcal{C}_w^{\mathcal{L}} +j_w)}\braket{\lambda_w}{\lambda_w+\eta_w}\right) \\
    &= U_{[v, v']}\Pi_{(-j_v)}\Pi_{(1-j_{v'})}\prod_{w\neq v, v'}\left(\Pi_{(-j_w)} + \Pi_{(1-j_w)}\right) \\
    &= U_{[v, v']}\Pi_{(-j_v)}\Pi_{(1-j_{v'})}\Pi_{\tilde{S}|\tilde{R}}.
\end{align}
In the very last step, we used the fact that $\Pi_{(-j_v)}^2 = \Pi_{(-j_v)}$ and $\Pi_{(-j_v)}\Pi_{(1-j_v)} = 0$ (and the same for $v'$) to extend the product from $w\neq v, v'$ to all $w\in\mathcal{V}$. To find \cref{eq:reduced_number}, we can go over all the same steps in a similar way.

\section{Proofs}
\label{app:proofs}

\corrAq*
\begin{proof}
    To see that this operator maps $\HS_\phys \to \HS_{\bm{q}}$, note that
    \begin{equation}
        A_{\bm{q}}\Pi_\phys = \frac{1}{\sqrt{|G|}}\int \D g \overline{\chi}_{\bm{q}}(g)U_{RS}(g)\mathcal{P}_R^eU_{RS}(g)^\dagger\Pi_\phys = \sqrt{|G|}\Pi_{\bm{q}}\mathcal{P}_R^e\Pi_\phys.
        \label{eq:A_q_def}
    \end{equation}
    The operator $A_{\bm{q}}$ is indeed unitary, which we show by direct calculation:
    \begin{align}
        A_{\bm{q}}^\dagger A_{\bm{q}} &= \frac{1}{|{G}|}\int \D g \D h\ \chi_{\bm{q}}(g)\overline{\chi}_{\bm{q}}(h)  \mathcal{P}_R^g\mathcal{P}_R^h \\
        &= \int \D g \D h \ \chi_{\bm{q}}(g)\overline{\chi}_{\bm{q}}(h)\delta(g, h)\left(\ket{\phi(g)}\bra{\phi(g)}_R\otimes I_S\right)\\
        &= \int \D g \ |\chi_{\bm{q}}(g)|^2 \left(\ket{\phi(g)}\bra{\phi(g)}_R\otimes I_S\right)= I_{RS} \,.
    \end{align}
    In the last equality, we used the fact that the orientation states form an orthonormal basis of $\HS_R$.
    For the inverse, we find by an almost identical calculation
    \begin{align}
        A_{\bm{q}}A_{\bm{q}}^\dagger &= \frac{1}{|{G}|} \int\D g\D h \ \overline{\chi}_{\bm{q}}(g) \chi_{\bm{q}}(h) \mathcal{P}_R^g\mathcal{P}_R^h \\
        &= \int \D g \D h \ \overline{\chi}_{\bm{q}}(g){\chi}_{\bm{q}}(h)\delta(g, h)\left(\ket{\phi(g)}\bra{\phi(g)}_R\otimes I_S\right)\\
        &= I_{RS} \,.
    \end{align}
    The operators $\{A_{\bm{q}}\}_{\bm{q}}$ are correctable since they are linear combinations of the correctable gauge-fixing errors.  
\end{proof}

\spanningQRF*
\begin{proof}
    The first property is clear by definition. 
    To construct the set of constraints in property 2, we split $T$ into two subtrees $T_+^l$ and $T_-^l$ by removing the link $l$. We label the subtrees such that $l$ is directed outwards from $T_+^l$ and into $T_-^l$. Then we pick $V_l \defeq \{v\in T_-^l\}$ such that
    \begin{equation}
        \mathcal{C}_{V_l, R} = \sum_{v\in V_l}\left(\sum_{l_{\text{out}} = [v, *] \in R}\epsilon_{l_{\text{out}}}-\sum_{l_{\text{in}} = [*, v] \in R}\epsilon_{l_{\text{in}}}\right) = - \epsilon_l.
    \end{equation}
    The last equality follows from the fact that every link inside $T_-^l$ appears once as $l_{\text{out}}$ and once as $l_{\text{in}}$ in the sum and the contributions thus cancel, leaving only the summand on the link $l\in R$ attached to a vertex in $V_l$ but not in $T_-^l$. 
    
    If we choose the lattice $\Gamma$ such that $|\mathcal{V}|$ is finite, then we can show that these constraints generate all of $\mathcal{G}$ by counting them. Indeed, there are $|\mathcal{V}|-1$ such constraints, as many as there are independent generators of $\mathcal{G}$. 

    If $\Gamma$ is infinite, then we can still show that the constraints $\mathcal{C}_{V_l}$ generate all of $\mathcal{G}$ by constructing an arbitrary local constraint $\mathcal{C}_v$ form them. This can be done by considering all the links attached to the vertex $v\in\mathcal{V}$ that are contained in $T$. We then claim
    \begin{equation}
        \mathcal{C}_{v} = \sum_{l_{\text{out}} = [v, *] \in T} -\mathcal{C}_{V_{l_{\text{out}}}} + \sum_{l_{\text{in}} = [*, v] \in T} \mathcal{C}_{V_{l_{\text{in}}}}.
        \label{eq:claim_equal_generator}
    \end{equation}
    To show this, we will see that the operator on the right hand side can be written as 
    \begin{equation}
        \sum_{v'\in\mathcal{V}\setminus\{v\}}-\mathcal{C}_{v'},
    \end{equation}
    which proves the claim due to \cref{eq:sum_to_zero}.

    Note that if we cut all the links in $T$ attached to $v$, we end up with a certain number of disconnected subtrees, one per link, and the isolated vertex $v$ (see Fig.~\cref{fig:subtrees}). Now we pick a subtree belonging to a link $l_{\text{out}}$ going out from $v$. In this case, the sum of all Gauss' law operators on the vertices of this subtree is, by definition, $\mathcal{C}_{V_{l_{\text{out}}}}$. For an incoming link $l_{\text{in}}$, we could have also constructed $\mathcal{C}_{V_{l_{\text{in}}}}$ from the subtree that goes into $v$ through $l_{\text{in}}$ by writing it as a sum over the Gauss' law operators on vertices in this subtree with a minus sign. That is, if $\tilde{V}_{l_{\text{in}}}$ is this set of vertices, we have
    \begin{equation}
         \sum_{\tilde{v}\in\tilde{V}_{l_\text{in}}}-\mathcal{C}_{\tilde{v}} = \mathcal{C}_{V_{l_\text{in}}}.
    \end{equation}

    Now, the right hand side of \cref{eq:claim_equal_generator} is a sum over the constraints associated with each of the disconnected subtrees, and each of these constraints is a sum of $-\mathcal{C}_{v'}$ for $v'$ on this subtree. By employing \cref{eq:sum_to_zero}, we have proved the claim.

    Property 3 follows form the previous one since $\tilde{G}_R(\bm{\lambda}) = \bigotimes_{l\in R}e^{-i\lambda_l\epsilon_l}$ and thus 
    \begin{equation}
        \ket{\phi(\bm{\lambda})}_{R} = \bigotimes_{l\in R}\ket{e^{i\lambda_l}}_l.
    \end{equation}

    \begin{figure}
        \centering
        \includegraphics[width=0.4\linewidth]{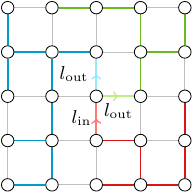}
        \caption{In an infinite lattice, we can split the tree into subtrees by cutting all links attached to a vertex $v$. The operator $-\mathcal{C}_{V_{l_{\text{out}}}}$ is a sum over minus the gauge constraints on all the vertices on the subtree belonging to $l_{\text{out}}$, and $\mathcal{C}_{V_{l_{\text{in}}}}$ is a sum over minus the gauge constraints on the subtree belonging to $l_{\text{in}}$. Adding them all up results in $\mathcal{C}_v$, proving that every gauge transformation can be written in terms of the generators $\mathcal{C}_{V_l}$.}
        \label{fig:subtrees}
    \end{figure}
    
\end{proof}

\holbasis*

\begin{proof}
    Since $(\mathcal{R}_R^{\bm{0}})^\dagger$ is unitary from $\HS_S^{\phys, \bm{0}}$ to $\HS_{\phys}$ and maps the orthonormal basis (in the continuous sense) $\bigotimes_{l\in S}\ket{e^{i\theta_l}}_l$ of $\HS_S$ to the holonomy states $\ket{\{h_l\}}$, these are an orthonormal basis of $\HS_{\phys}$.

    To find their eigenvalues with respect to the fundamental holonomies, note first that $[H_l, \Pi_\phys] = 0$ since they are gauge invariant and that $H_{l, R} = U_l$ while $H_{l, S}$ is a string of $U$-type operators. We can then calculate
    \begin{align}
        H_{l'} \ket{\{h_l\}} &= H_{l'} \cdot (2\pi)^\frac{|\mathcal{V}|-1}{2}\Pi_\mathcal{G} \left(\bigotimes_{l\in R}\ket{e^{i0}}_l\bigotimes_{l\in S}\ket{e^{i\theta_l}}_l \right) \\
        &= (2\pi)^\frac{|\mathcal{V}|-1}{2}\Pi_\mathcal{G} H_{l'}\left(\bigotimes_{l\in R}\ket{e^{i0}}_l\bigotimes_{l\in S}\ket{e^{i\theta_l}}_l \right) \\
        &= (2\pi)^\frac{|\mathcal{V}|-1}{2}\Pi_\mathcal{G} \cdot e^{i\theta_{l'}}\left(\bigotimes_{l\in R}\ket{e^{i0}}_l\bigotimes_{l\in S}\ket{e^{i\theta_l}}_l \right).
    \end{align}
\end{proof}

\matterfieldqrfs*
\begin{proof}
    We will start by considering a single vertex $v\in\mathcal{V}$ and the action of $e^{-i\eta\rho_v}$ on the state $\ket{\lambda}_v = \frac{1}{\sqrt{2}}(\ket{0}_v + e^{-i\lambda}\ket{1}_v)$. We are working with staggered fermions and $\rho_v = \psi^\dagger_v\psi_v - \frac{1}{2}(1-(-1)^{|v|})$. Let us call the parity indicator function $j_v = \frac{1}{2}(1-(-1)^{|v|})$. We then find
    \begin{equation}
        e^{-i\eta\rho_v}\ket{\lambda}_v = e^{i\eta j_v}\frac{1}{\sqrt{2}}(\ket{0}_v + e^{-i(\lambda+\eta)}\ket{1}_v) = e^{i\eta j_v} \ket{\lambda + \eta}_v.
    \end{equation}
    
    Now, moving up to the full orientation states $\ket{\bm{\lambda}}_{\tilde{R}}$ and $G_{\tilde{R}}(\bm{\eta})$, we find that these transform as
    \begin{equation}
        G_{\tilde{R}}(\bm{\eta})\ket{\bm{\lambda}}_{\tilde{R}} = e^{i\sum_{v\in\mathcal{V}}\eta_v j_v}\ket{\bm{\lambda} + \bm{\eta}}_{\tilde{R}}.
    \end{equation}
    They can thus be obtained from the seed state $\ket{\bm{0}}_{\tilde{R}}$, where $\bm{0} = \{0\}_{l\in \tilde{R}}$. This system of orientation states yields a valid QRF because it also resolves identity. On a single site, we have
    \begin{equation}
        \int_0^{2\pi} \D \lambda \ \ket{\lambda}\bra{\lambda}_v  = \int_0^{2\pi} \D \lambda \ \frac{1}{2}(\ket{0}\bra{0}_v+e^{-i\lambda}\ket{1}\bra{0}_v + e^{i\lambda}\ket{0}\bra{1}_v + \ket{1}\bra{1}_v) = \pi I_v.
    \end{equation}
    In total, this results in
    \begin{equation}
        \int_0^{2\pi}\left(\prod_{v\in\mathcal{V}} \D \lambda_v\right) \ket{\bm{\lambda}}\bra{\bm{\lambda}}_{\tilde{R}} = \pi^{|\mathcal{V}|}I_{\tilde{R}}.
        \label{eq:res_of_id_field}
    \end{equation}
\end{proof}

\end{document}